\documentclass[11pt,a4paper]{article}

\usepackage{hyperref}
\hypersetup{
    linktocpage,
    colorlinks,
    citecolor=blue,
    filecolor=blue,
    linkcolor=blue,
    urlcolor=blue
}

\newcommand {\ignore} [1] {}

\usepackage{authblk}

\def\Prob{\mathbb{P}\mathrm{r}}
\def\F{\mathbb{F}}
\def\G{\mathbb{G}}
\def\Exp{\mathbb{E}}

\def\reals{\mathbb{R}}
\def\nats{\mathbb{N}}

\newcommand{\cL}{\mathcal{L}}

\usepackage{graphicx}

\usepackage{fullpage}
\usepackage{graphics}
\usepackage{xcolor}
\usepackage{amsfonts,amsxtra}
\usepackage{amsmath,amssymb}
\usepackage{float}
\usepackage[english]{babel}
\usepackage{etex}
\usepackage{pictex}
\usepackage[linesnumbered,boxed,vlined,ruled]{algorithm2e}
\usepackage{times}
\usepackage{dsfont}

\renewcommand{\Pr}[1]{\Prob\left[#1\right]}

\newcommand{\cH}{{\mathcal{H}}}

\floatstyle{ruled}
\newfloat{Figure}{t}{fig}

\providecommand{\myfloor}[1]{\lfloor {#1} \rfloor}

\newcommand{\polylog}{\mbox{polylog }}

\newenvironment{proof}{\noindent   {\bf Proof.}}{\hspace*{\fill}$\Box$\par\vspace{2mm}}

\newenvironment{proofof}[1]{\noindent {\bf Proof of #1.}}{\hspace*{\fill}$\Box$\par\vspace{2mm}}

\newtheorem{lemma}{Lemma}
\newtheorem{theorem}{Theorem}
\newtheorem{corollary}{Corollary}

\newtheorem{proposition}{Proposition}
\newtheorem{definition}{Definition}

\newcommand{\E}{\mathbb{E}}

\usepackage[textsize=small,textwidth=2cm]{todonotes}
\newcommand{\dk}[1]{\todo[inline,color=red!25!white]{Darek: #1}}

\newcommand{\pk}[1]{{\color{blue}Piotr: #1}}

\begin{document}

\title{
Online Sampling and Decision Making with Low Entropy
}

\author[1]{
Mohammad Taghi Hajiaghayi
}

\author[2]{Dariusz R. Kowalski}

\author[3]{Piotr Krysta}

\author[4]{Jan Olkowski}

\affil[1,4]{Department of Computer Science, University of Maryland, College Park, Maryland, USA. {\tt {hajiagha,olkowski}@umd.edu}.}
\affil[2,3]{School of Computer and Cyber Sciences, Augusta University, Augusta, Georgia, USA. {\tt {pkrysta,dkowalski}@augusta.edu}.}
\affil[3]{Department of Computer Science, University of Liverpool, Liverpool, UK.}

\date{}

\maketitle

\begin{abstract}
Consider the problem: we are given $n$ boxes, labeled $\{1,2,\ldots, n\}$ by an adversary, each containing a single number chosen from an unknown distribution; these $n$ distributions are not necessarily identical. We are also given an integer $k \leq n$. We have to choose an order in which we will sequentially open these boxes, and each time we open the next box in this order, we learn the number in the box. Once we reject a number in a box, the box cannot be recalled. Our goal is to accept $k$ of these numbers, without necessarily opening all boxes, such that the accepted numbers are the $k$ largest numbers in the boxes (thus their sum is maximized). This problem, a.k.a.~{\em free order multiple-choice secretary problem}, is one of the classic examples of online decision making problems.

A natural approach to solve such problems is to use randomness to sample randomly ordered elements, however, as indicated in several sources, e.g., Turan et al. NIST'15~\cite{turan2015random}, Bierhorst et al. Nature'18~\cite{bierhorst2018experimentally}, pure randomness is hard to get in reality. Thus, a pseudo-randomness has to be used, with a small entropy. 
  We show that with a very small
$O(\log\log n)$ entropy an almost-optimal approximation of the value of the set of $k$ largest numbers can be selected, with only a polynomially small additive error, for $k < \log n/\log \log n$. Our solution works for exponentially larger range of parameter $k$ comparing to previously known algorithms (STOC'15 \cite{KesselheimKN15}). 

We present an algorithm for this problem, which is provably and simultaneously near-optimal with respect to the achieved competitive ratio and the used amount of randomness. In particular, we construct a distribution on the orders with entropy $\Theta(\log\log n)$ such that a deterministic multiple-threshold algorithm gives a competitive ratio $1-O(\sqrt{\log k/k})$, for $k < \log n/\log \log n$. Our competitive ratio is simultaneously optimal and uses optimal entropy $\Theta(\log\log n)$, improving in three ways the previous best known algorithm, whose competitive ratio is $1 - O(1/k^{1/3}) - o(1)$. Our algorithm achieves a competitive ratio that is provably optimal for the problem, up to a $\sqrt{\log k}$ factor. Moreover, our solution works for exponentially larger range of parameter $k$. In addition, our algorithm is a simple \textit{deterministic} multiple-threshold algorithm, while the previous algorithms use additional randomness. We also prove a corresponding lower bound on the entropy of optimal and even close-to-optimal (with respect to competitive ratio) solutions for the multiple-choice secretary problem, matching the entropy of our algorithm. No previous lower bound on entropy was known for this problem if $k>1$.

The main takeaway idea is a new framework to derandomize threshold algorithms for secretary-like problems while preserving desirable features of these algorithms. The main building blocks of our framework include decompositions of the random success events of secretary algorithms into symmetric atomic events, derandomization of concentration bounds using special pessimistic estimators, dimensionality reductions, and success probability lifting techniques. We show the strength and wider applicability of our framework by significantly improving the previous results about constructing entropy-optimal distributions for the classic free order secretary problem. Our framework has a potential for applications to threshold algorithms for other problems.

\noindent
{\bf Keywords:}
free order, multiple-choice secretary, optimal stopping theory, online sampling and decision making algorithms
\end{abstract}



\newpage

\tableofcontents

\newpage

\section{Introduction}

Online decision making problems and their important part -- sampling, heavily depend on randomness. Efficient sampling has a long history, since 
Knuth, 
Vitter and others~\cite{knuth1981art, vitter1987efficient}, 
who studied how to generate a randomly ordered sample close to uniformly random. Since then, efficient and accurate sampling has become a fundamental problem in data science, see the recent survey by Mahmud et al.~\cite{mahmud2020survey}. However, as indicated in several prominent sources, such as National Institute of Standards and Technology (NIST)~\cite{turan2015random} or Nature~\cite{bierhorst2018experimentally},  pure randomness is hard to get, thus in practice pseudo-randomness with a small entropy has to be~used.
Recently, the problem has brought more awareness and detailed analysis also in the online 
community~\cite{buchbinder2023lossless, KesselheimKN15}.

Prominent examples of the online decision making problems are secretary and prophet-types of problems.
The most classical version of the secretary problems, introduced by statisticians in the 60s, asks for irrevocably hiring the best secretary among $n$ rankable applicants and was first analyzed in~\cite{Lindley61,dynkin1963optimum,ChowMRS64,GilbertM66, BHZ13, rubinstein2016beyond, RS-SODA17}. In the simplest problem's version, the goal is to find the best strategy when choosing from a sequence of randomly ordered applicants. A multiple-choice refers to the fact that the strategy is allowed to choose up to $k$ applicants.
The secretary problem has also played a fundamental role in advancing stopping theory, online algorithms, and various other fields~\cite{babaioff2007matroids, BabaioffIKK08, HajiaghayiKP04, EHLM17, hajiaghayi2007automated}. 
Only recently, the problem has attracted attention from the learning communities who considered modifications in which the algorithm is given a prediction about the best among applicants~\cite{antoniadis2020secretary}, a prediction interval for evaluation of each applicant~\cite{jiang2021online}, or the objective is to rank all applicants rather than choosing best subset of them~\cite{assadi2019secretary}.

\noindent
{\bf Free-order secretary.} 
We focus on a variant of the secretary problem called {\em the free-order secretary problem}. In this problem, we are given $n$ boxes labeled $\{1,2,\ldots, n\}$ by an adversary, each containing a single number chosen from an unknown distribution, where these  distributions are non-identical. We have to choose an order in which we will sequentially open these boxes. Each time we open the next box, we learn the number in the box. Then we can either stop opening the remaining boxes and accept the current number, in which case the game ends. Or we can continue opening the boxes until the last box, in which case we have to accept this last number. The goal is to accept the largest number in these $n$ boxes.

The free-order assumption reflects the fact that there is a third party that can decide the order of appearing applicants. This is in contrast to the original secretary problem where the order is uniformly random, but also generalizes the model as the third party can always randomly permute the applicants. On the other hand, the hardness of the problem remains in the fact that evaluations of candidates are still unknown to both the third party and an algorithm. For this model, different metrics for creating a fair order cheaper (i.e., an order for which there exist competitive algorithms) have been proposed. Kesselheim, Kleinberg, and Niazadeh~\cite{KesselheimKN15} were the first to ask this question for secretary problems, showing how to construct in polynomial-time, a probability distribution on orders (permutations of size $n$), with entropy $O(\log \log n)$ such that when the classic secretary algorithm ($k=1$) is executed on this distribution, it is successful with probability close to the optimal probability $\frac{1}{e}$. They also prove that if this distribution has entropy $o(\log \log n)$ then no $1$-secretary algorithm can achieve constant success probability. 
When the evaluations of applicants are not deterministic, but rather follow some probabilistic distributions, Arsenis, Drosis, and Kleinberg \cite{ArsenisDK21} prove that a small set of orders can be pre-computed, such that for any collection of $n$ distributions of evaluations, the classic secretary algorithm achieves a constant competitive ratio when sampled from the pre-computed distribution. 

Despite efforts, previous work does not provide exhaustive answers to the following important questions: 
\begin{quote}
{\em What is the minimum entropy of a random order distribution that allows to sample and choose $k$ elements in polynomial time with nearly-optimal competitive ratio?
How large $k$ could be?}
\end{quote}
We take a significant step toward answering both questions. We show that in the free-order model, one can sample from a polynomially constructible distribution of orders with entropy $O(\log\log{n})$, achieving a competitive ratio 
$1-\epsilon(k)$, where $\epsilon(k)$ part is polynomially smaller than that in \cite{KesselheimKN15} and only $\sqrt{\log{k}}$ factor from the absolute bound $\Theta(1/\sqrt{k})$ for any entropy \cite{kleinberg2005multiple}.
We also prove that no smaller entropy could yield good competitiveness, for {\em any} $k$.
Our construction allows for selecting up to $k = O(\log{n} / \log\log{n})$ candidates, a doubly exponential improvement over $k = O\left(\log \log \log 
 n\right)$~of~\cite{KesselheimKN15}. 

The main takeaway idea of our work is a new framework to derandomize threshold algorithms for secretary-like problems by constructing probability distributions on permutations, while preserving desirable features of these algorithms, e.g., the competitive ratio, algorithm’s entropy, or, the number of items to pick. The main building blocks of our framework include decompositions of the random success events of secretary algorithms into symmetric atomic events, derandomization of concentration bounds based on special pessimistic estimators, dimensionality reductions, and success probability lifting techniques. Our framework has a potential for applications to threshold algorithms for other problems.

We further show the strength of our framework by obtaining fine-grained results for optimal distributions of permutations for the free order secretary problem (equivalent to $1$-secretary), improving the results by Kesselheim, Kleinberg and Niazadeh~\cite{KesselheimKN15}. For entropy $\Theta(\log\log n)$, we precisely characterize the success probability of uniform distributions that is below, and close to, $1/e$, and construct such distributions in polynomial time. 

For an overview of our results see Table~\ref{tab:our results} and for their detailed description see Section~\ref{sec:our-results}. \ignore{The full version of all sections and thorough description of all our results can be found in the Appendix.
}

\ignore{
Last but not least, with maximum entropy, $\Theta(n \log(n))$, of the uniform distribution with support $n!$, we find the precise formula $OPT_n$ for the optimal success probability of any secretary algorithm. In addition, we prove that any secretary algorithm that uses any, not necessarily uniform distribution, has success probability at most $OPT_n$.  This improves the result of Samuels from 1981 \cite{Samuels81}.}

\begin{table}
    \centering
    \begin{tabular}{|c|c|c|c|c|c|}
         \hline
         &
         problem & comp.~ratio/succ.~prob. & entropy & 
         range of $k$ & ref. \\
         \hline
         \hline
         upper &
         $k$-secretary & $1 - 4 \sqrt{(\log{k})/k} $ & $O(\log\log{n})$ & 
         $\le \frac{\log{n}}{\log\log{n}}$ & Sec.~\ref{sec:application_k_secretary}\\
         bounds &
         $k$-secretary & $1-O(1/k^{1/3})$  & $O(\log\log{n})$ & 
         $\le (\log^{(3)}{n})^{\delta}$ &\cite{KesselheimKN15}\\
         &
         $1$-secretary & $\frac{1}{e} - \frac{6\log\log{n}}{e\log^{1/2}{n}} $  & $O(\log\log{n})$ & $1$ &Sec. 7.2 \\
         &
         $1$-secretary & $\frac{1}{e}-\omega((1/\log^{(3)}{n})^{c})$  & $O(\log\log{n})$ & $1$ &\cite{KesselheimKN15} \\
         \hline
         lower & 
         $k$-secretary & $\ge 1 - \varepsilon$ & $\ge \frac{1-\varepsilon}{9}\log\log{n}$ & 
         $\le \log^{a} n$ &Sec. 8\\
         bounds & 
         $k$-secretary & $\ge 1 - \varepsilon$ & $\min\{\log 1/\varepsilon,\log \frac{n}{2k}\}^*$ & 
         $<n/2$ 
         &Sec. 8\\
         &
         $1$-secretary & $\ge \frac{\cH(n)}{\log\log n}$ & $\ge \cH(n)^{**}$ & $1$ &\cite{KesselheimKN15} \\
         \hline
         \hline
         notes & \multicolumn{2}{l}{$a, \delta \in (0,1)$ - some constants} & \multicolumn{1}{l}{$c, \varepsilon \in (0,1)$ - any const.} 
         & \multicolumn{2}{l|}{$\log^{(3)}n = \log\log\log{n}$} \\
         \hline
    \end{tabular}
    \caption{Our main results compared to previous results. The upper bounds are with respect to the competitive ratio (success probability, resp.) for the $k$-secretary ($1$-secretary, resp.) problems and appropriate  range of $k$, for which the algorithm runs in polynomial time. Our lower bounds are the first lower bound results for the general $k$-secretary problem, for $k>1$.
    $^*$ is for wait-and-pick algorithms. \\ $^{**}$ $\cH(n)=o(\log\log n) \Rightarrow \mbox{success probability} \leq \frac{\cH(n)}{\log \log n} = o(1)$.
    }
\label{tab:our results}
\end{table}

\subsection{Preliminaries}
\label{sec:prel}

\noindent
{\bf Notation.} Let $[i] = \{1,2,\ldots,i\}$, and
$n$ be the number of arriving elements/items. Each of them has a unique index $i\in [n]$, and corresponding unique value $v(i)$ assigned to it by an adversary. The adversary knows the algorithm and the distribution of random arrival orders.

Let $\Pi_n$ denote the set of all $n !$ permutations of the sequence $(1,2,\ldots,n)$. A {\em probability distribution} $p$ over $\Pi_n$ is a function $p : \Pi_n \longrightarrow [0,1]$ such that $\sum_{\pi \in \Pi_n} p(\pi) = 1$. {\em Shannon entropy}, or simply, {\em entropy}, of the probability distribution $p$ is defined as $\cH(p) = - \sum_{\pi \in \Pi_n} p(\pi) \cdot \log(p(\pi))$, where $\log$ has base $2$\footnote{Consequently, all logarithms in the paper are of base $2$}, and if $p(\pi)=0$ for some $\pi \in \Pi_n$, then we assume that $0 \cdot \log(0) = 0$.
Given a distribution $\mathcal{D}$ on $\Pi_n$, $\pi \sim \mathcal{D}$ means that $\pi$ is sampled from $\mathcal{D}$.
A special case of a distribution, convenient to design efficiently, is when we are given a (multi-)set $\cL\subseteq \Pi_n$ of permutations, called a {\em support}, and random order is selected uniformly at random (u.a.r. for short) from this set; in this case we write $\pi \sim \mathcal{L}$. The entropy of this distribution is $\log |\cL|$. We call such an associated probabilistic distribution {\em uniform}, and otherwise {\em non-uniform}. We often abbreviate ``random variable" to r.v., and ``uniformly at random'' to u.a.r. We will routinely denote as $poly(x)$ a fixed polynomial in single variable $x$, where its form will always be clear from the context.

For a positive integer $k < n$, let $[n]_k$ be the set of all $k$-element subsets of $[n]$.
Given a sequence of (not necessarily sorted) {\em values} $v(1),v(2),\ldots,v(n) \in \reals$, we denote by $ind(k\rq{}) \in \{1,2,\ldots,n\}$ the index of the element with the $k\rq{}$th largest value, that is, the $k\rq{}$th largest value is $v(ind(k\rq{}))$.

\noindent
{\bf Problems.} In the {\em free order multiple-choice secretary} problem, we are given integers $k,n$, $1 \leq k \leq n$, and $n$ boxes labeled $[n] = \{1,2,,\ldots,n\}$ by an adversary, each $i \in [n]$ containing a single number $v(i)$, chosen by the adversary from an unknown distribution. The goal is to choose an order in which we will be sequentially opening these boxes. Each time we open the next box $i$ in the chosen order, we learn the number $v(i)$ and decide to accept $v(i)$ or not. This decision is irrevocable, and we cannot revisit any box. We have to accept $k$ of these numbers without necessarily opening all boxes, where the objective is to accept $k$ largest among them, i.e., to accept $k$ elements with maximum possible sum. Once we have accepted $k$ of the numbers, we stop opening the remaining boxes, if any. If we have accepted $j$ numbers so far and there are only $k-j$ remaining boxes, we have to open all these boxes and accept their numbers. This problem is also called the {\em free order $k$-secretary problem}, and when $k=1$, it is {\em free order secretary problem}.

\noindent
{\bf Competitive ratio.}
As is common, e.g., 
\cite{jaillet2013online}, we quantify the performance of an algorithm $A$ for the free order $k$-secretary problem by the {\em competitive ratio}, saying that $A$ is {\em $\alpha$-competitive} or has {\em competitive ratio} $\alpha \in (0,1)$ if it accepts $k$ numbers whose sum is at least $\alpha$ times the sum of the $k$ largest numbers in the $n$ boxes; the competitive ratio is usually in expectation with respect to randomization in the chosen random order.

\noindent
{\bf Wait-and-pick algorithms.}
An algorithm for the $k$-secretary problem is called {\em wait-and-pick} if it only observes the first $m$ values (position $m \in \{1,2,\ldots,n-1\}$ is a fixed observation {\em checkpoint}), selects one of the observed values $x$ ($x$ is a fixed {\em value threshold} or simply a {\em threshold}), and then selects every value of at least $x$ received after checkpoint position $m$; however, it cannot select more than $k$ values in this way, and it may also select the last $i$ values (even if they are smaller than $x$) provided it selected only $k-i$ values~before~that.

We also consider a sub-class of wait-and-pick algorithms, which as their value threshold $x$ choose the $\tau$-th largest value, for some $\tau \in \{1,2,\ldots,m\}$, among the first $m$ observed values. In this case we say that such wait-and-pick algorithm has a {\em statistic} $\tau$ and value $x$ is also called a statistic in this case.

The definition of the wait-and-pick algorithms applies also to the secretary problem, i.e., with $k=1$. It has been shown that some wait-and-pick algorithms are optimal in case of perfect randomness in selection of random arrival order, see \cite{GilbertM66}.

\noindent
{\bf Threshold algorithms.}
An extension of wait-and-pick algorithms, considered in the literature (c.f., the survey~\cite{GuptaSingla}), allows partition of the order into consecutive phases, and selecting potentially different threshold value for each phase (other than $1$) based on the values observed in the preceding phases. Thresholds are computed at checkpoints -- last positions of phases. These thresholds could also be set based on statistics. We call such algorithms {\em threshold algorithms}, or more specifically {\em multiple-threshold algorithms} if there are more than two phases.

\noindent
{\bf Related work.} The free order $k$-secretary problem is directly related to the {\em prophet inequality problem with order selection}, introduced by Hill \cite{Hill1983}.
Namely, the free order $1$-secretary problem can be modelled by Hill's problem where the numbers in the boxes are given by unknown single-point distributions, whereas the prophet inequality problem with order selection assumes arbitrary (un)known distributions. The free order $k$-secretary problem is also the free order matroid secretary problem, studied by Jaillet, Soto and Zenklusen \cite{jaillet2013online}, where the matroid is uniform. Recently, the {\em free order} (a.k.a., {\em best order}) variants have been studied extensively for the secretary and prophet problems, 
cf.,~\cite{abolhassani2017beating, ArsenisDK21,DBLP:journals/ior/BeyhaghiGLPS21,10.1145/1806689.1806733,DBLP:journals/mp/CorreaSZ21,jaillet2013online,LiuLPSS21,DBLP:conf/sigecom/LiuLPSS21,PT22}. The fundamental~question that this problem asks is: {\em What is the best order to choose to maximize the chance of accepting the largest number?}


For the classic secretary problem, asymptotically optimal algorithm with success probability $\frac{1}{e}$ was proposed, when the order is chosen uniformly at random from the set of all $n!$ permutations. Gilbert and Mosteller~\cite{GilbertM66} showed  with perfect randomness, no algorithm could achieve better probability of success than a simple {\em wait-and-pick} algorithm with checkpoint $m \in \{\lfloor n/e \rfloor,\lceil n/e \rceil\}$.

\section{Our results and techniques}
\label{sec:our-results}


\ignore{
\noindent
{\bf Algorithms.} A free order $k$-secretary algorithm can compute its order non-adaptively, or possibly recompute the order based on the history, i.e., adaptively. We focus in this paper on non-adaptively order algorithms, which we show is without loss of generality for the free order $1$-secretary problem. We leave this as an open question for the
free order $k$-secretary problem.

\begin{proposition}
 
\end{proposition}

\begin{proof}
 We prove in Proposition \ref{Thm:optimum_expansion}, Part 1 and 2, that the best way of computing non-adaptively the order of any (that is also the best) algorithm for the $1$-secretary problem is to use the uniform random order on $\Pi_n$. 
 
 Hill shows ...
\end{proof}
} 


\paragraph{Main contribution: algorithmic results.} 
Our main result is a tight result for optimal policy for the free order multiple-choice secretary problem under low entropy distributions. Below we assume that the adversarial values are such that $v(1) \geq v(2) \geq \cdots \geq v(n)$.

\begin{theorem}\label{thm:k_secretary_main_result}
For any $k < \log{n}/\log \log n$, there exists a multi-set of $n$-element permutations $\mathcal{L}_{n}$ such that a deterministic multiple-threshold algorithm for the free order multiple-choice secretary achieves an expected
$$1 - 4 \sqrt{\frac{\log{k}}{k}} $$
competitive ratio when it uses the order chosen uniformly at random from $\mathcal{L}_{n}$. The set $\mathcal{L}_{n}$ is computable in time $O(\text{poly } (n))$ and the uniform distribution over $\mathcal{L}_{n}$ has the optimal $O(\log\log{n})$ entropy.
\end{theorem}

\noindent
Detailed proof of this theorem can be found in Section~\ref{sec:application_k_secretary}: Theorem \ref{thm:k_secretary_main_result} is reformulated as Theorem~\ref{thm:k_secretary_main}. As a side result, by using the framework developed for the free-order $k$-secretary problem, we obtain the following fine-grained analysis results for the classical secretary problem.

\begin{theorem}\label{thm:1_secretary_results}
There exists a multi-set of $n$-element permutations $\mathcal{L}_{n}$ such that the wait-and-pick algorithm with checkpoint $\lfloor n/e \rfloor$ achieves
$$\frac{1}{e} - \frac{6\log\log{n}}{e\log^{1/2}{n}} $$
success probability for the free order $1$-secretary problem, when the adversarial elements are presented in the order chosen uniformly from $\mathcal{L}_{n}$. The set is computable in time $O(\text{poly } n)$ and the uniform distribution on it has $O(\log\log{n})$ entropy.
\end{theorem}

\noindent Proof of Theorem \ref{thm:1_secretary_results} can be found in Section \ref{sec:application_1_secretary} (as Theorem \ref{thm:1_secretary}).

\medskip

\noindent
{\bf Optimality of our results vs previous results.} 
The problem of finding optimal entropy distributions for secretarial problems has been introduced by \cite{KesselheimKN15}. In this paper, the authors, in the pursuit for the properties of a distribution that guarantee $s$-approximation ratio for $1$-secretary problem ($s$-admissible distributions for short), asked the following natural, yet important, questions.
\begin{quote}
\noindent \textbf{Question(s) ($1$ and $3$ in \cite{KesselheimKN15}):}
\it{"What natural properties of a distribution suffice to guarantee that it is $s$-admissible? What properties suffice to guarantee that it is $s$-optimal? [...] Is there an explicit construction that achieves the minimum entropy?"}
\end{quote}

\noindent While these questions have been already answered in~\cite{KesselheimKN15}, we found those answers unsatisfactory, and in our paper we improve these answers in several places. 
Starting from the most important result, our Theorem \ref{thm:k_secretary_main_result} shows a polynomial construction of a distribution that achieves a (almost)\footnote{The competitive ratio is optimal up to a factor of $\sqrt{\log k}$, see \cite{kleinberg2005multiple,GuptaSingla,AgrawalWY14} for a matching lower bound} optimal competitive ratio of $1-O(\sqrt{\log k/k})$, with minimal entropy $O(\log \log n)$, when $k < \log n/\log \log n$, 
for the $k$-secretary problem in the non-uniform arrival model. The non-uniform arrival model has been introduced in \cite{KesselheimKN15} and is subsumed by the free-order considered in the paper.
This improves over the previous best result of~\cite{KesselheimKN15}, who achieved competitive ratio $(1-O(1/k^{1/3})-o(1))$ using entropy $O(\log \log n)$ and constructing the distribution for $k=O((\log\log\log n)^{\epsilon})$ only. Note, that the construction in this paper improves the range of working values $k$ exponentially.
Optimality of the entropy follows by our new lower bounds in Theorem \ref{thm:lower-general} and \ref{thm:lower}, see the discussion after Theorem \ref{thm:lower} below. Such lower bounds were not known before for the $k$-secretary problem.
As for $1$-secretary problem, our Theorem~\ref{thm:1_secretary_results} improves, over doubly-exponentially, on the additive error to $\frac{1}{e}$ of 
$\omega(\frac{1}{(\log\log\log(n))^{c}})$ proposed in ~\cite{KesselheimKN15,KesselheimKN15-arxiv}, which holds for any positive constant $c < 1$, by achieving a polynomial-time construction with the additive error $\Theta(\frac{\log\log{n}}{\log^{1/2}{n}})$. Interestingly, we achieve our new results by defining and studying some natural properties of distributions on the orders, that are stronger than those in \cite{KesselheimKN15}.

\noindent Finally, we make a step toward answering another important question of \cite{KesselheimKN15}:

\begin{quote}
\noindent \textbf{Question ($4$ in \cite{KesselheimKN15}):}
\it{"Are the performance guarantees of other online algorithms in the uniform-random order model (approximately) preserved when one relaxes the assumption about the input order to the $(k, \delta)$-UIOP or the $(p, q, \delta)$-BIP?[...]"}
\end{quote}

\noindent On the above examples, it becomes clear that the properties of $(k, \delta)$-UIOP and $(p, q, \delta)$-BIP\footnote{See \cite{KesselheimKN15} for precise definitions.}, originally proposed in \cite{KesselheimKN15},  are too weak even to provide the optimal bounds in the cases of $1$-secretary and $k$-secretary problem. Partially it is caused by the fact that different combinatorial structures of different problems are usually hard to optimally capture by a generic set of constraints. To address this problem we propose a simple, but more problem-specific oriented framework for constructing low-entropy distributions that achieves approximation factors of uniform distributions in the non-uniform arrival model for a natural class of threshold algorithms. We expand on this framework below.

\medskip

\noindent
{\bf Technical contributions: algorithms/upper bounds.} We will explain here how we obtain our main result in Theorem \ref{thm:k_secretary_main_result} for the free-order $k$-secretary problem. The same explanations apply to obtain Theorem \ref{thm:1_secretary_results} for the $1$-secretary problem. However, problem specific differences of our framework when applied to the $k$-secretary ($1$-secretary, resp.) are described in Section \ref{sec:application_k_secretary} (\ref{sec:application_1_secretary}, resp.).

Suppose that $\mathcal{P}$ is a secretary-like problem with $n$ items, such as $k$-secretary or  $1$-secretary problem; we will focus here on $\mathcal{P}$ being the $k$-secretary problem. Our goal is to compute a small set $\mathcal{L} \subseteq \Pi_n$, such that there is an algorithm $\mathcal{A}$ for problem $\mathcal{P}$ that chooses $\pi \in \mathcal{L}$ u.a.r., $\pi \sim \mathcal{L}$, as its random order and achieves good expected competitive ratio (success probability, resp.) for the $k$-secretary ($1$-secretary, resp.) problem.
  Let a random
permutation $\pi$, $\pi \sim \Pi_n$, be a random order for the problem $\mathcal{P}$.
Algorithm $\mathcal{A}$ faces adversarial values $v(1) \geq v(2) \geq \cdots \geq v(n)$ with indices $ind(k') \in [n]$ for value $v(k')$, $k' \in [n]$, where $(ind(1),ind(2),\ldots,ind(n))$ is the adversarial permutation. Algorithm $\mathcal{A}$ considers these values sequentially in order $(\pi(ind(1)),\pi(ind(2)),\ldots,\pi(ind(n)))$.

The starting point of our framework is to define a very general family of events, which we call {\em atomic events}, see Definition \ref{def:atomic_event}, in the uniform probability space $(\pi \sim \Pi_n)$. Suppose that we are given a partition of the positions in $\pi$ into $t$ consecutive blocks (buckets) of positions and a mapping $f$ of the $k$ adversarial indices $\{ind(1),\ldots,ind(k)\}$ to the $t$ buckets. Let $\sigma$ be any ordering of the indices $\{ind(1),\ldots,ind(k)\}$. Given $f$ and $\sigma$, we define an atomic event as the event that contains all $\pi \in \Pi_n$ that obey the mapping $f$ and preserve the ordering $\sigma$.

Atomic event has some flavor of ``$k$-wise independent" and ``block independent" random permutation; in fact, our atomic event means that $\pi$ has both block-independence property (BIP) and uniform-induced-ordering property (UIOP) combined, introduced by Kesselheim, Kleinberg, and Niazadeh~\cite{KesselheimKN15}. The crucial property of atomic events 
are that they are very general and very symmetric. Kesselheim, Kleinberg, and Niazadeh~\cite{KesselheimKN15} show how to construct probabilistic distributions on $\Pi_n$ that are UIOP, have small entropy and lead to high competitive ratios for secretary problems. 
  We observe, however, that it 
might be difficult to construct such general distributions obeying atomic events, and thus also UIOP (BIP), with small entropy. The number of all atomic events is $t^k \cdot k!$, where $t$ was defined above. This implies that if we, e.g., would like to preserve a constant fraction of all atomic events, then the resulting entropy would be at least $\Omega(\log (t^k \cdot k!)) = \Omega(k \log k)$. This value of entropy is sufficiently small, i.e., $O(\log\log n)$, only for very small values of $k$, i.e., $k=O(\frac{\log\log n}{\log\log\log n})$.

Atomic events are appealing due to their symmetry, which could lead to simple algorithms and analysis. But, it might be impossible to preserve atomic events directly while keeping the entropy small. Our new approach is to balance these two tensions by introducing
{\em grouping} of atomic events. Atomic events are grouped into {\em positive events}, see Definition \ref{def:positive_event}, that model success probability of threshold algorithms.~%
The next step is a derandomization using Algorithm \ref{algo:Find_perm_2_111} and \ref{algo:Cond_prob_2_111}, that operate on atomic events to handle conditional probabilities, but explicitly preserve only positive events. 
However, the time of this construction is at least $n^k$.~To reduce the running time to polynomial and obtain Theorem~\ref{thm:k_secretary_main_result} (and Theorem
\ref{thm:1_secretary_results}, resp.),
in the final construction described in Section~\ref{sec:application_k_secretary} (Section \ref{sec:application_1_secretary}, resp.), we first use a dimension-reduction of the problem from $n$ to $\text{polylog}(n)$ via a single refined Reed-Solomon code,\footnote{A similar idea appeared in [35], however in our case we managed to significantly simplify and strengthen the idea.} and only after that apply the above  framework.
  We believe this new approach 
enables a general analysis, which should relatively easy generalize to derandomization of threshold algorithms for other problems.

\smallskip

We will now specify the technical details behind each step of our framework\footnote{We would like to note that while in Sections~\ref{section:prob_analysis_k-secr}-\ref{section:dim_reduction_2} we present our techniques on the highest level of generality, below, to carry the necessary intuition, we explain the use of the framework on the example of the multiple-choice secretary problem. The problem-specific applications to the multiple-choice secretary and $1$-secretary problems can be found in Section~\ref{sec:Applications}.}: 

\smallskip

\noindent
{\bf 1. Probabilistic analysis and defining positive events.} (Section \ref{section:prob_analysis_k-secr})\\
\noindent
{\bf 2. Decomposing positive event into atomic events.} (Section \ref{section:prob_analysis_k-secr})\\
\noindent
{\bf 3. Abstract derandomization of positive events via concentration bounds.} (Section \ref{sec:abstract_derand})\\
\noindent
{\bf 4. Dimension reduction and lifting positive events.} (Section \ref{section:dim_reduction_2})

\ignore{
We will choose for $\mathcal{A}$ a threshold algorithm and will first undertake its probabilistic analysis (step {\bf 1.}) assuming that $\mathcal{A}$ uses a random order $\pi \sim \Pi_n$, defining positive events which usually mean that $\mathcal{A}$ chooses some of these $n$ items (one with maximum value in case when $k=1$). The probabilistic analysis will use concentration bounds (such as Chernoff bound). Each positive event will be decomposed into a set of atomic events (step {\bf 2.}). Section \ref{section:prob_analysis_k-secr} defines positive and atomic events abstractly. Unlike positive events, atomic events are highly symmetric and thus help facilitate the derandomization of the concentration bounds (step {\bf 3.} in Section \ref{sec:abstract_derand}), helping in particular to compute the conditional probabilities, see Section \ref{sec:Cond_Prob_Thm_4_2_111}. This derandomization gives an algorithm to compute the set $\mathcal{L}$ in time roughly $O(n^k)$. Then, the dimension reduction from $n$ to $poly \log(n)$ reduces this running time to polynomial in $n$ (step {\bf 4.} in Section \ref{section:dim_reduction_2}). Problem specific descriptions of steps {\bf 1., 2.,} and {\bf 4.} are in Section \ref{sec:application_k_secretary}.
} 

\medskip

\noindent
{\bf 1.~Probabilistic analysis and defining positive events.} The first step is a probabilistic analysis of an algorithm $\mathcal{A}$ for the problem $\mathcal{P}$, assuming that $\mathcal{A}$ uses a random order $\pi \sim \Pi_n$. We need an algorithm whose success probability (expected competitive ratio) can be analyzed by probabilistic events modelled by atomic events. We chose as $\mathcal{A}$, a modified version of multiple-threshold algorithm presented in the survey by Gupta and Singla \cite{GuptaSingla}, see Algorithm \ref{algo:k_secr_algo_1}. We use the probabilistic analysis from \cite{GuptaSingla}, where they apply Chernoff bound to a collection of indicator random variables, which indicate if indices fall in an interval in a random permutation. These random variables are not independent, but they are {\em negatively associated}. We show an additional fact (Lemma \ref{lemma:Chernoff-per}) to justify the application of the Chernoff bound to these negatively-associated random variables. This analysis 
allows us to define positive events for Algorithm \ref{algo:k_secr_algo_1}, which mean that this algorithm picks the item with $i$th largest adversarial value, for $i \in [k]$. 

\noindent
{\bf 2.~Decomposing positive event into atomic events.} Here we show how to define any positive event from the previous step as union of appropriate atomic events. Once we prove that such a decomposition exists, it is easy to find it by complete enumeration over the full space of atomic events, because the size of this space, $t^k \cdot k!$, essentially only depends on $k$. Section \ref{section:prob_analysis_k-secr} defines positive and atomic events abstractly. The decomposition of a positive event for Algorithm \ref{algo:k_secr_algo_1} for the $k$-secretary problem into atomic events can be
found in Section \ref{sec:application_k_secretary}. Section \ref{sec:application_1_secretary} contains an analogous decomposition for the $1$-secretary problem.

\noindent
{\bf 3.~Abstract derandomization of positive events via concentration bounds.} Suppose that each positive event $P_{\gamma}$, defined in terms of atomic events $A \in Atomic(P_{\gamma})$, holds with probability which we denote $p_{\gamma}$. We prove by Chernoff bound (Theorem \ref{theorem:Chernoff_Positive_Events}, Section \ref{sec:abstract_derand}) that there exists a small multi-set $\mathcal{L}$, $|\mathcal{L}| = \ell$, of permutations that nearly preserves probabilities $p_{\gamma}$ of all positive events, when Algorithm \ref{algo:k_secr_algo_1} uses $\pi \sim \mathcal{L}$ as the random order. We show how to algorithmically derandomize this theorem 
(c.f., Theorem \ref{Thm:Derandomization_2_111}, Section \ref{sec:abstract_derand}) by the method of conditional expectations with a special pessimistic estimator for the failure probability, see Section \ref{sec:abstract_derand}. This estimator is derived from the proof of Chernoff bound and inspired by Young's \cite{Young95} oblivious rounding. Our abstract derandomization algorithm in Algorithm \ref{algo:Find_perm_2_111} uses as an oracle an algorithm (from Lemma \ref{lem:mult-computation-time}) that computes a decomposition of any positive event $P$ into atomic events 
$P = \, \, \stackrel{\cdot}{\cup} Atomic(P)$.~%
Algorithm \ref{algo:Find_perm_2_111} calls Algorithm \ref{algo:Cond_prob_2_111} for computing conditional probabilities. This crucially uses a symmetric nature of atomic events (positive events are usually not symmetric).  Algorithm \ref{algo:Cond_prob_2_111} for computing conditional probabilities for atomic events and its analysis are much simpler compared to directly computing conditional probabilities for positive events, and this also gives a promise for further applications.


\noindent
{\bf 4.~Dimension reduction and lifting positive events.} The above abstract derandomization algorithm has time complexity of order $n^k$, which is polynomial only for non-constant $k$. To make it polynomial, we design {\em dimension reductions} from $n$ to $poly \log(n)$. We build on the idea of using Reed-Solomon codes from Kesselheim, Kleinberg and Niazadeh~\cite{KesselheimKN15}, with two changes. First, we only use a single Reed-Solomon code in our dimension reduction, whereas they use a product of 2 or 3 such codes. Second, we replace their second step based on complete enumeration by our derandomization algorithm from  step {\bf 3.} To define the dimension reduction, we propose a new technical ingredient: an algebraic construction of a family of functions that have bounded number of collisions and their preimages are of almost same sizes up to additive $1$ (Lemma \ref{lem:Reed_Solomon_Construction}). We prove this lemma by carefully using algebraic properties of polynomials.
Our construction significantly improves and simplifies the constructions in \cite{KesselheimKN15} by adding the constraint on sizes of preimages and using only one Reed-Solomon code. The constraint on preimages, precisely tailored for the $k$-secretary problem, is crucial for proving the competitive ratios, and allows us to apply more direct techniques of finding permutations distributions over a set with reduced dimension. Our construction is computable in polynomial time and we believe that it is of independent interest.  The last step of our technique is to lift the lower-dimensional permutations back to the original dimension. That is, to prove that we lose only slightly on the probability of positive events when going from the low-dimensional permutations back to the original dimension $n$, see Section \ref{section:dim_reduction_2}.

\vspace*{-2mm}

\paragraph{Algorithms for free-order secretary problems.} Our algorithms can be seen as randomness-efficient algo\-rithms for the $k$-secretary and $1$-secretary problems in the free-order model. \ignore{They are in principle different from the free-order $k$-secretary and free-order prophet inequality algorithms. However,} Our algorithms can be used to solve the free-order $k$-secretary (and $1$-secretary) problem by assuming that their first stage is to choose the random order in which they will open the boxes from the constructed set $\mathcal{L} \subseteq \Pi_n$ u.a.r. Thus, they can be seen as randomized algorithms for the free-order $k$-secretary (and $1$-secretary) problem which use small amount of randomness. We have also mentioned and discussed the prophet inequality problems as only relevant to our problem. The prophet inequality problems rely on values being drawn from known distributions, which is a crucial aspect in prophet inequalities. On the other hand, unknown distributions, which usually are defined as just having $n$ unknown numbers rather than distributions, is crucial for the secretary problems.

\paragraph{Lower bounds.}
We are the first to prove two lower bounds on entropy of $k$-secretary algorithms achieving expected competitive ratio $1-\epsilon$.
These proofs can be found in Section~\ref{sec:lower-bounds}.
The first one is for any algorithm, but works only for $k\le \log^a n$ for some constant $a\in (0,1)$.

\begin{theorem}
\label{thm:lower-general}
Assume $k\le \log^a n$ for some constant $a\in (0,1)$.
Let $\epsilon\in (0,1)$ be a given parameter.
Then, any algorithm (even fully randomized) solving $k$-secretary problem while drawing permutations from some distribution on $\Pi_n$ with an entropy $H\le \frac{1-\epsilon}{9} \log\log n$, cannot achieve the expected competitive ratio of at least $1-\epsilon$ for sufficiently large $n$. 
\end{theorem}

The second lower bound on entropy is for the wait-and-pick algorithms for any $k<n/2$.

\begin{theorem}
\label{thm:lower}
Any wait-and-pick algorithm solving $k$-secretary problem, for $k<n/2$, with expected competitive ratio of at least $(1-\epsilon)$ requires entropy $\Omega(\min\{\log 1/\epsilon,\log \frac{n}{2k}\})$.
\end{theorem}

By Theorem~\ref{thm:lower-general}, entropy $\Omega(\log\log n)$ is necessary for any algorithm to achieve even a constant competitive ratio $1-\epsilon$, for $k=O(\log^a n)$, where $a<1$. It implies that our upper bound in Theorem~\ref{thm:k_secretary_main_result} is tight.
Theorem~\ref{thm:lower} 
implies that entropy $\Omega(\log\log n)$ is necessary for any wait-and-pick algorithm to achieve a close-to-optimal competitive ratio $1-\Omega(\frac{1}{k^a})$, for {\em any} $k<n/2$, where constant $a\le 1/2$.  
Even more, in such case entropy $\Omega(\log k)$ is necessary, which could be $\Omega(\log n)$
for $k = poly(n)$.

\medskip

\noindent
{\bf Technical contributions: lower bounds.} 
The lower bound for all algorithms builds on the concept of semitone sequences with respect to the set of permutation used by the algorithm. It was proposed in \cite{KesselheimKN15} in the context of $1$-secretary problem. Intuitively, in each permutation of the set, the semitone sequence always positions next element before or after the previous elements of the sequence (in some permutations, it could be before, in others -- after).
Such sequences proved useful in cheating $1$-secretary algorithms by assigning different orders of values, but seemed hard to extend to the general $k$-secretary problem.
The reason is that, in the latter,
there are 
two challenges requiring new concepts. First, there are $k$ picks of values by the algorithm, instead of one --
this creates additional dependencies in probabilistic part of the proof (c.f., Lemma~\ref{lem:lower-random-adv}), which we overcome by introducing more complex parametrization of events and inductive proof.
Second, the algorithm does not always have to choose maximum value to guarantee competitive ratio $1-\epsilon$, or can still choose the maximum value despite of the order of values assigned to the semitone sequence -- to address these challenges, we not only consider different orders the values in the proof (as was done in case of $1$-secretary in~\cite{KesselheimKN15}), but also expand them in a way the algorithm has to pick the largest value but it cannot pick it without considering the order (which is hard for the algorithm working on semitone sequences). It leads to so called hard assignments of values and their specific distribution in Lemma~\ref{lem:lower-random-adv} resembling biased binary search, see details in Section~\ref{sec:lower-general}.

The lower bound for wait-and-pick algorithms, presented in Section~\ref{sec:lower-wait-and-pick}, is obtained by constructing a virtual bipartite graph with neighborhoods defined based on elements occurring on left-had sides of the permutation 
checkpoint, and later by analyzing relations between sets of elements on one side of the graph and sets of permutations represented by nodes on the other side of the graph.


\subsection{Complementary results}
\noindent
We prove in Proposition \ref{Thm:optimum_expansion} (Section \ref{section:lb_classic_secr}), a characterization of the optimal success probability $OPT_n$ of secretary algorithms. When the entropy is maximum, $\Theta(n \log(n))$, of the uniform distribution on the set of $n!$ permutations, we find the precise formula for the optimal success probability of the best secretary algorithm, $OPT_n = 1/e + c_0/n + \Theta((1/n)^{3/2})$, $c_0 = 1/2 - 1/(2e)$ (Part 1, Proposition \ref{Thm:optimum_expansion}). We prove that any secretary algorithm that uses any, not necessarily uniform distribution, has success probability at most $OPT_n$ (Part 2, Proposition \ref{Thm:optimum_expansion}). This improves the result of Samuels \cite{Samuels81}, who proved that under uniform distribution no secretary algorithm can achieve success probability of $1/e + \varepsilon$, for any constant $\varepsilon > 0$. Interestingly, no uniform probability distribution with small support and entropy $< \log (n)$ can have success probability above $1/e$ (Part 3, Proposition~\ref{Thm:optimum_expansion}). 

\vspace*{1ex}
\noindent

\section{Further related work}
In this section, we present recent related literature on important online stopping theory concepts such as secretary, prophet inequality, and  prophet secretary. 
\vspace{-0.11in}
\paragraph{Secretary Problem.}
In this problem, we receive a sequence of randomly permuted numbers
in an online fashion. Every time we observe a new number, we have
the option to stop the sequence and select the most recent number.
The goal is to maximize the probability of selecting the maximum of
all numbers. The pioneering work of Lindley~\cite{Lindley61} and Dynkin~\cite{dynkin1963optimum}
present a simple but elegant algorithm that succeeds with
probability $1/e$. In particular, they show that the best strategy, a.k.a.~wait-and-pick, is
to skip the first $1/e$ fraction of the numbers and then take the
first number that exceeds all its predecessors. Although simple,
this algorithm specifies the essence of best strategies for many
generalizations of secretary problem. Interestingly, Gilbert and Mosteller~\cite{GilbertM66} show that when the values are drawn i.i.d.~from a known
distribution, there is a wait-and-pick algorithm that selects the best value with probability approximately 0.5801 (see~\cite{DBLP:conf/aistats/EsfandiariHLM20} for generalization to non-identical distributions). 

The connection between secretary problem and online auction
mechanisms has been explored by the 
work of Hajiaghayi, Kleinberg and
Parkes~\cite{HajiaghayiKP04}. In particular, they introduce the
{\em multiple-choice value version} of the problem, the $k$-secretary problem, in which the goal
is to maximize the expected sum of the selected numbers, with applications to limited-supply online auctions.
Kleinberg~\cite{kleinberg2005multiple} later presents a tight
$(1-O(\sqrt{1/k}))$-competitive algorithm for $k$-secretary resolving an open problem of~\cite{HajiaghayiKP04}. The
bipartite matching variant is studied by Kesselheim et
al.~\cite{kesselheim2013optimal} for which they give a
$1/e$-competitive solution. 

Babaioff et al.~\cite{babaioff2007matroids} consider the
{\em matroid} version and give an $\Omega(1/\log k)$-competitive
algorithm when the set of selected items have to be an independent
set of a rank $k$ matroid. After the introduction of the problem, many papers have improved the competitive ratio which led eventually to $\Omega(1/\log\log{k})$-competitive algorithms obtained independently in \cite{DBLP:conf/soda/FeldmanSZ15, DBLP:conf/focs/Lachish14}.
Other generalizations of secretary problem such as the submodular variant has been initially studied by the Bateni, Hajiaghayi, and ZadiMoghaddam~\cite{BHZ13} and  Gupta, Roth, Schoenebeck, and
Talwar~\cite{DBLP:conf/wine/GuptaRST10}.

\vspace{-0.11in}
\paragraph{Prophet Inequality.} In prophet inequality, we are initially given $n$ distributions for each of the numbers in the sequence.
Then, similar to the secretary problem setting, we observe the
numbers one by one, and can stop the sequence at any point and
select
 the most recent observation. The goal is to maximize the ratio between the expected value of the selected number
 and the expected value of the maximum of the sequence.
  This problem was first introduced by Krengel-Sucheston~\cite{krengel1977semiamarts,krengel1978semiamarts},
  for which they gave a tight $1/2$-competitive algorithm. Later on,
the research investigating the relation between prophet inequalities
and online auctions was initiated by Hajiaghayi, Kleinberg, and Sandholm~\cite{hajiaghayi2007automated}. 
In
particular this work considers the multiple-choice variant
of the problem in which a selection of $k$ numbers is allowed and
the goal is to maximize the ratio between the sum of the selected
numbers and the sum of the $k$ maximum numbers. The best result on
this topic is due to Alaei~\cite{alaei2014bayesian} who gives a
$(1-{1}/{\sqrt{k+3}})$-competitive algorithm. This factor almost
matches the lower bound of $1-\Omega(\sqrt{1/k})$ already known from
the prior work of Hajiaghayi et al.~\cite{hajiaghayi2007automated}. Motivated
by applications in online ad-allocation, Alaei, Hajiaghayi and Liaghat~\cite{AHL13} study the bipartite matching variant
of prophet inequality and achieve the tight factor of $1/2$.
Feldman et al.~\cite{feldman2015combinatorial} study the
generalizations of the problem to combinatorial auctions in which
there are multiple buyers and items and every buyer, upon her
arrival, can select a bundle of available items. Using a posted
pricing scheme they achieve the same tight bound of $1/2$.
Furthermore, Kleinberg and Weinberg~\cite{KW-STOC12} study the problem
when a selection of multiple items is allowed under a given set of
matroid feasibility constraints and present a $1/2$-competitive
algorithm. Yan \cite{yan2011mechanism} improves this bound to
$1-1/e\approx 0.63$ when the arrival order can be determined by the
algorithm. More recently Liu, Paes Leme, P{\'{a}}l, Schneider, and
Sivan~\cite{DBLP:conf/sigecom/LiuLPSS21} obtain the first Efficient PTAS (i.e., a $1+\epsilon$ approximation for any constant $\epsilon> 0$) for the free order (best order) case when the arrival order can be determined by the algorithm (the task of selecting
the optimal order is NP-hard~\cite{DBLP:conf/sigecom/0001SZ20}).
In terms of competitive ratio for the free order  case, very recently, Peng and Tang~\cite{PT22} (FOCS'22),
obtain a 0.725-competitive algorithm, that substantially improves the state-of-the-art 0.669
ratio by Correa, Saona and Ziliotto~\cite{DBLP:journals/mp/CorreaSZ21}.

 Prophet inequality (and the secretary problem) has also been studied beyond a matroid or a matching. For the intersection of $p$ matroids, Kleinberg and Weinberg~\cite{KW-STOC12} gave an $O(1/p)$-competitive prophet inequality. Later, D\"{u}tting and Kleinberg~\cite{dutting2015polymatroid} extended this result to polymatroids. Rubinstein~\cite{rubinstein2016beyond} and Rubinstein and Singla~\cite{RS-SODA17} consider prophet inequalities and secretary problem for arbitrary downward-closed set systems. Babaioff et al.~\cite{babaioff2007matroids} show a lower bound of $\Omega(\log n \log\log n)$ for this problem. 
Prophet inequalities have also been studied for many  combinatorial
optimization problems, e.g.,
\cite{DEHLS17,garg2008stochastic,gobel2014online,Meyerson-FOCS01}.
\paragraph{Prophet Secretary.}~The prophet inequality setting assumes either the buyer values or the buyer arrival order is chosen by an adversary. In practice, however, it is
often conceivable that there is no adversary acting against you. Can
we design better strategies in such settings? The  {\em prophet
secretary} model introduced by the Esfandiari, Hajiaghayi, Liaghat, and Monemizadeh~\cite{EHLM17}  is a natural way to consider
such a process assuming both {\em stochastic knowledge} about
buyer values and uniformly random buyers' arrival order.
The goal is to design a strategy that maximizes expected accepted
value, where the expectation is over the random arrival order, the
stochastic buyer values, and internal randomness of the
strategy.
 
This work introduced a natural combination of the fundamental
problems of prophet and secretary. More formally, in the \textit{prophet secretary} problem we are given $n$ distributions
$\mathcal{D}_1,\ldots,\mathcal{D}_n$ from which $X_1,\ldots,X_n$ are drawn. After applying a random permutation $\pi(1),\ldots,\pi(n)$ the values of the items are presented in an online fashion:
in step $i$ both $\pi(i)$ and $X_{\pi(i)}$ are revealed. The goal is to stop the sequence to maximize the expected
value\footnote{Over all random permutations and draws from distributions} of the most recent item. Esfandiari, Hajiaghayi, Liaghat, and Monemizadeh~\cite{EHLM17}
provide an algorithm that achieves a competitive factor of $1-1/e$ when $n$ tends to infinity. Beating the factor of 
$1-\frac{1}{e}\approx 0.63$ substantially for the prophet secretary problems, however, has been challenging. A recent result by Azar et al.~\cite{ACK18} and Correa et
al.~\cite{CorreaSZ19} improves this bound to
$1-\frac{1}{e}+\frac{1}{30}\approx 0.665$. For the special case of {\em single item
i.i.d.}, Hill and Kertz~\cite{hill1982comparisons}~give a
characterization of the hardest distribution, and Abolhasani et
al.~\cite{abolhassani2017beating} show that one can get a
$0.73$-competitive ratio. Recently, this factor has been improved to the
tight bound of $0.745$ by Correa et al.~\cite{correa2017posted}. However finding
the tight bound for the general prophet secretary problem still remains the main~open~problem.

\section{Probabilistic atomic and positive events for threshold algorithms}\label{section:prob_analysis_k-secr}

Let $\Omega = (\Pi_n,\mu)$ denote the probabilistic space of all $n!$ permutations of $n$ elements with uniform probabilities, i.e., $\Prob[\pi \in \Pi_n] = \mu(\pi) = 1/n!$. In the next definition we will define atomic events in space $\Omega$. %
\ignore{
\begin{definition}[Atomic events]
Given any integer $t \in \{2,3,\ldots, n\}$, let $1 = \tau_0 < \tau_{1} < \tau_{2} < \cdots < \tau_{t-1} < \tau_t = n$, $\forall j \in [t-1] : \tau_j \in [n]$, be a set of fixed thresholds, $\mathcal{T} = \{\tau_{1}, \tau_{2}, \cdots,\tau_{t-1}\}$, $\forall j \in [t-1] : \tau_j \in [n]$. Let $K \subseteq [n]$ be any subset of $|K| = k \in [n]$ indices. Let also $\psi : K \longrightarrow [t]$ be any mapping of indices from set $K$ to $t$ "time intervals" $\{\tau_0,\ldots,\tau_1\}$, $\{\tau_1+1,\dots,\tau_2\}$, $\{\tau_2+1,\dots,\tau_3\}$, $\cdots$, $\{\tau_t+1,\dots,\tau_t\}$, such that
$$
  \forall j \in [t] : |\psi^{-1}(\{j\})| \leq \tau_j - \tau_{j-1} + 1_{j > 1},
$$ where $1_{j > 1} = 1$ if $j>1$ and $1_{j > 1} = 0$ otherwise.

The following event, called an atomic event, is important for the threshold algorithms:
$$
 A_{\mathcal{T},K,\phi} = \{ \pi \in \Omega \,\, | \,\, \forall i \in K : \tau_{\psi(j)-1} < \pi^{-1}(i) \leq \tau_{\psi(i)}\} \, .
$$

We consider a family of all atomic events for a fixed $n$ and for all possible thresholds, sets $K$ and mappings $\psi : K \longrightarrow [t]$.
\end{definition}
} %
Given any integer $t \in \{2,3,\ldots, n\}$, let $\mathcal{B} := B_{1}, B_{2}, \ldots, B_{t}$, be a \textit{bucketing} of the sequence $(1, \ldots, n)$, i.e., partition of the sequence $(1,\ldots, n)$ into $t$ disjoint subsets (buckets) of consecutive numbers whose union is the whole sequence. Formally, there are indices $\tau_{1} < \tau_{2} < \cdots < \tau_{t-1} < \tau_t = n$, $\forall j \in [t-1] : \tau_j \in [n]$, such that $B_1 = \{1,\ldots,\tau_{1}\}$, and $B_j = \{\tau_{j-1}+1,\ldots,\tau_{j}\}$, for $j \in \{2,3,\ldots, t\}$.

\ignore{
\begin{definition}[Atomic events with respect to a bucketing]
Consider any k-tuple $\hat{S} = (a_{1}, \ldots, a_{k})$ of the set $[n]$. Let $f : [k] \rightarrow [t]$ be a non-decreasing mapping of elements from the sequence into $t$ buckets of the bucketing $\mathcal{B}$. Then an atomic event in probability space $\Omega$ for chosen $K$ and $f$ is defined as:
$$A_{K, f} = \{\pi \in \Omega : \forall_{i \in [k]} \pi^{-1}(a_{i}) \in B_{f(i)} \text{ and } \pi^{-1}(a_{1}) < \pi^{-1}(a_{2}) < \ldots < \pi^{-1}(a_{k}) \}.$$
The family $\mathcal{A}_{k, \mathcal{B}}$, parameterized by the number $k$ and the bucketing $\mathcal{B}$, of all atomic events is defined as:
$$\mathcal{A}_{k, \mathcal{B}} = \bigcup_{K, f} A_{K, f}.$$
\end{definition}
} 

\begin{definition}[Atomic events with respect to a bucketing]\label{def:atomic_event}
Consider any k-tuple $\sigma = (\sigma_{1}, \ldots, \sigma_{k}) = (\sigma(1), \ldots, \sigma(k))$ of the set $[n]$, i.e., $\sigma : [k] \longrightarrow [n]$ and $\sigma$ is injective. Let $f : [k] \rightarrow [t]$ be a non-decreasing mapping of elements from the sequence into $t$ buckets of the bucketing $\mathcal{B}$. Then an atomic event in probability space $\Omega$ for the chosen $\sigma$ and $f$ is defined as:
$$A_{\sigma, f} = \{\pi \in \Omega : \forall_{i \in [k]} \pi^{-1}(\sigma_{i}) \in B_{f(i)} \text{ and } \pi^{-1}(\sigma_{1}) < \pi^{-1}(\sigma_{2}) < \ldots < \pi^{-1}(\sigma_{k}) \}.$$
The family $\mathcal{A}_{k, \mathcal{B}}$, parameterized by the number $k$ and the bucketing $\mathcal{B}$, of all atomic events is defined as:
$$\mathcal{A}_{k, \mathcal{B}} = \bigcup_{\sigma, f} A_{\sigma, f}.$$
\end{definition}

While the above definition captures with the most details, the structure of 
multiple-threshold algorithms, it is often impractical for derandomization purposes. The measure of an atomic event $A_{\sigma, f}$ is proportional to the inverse of $k! \cdot t^{k}$. Even for small parameters $k$ and $t$, this measure can be too small to allow constructing low-entropy permutation distributions reflecting measures of atomic events. We will show that it is often the case that threshold algorithms are interested in preserving measures of some super-sets of disjoint atomic events whose measure is some constant depending on the competitive ratio of these algorithms. Motivated by this observation, for a family of atomic events $\mathcal{A}_{k, \mathcal{B}}$, we define an abstract notion of a family of positive events $\mathcal{P} \subseteq 2^{\mathcal{A}_{k, \mathcal{B}}}$ subject to some structural properties.

\begin{definition}[Positive events based on atomic family $\mathcal{A}_{k, \mathcal{B}}$]\label{def:positive_event}
A positive event $P$ is any subset of the atomic family $\mathcal{A}_{k, \mathcal{B}}$, denoted $Atomic(P) \subseteq \mathcal{A}_{k, \mathcal{B}}$, such that every two atomic events belonging to $P$ are disjoint, $P = \, \, \stackrel{\cdot}{\bigcup}_{A \in Atomic(P)} A$. Any set of positive events based on the atomic family $\mathcal{A}_{k, \mathcal{B}}$ is called a positive family of events. 
\end{definition}

We will propose in Section \ref{sec:Applications} two different positive families that capture behaviors of two optimal algorithms for respectively the multiple-choice secretary problem and the classic secretary problem. We will also show how these positive events can be expressed by atomic events.

\section{Abstract derandomization of positive events via concentration bounds: Theorem \ref{Thm:Derandomization_2_111}}\label{sec:abstract_derand}

Let $\Omega = (\Pi_n,\mu)$ denote the probabilistic space of all $n!$ permutations of $n$ elements with uniform probabilities. Given any integer $t \in \{2,3,\ldots, n\}$, let $\mathcal{B} := B_{1}, B_{2}, \ldots, B_{t}$, be a \textit{bucketing} of the sequence $(1, \ldots, n)$. We will derandomize, in Theorem \ref{Thm:Derandomization_2_111}, the following generic theorem (Theorem \ref{theorem:Chernoff_Positive_Events}), where we only need that any positive event can 
be expressed as union of (any set of) atomic events. 

\begin{theorem}\label{theorem:Chernoff_Positive_Events}
  Let $k \in [n], k > 2$ and $\mathcal{A}_{k,\mathcal{B}}$ be the
family of atomic events in the space $\Omega$. Let $\mathcal{P} =\{P_1,\ldots,P_q\}$ be a family of positive events based on family $\mathcal{A}_{k,\mathcal{B}}$, for some integer $q > 1$, such that for any $P_{\gamma} \in  \mathcal{P}$, $\gamma \in [q]$, we have $\Prob_{\pi \sim \Pi_n}[P_{\gamma}] \geq p_{\gamma} > 0$ for some $p_{\gamma} \in (0,1)$. Let $p_0 = \min \{p_1,\ldots,p_q\}$. Then, for any $\delta \in (0,1)$, there exists a multi-set $\mathcal{L}$ of permutations of size at most $\ell = \frac{2\log{q}}{\delta^2 p_0}$ such that 
\[
\Prob_{\pi \sim \mathcal{L}}[P_{\gamma}] \ge (1-\delta) \cdot p_{\gamma} \, , 
 \mbox{ for each } P_{\gamma} \in \mathcal{P} \, .
\] 
\end{theorem}

\begin{proof}
Let us fix any $P_{\gamma} \in \mathcal{P}$. We choose independently $\ell$ permutations $\pi_1, \ldots, \pi_{\ell}$ from $\Pi_n$ u.a.r., and define the multi-set $\mathcal{L} = \{\pi_1, \ldots, \pi_{\ell}\}$.
Let $X_1(P_{\gamma}), \ldots, X_{\ell}(P_{\gamma})$ be random variables such that $X_s(P_{\gamma}) = 1$ if event $P_{\gamma}$ holds for
random permutation $\pi_s$, and $X_s(P_{\gamma}) = 0$ otherwise, for $s \in [\ell]$.
 Then for $X(P_{\gamma}) =
X_1(P_{\gamma}) + \cdots + X_{\ell}(P_{\gamma})$ we have that $\Exp[X(P_{\gamma})] \geq p_{\gamma} \ell$ and by Chernoff bound, we have  
\begin{eqnarray}
 \Prob[X(P_{\gamma}) < (1-\delta) \cdot p_{\gamma}  \ell]  <  \exp(-\delta^2 p_{\gamma} \ell/2) \, , \,\, \mbox{ for any } \, \, 0 < \delta < 1 \, , \label{eqn:Chernoff_Hoeffding_111}
\end{eqnarray}

The probability that there is an event $P_{\gamma} \in \mathcal{P}$ for which there does not exists a $(1-\delta) p_{\gamma}$ fraction of permutations among these $\ell$ random permutations which make this event false, by the union bound, is: 
\vspace*{-1ex}
\[
\Prob[\exists P_{\gamma} \in \mathcal{P} : X(P_{\gamma}) <
(1-\delta) \cdot p_{\gamma}  \ell] \,\, < \,\, 
\sum_{i=1}^q \exp(-\delta^2 p_{\gamma} \ell/2) \, .
\] This probability is strictly smaller than $1$ if $\sum_{{\gamma}=1}^q \exp(-\delta^2 p_{\gamma} \ell/2) \leq 1$, which holds if $\ell \geq \frac{2\log{q}}{\delta^2 p_0}$. Therefore, each positive event $P_{\gamma}$ has at least a $(1-\delta) \cdot p_{\gamma}$ fraction of permutations in $\mathcal{L}$ on which it is true. This means that there exist $\frac{2\log{q}}{\delta^2 p_0}$ permutations such that if we choose one of them u.a.r., then for any positive event $P_{\gamma} \in \mathcal{P}$, this permutation will make $P_{\gamma}$ true with probability at least $(1-\delta) p_{\gamma}$.
\end{proof}

\begin{theorem}\label{Thm:Derandomization_2_111}
  Let $k \in [n], k > 2$ and $\mathcal{A}_{k,\mathcal{B}}$ be the family of atomic events in the space $\Omega$, where bucketing $\mathcal{B}$ has $t$ buckets. Let $\mathcal{P} =\{P_1,\ldots,P_q\}$ be a family of positive events based on the atomic family $\mathcal{A}_{k,\mathcal{B}}$, for some integer $q > 1$, such that for any $P_{\gamma} \in  \mathcal{P}$, $\gamma \in \{1,\ldots,q\}$, we have $\Prob_{\pi \sim \Pi_n}[P_{\gamma}] \geq p_{\gamma} > 0$ for some $p_{\gamma} \in (0,1)$. Let $p_0 = \min \{p_1,p_2\ldots,p_q\}$, and for any $P_{\gamma} \in \mathcal{P}$, $Atomic(P_{\gamma})$ be the set of atomic events that define $P_{\gamma}$. Then, for any $\delta \in (0,1)$, a multi-set $\mathcal{L}$ of permutations with $|\mathcal{L}| \leq \ell = \frac{2\log{q}}{\delta^2 p_0}$, such that 
\[
\Prob_{\pi \sim \mathcal{L}}[P_{\gamma}] \ge (1-\delta) \cdot p_{\gamma} \, , 
 \mbox{ for each } P_{\gamma} \in \mathcal{P} \, ,
\] can be constructed by using Algorithm \ref{algo:Find_perm_2_111} in deterministic time
$$
O\left( \ell n^3 q \cdot t^{2k} \cdot (k!)^2 \cdot \left(n + k k! + k \log^2 n\right) \right) \, . 
$$
\end{theorem}

\vspace*{-1ex}
\noindent
We present here the proof of Theorem \ref{Thm:Derandomization_2_111}, whose missing details can be found in Section \ref{sec:derandomization-proofs_2_111}.

\smallskip

\noindent
{\bf Preliminaries.} To derandomize the Chernoff argument of Theorem \ref{theorem:Chernoff_Positive_Events}, we will derive a special conditional expectations method with a pessimistic estimator. We will model an experiment to choose u.a.r.~a permutation $\pi_j \in \Pi_n$ by independent \lq\lq{}index\rq\rq{} r.v.'s $X^i_j$: $\Prob[X^i_j \in \{1,2,\ldots, n-i+1\}] = 1/(n-i+1)$, for $i \in [n]$, to define $\pi = \pi_j \in \Pi_n$ ``sequentially": $\pi(1) = X^1_j$, $\pi(2)$ is the $X^2_j$-th element in $I_1 = \{1,2,\ldots,n\} \setminus \{\pi(1)\}$, $\pi(3)$ is the $X^3_j$-th element in $I_2 = \{1,2,\ldots,n\} \setminus \{\pi(1), \pi(2)\}$, etc, where elements are increasingly ordered.

Since the probability of choosing the index $\pi(i)$ for $i = 1,2,\ldots, n$ is $1/(n-i+1)$, and these random choices are independent, the final probability of choosing a specific random permutation is 
  $
     \frac{1}{n} \cdot \frac{1}{n-1} \cdot \ldots \cdot \frac{1}{n-n+1} = \frac{1}{n!} \ ,
  $ thus, this probability distribution is uniform on the set $\Pi_n$ as we wanted.
  
 Suppose random permutations
 $\mathcal{L} = \{\pi_1, \ldots, \pi_\ell\}$ are generated using $X^1_j, X^2_j, \ldots, X^n_j$ for $j\in[\ell]$.  Given a positive event $P_{\gamma} \in \mathcal{P}$, ${\gamma} \in [q]$, recall the definition of r.v.~$X_j(P_{\gamma})$ for $j \in [\ell]$ given above.
  For $X(P_{\gamma}) =
X_1(P_{\gamma}) + \cdots + X_{\ell}(P_{\gamma})$ and $\delta \in (0,1)$, we have that $\Exp[X(P_{\gamma})] \geq  p_{\gamma} \ell$ and by (\ref{eqn:Chernoff_Hoeffding_111}) we have 
$
 \Prob[X(P_{\gamma}) < (1-\delta) \cdot p_{\gamma}  \ell]  <  \exp(-\delta^2 p_{\gamma} \ell/2)
$, and so $$\Prob[\exists P_{\gamma} \in \mathcal{P} : X(P_{\gamma}) <
(1-\delta) \cdot p_{\gamma}  \ell] \, < \, 1 \,\, \mbox{ for } \, \ell \geq \frac{2\log{q}}{\delta^2 p_0} \, .
$$ We call the positive event $P_{\gamma} \in \mathcal{P}$ {\em not well-covered} if $X(P_{\gamma}) < (1-\delta) \cdot p_{\gamma} \ell$ (then a new r.v.~$Y(P_{\gamma}) = 1$), and {\em well-covered} otherwise (then $Y(P_{\gamma}) = 0$). Let $Y = \sum_{P \in \mathcal{P}} Y(P)$. By the above argument $\Exp[Y] = \sum_{P \in \mathcal{P}} \Exp[Y(P)] < 1$ if $\ell \geq \frac{2\log{q}}{\delta^2 p_0}$. We will keep the expectation $\Exp[Y]$ below $1$ in each step of the derandomization, and these steps will sequentially define the permutations in $\mathcal{L}$.

\smallskip

\noindent
{\bf Outline of derandomization.} We will choose permutations $\{\pi_1,\pi_2,\ldots,\pi_{\ell}\}$ sequentially, where $\pi_1 = (1,2,\ldots,n)$ is the identity permutation. For some $s \in [\ell-1]$ let permutations $\pi_1,\ldots,\pi_s$ have already been chosen ({\em ``fixed"}). We will chose a {\em ``semi-random"} permutation $\pi_{s+1}$ position by position using $X^i_{s+1}$. Suppose that $\pi_{s+1}(1),$ $ \pi_{s+1}(2),..., \pi_{s+1}(r)$ are already chosen for some $r \in [n-1]$, where all $\pi_{s+1}(i)$ ($i \in [r-1]$) are fixed and final, except $\pi_{s+1}(r)$ which is fixed but not final yet. We will vary $\pi_{s+1}(r) \in [n] \setminus \{\pi_{s+1}(1), \pi_{s+1}(2),..., \pi_{s+1}(r-1)\}$ to choose the best value for $\pi_{s+1}(r)$, assuming that
$\pi_{s+1}(r+1), \pi_{s+1}(r+2),..., \pi_{s+1}(n)$ are random. Permutations $\pi_{s+2},\ldots,\pi_n$ are {\em ``fully-random"}.

\smallskip

\noindent
{\bf Conditional probabilities.}~Given $P_{\gamma} \in \mathcal{P}$, $r \in [n-1]$, note that $X_{s+1}(P_{\gamma})$ depends only on $\pi_{s+1}(1),$ $\pi_{s+1}(2),$ $\ldots,$ $ \pi_{s+1}(r)$. We will show how to compute the conditional probabilities (Algorithm \ref{algo:Cond_prob_2_111}, Section \ref{sec:Cond_Prob_Thm_4_2_111}) $\Prob[X_{s+1}(P_{\gamma}) = 1 \, | \, \pi_{s+1}(1),\pi_{s+1}(2), \ldots, \pi_{s+1}(r)]$, where randomness is over random positions $\pi_{s+1}(r+1),\pi_{s+1}(r+2), \ldots, \pi_{s+1}(n)$. We define $\Prob[X_{s+1}(P_{\gamma}) = 1 \, | \, \pi_{s+1}(1),\pi_{s+1}(2), \ldots, \pi_{s+1}(r)] = \Prob[X_{s+1}(P_{\gamma}) = 1]$ when $r=0$. 

In Section \ref{sec:Cond_Prob_Thm_4_2_111} we will show how to compute conditional probabilities of the atomic events and prove the following Theorem~\ref{theorem:semi-random-conditional_2_111} -- we use this result at the end of the proof of Theorem~\ref{Thm:Derandomization_2_111}:

\begin{theorem}\label{theorem:semi-random-conditional_2_111}
 Suppose values $\pi_{s+1}(1),\pi_{s+1}(2), \ldots, \pi_{s+1}(r)$ have already been fixed for some $r \in \{0\} \cup [n]$. There is a deterministic algorithm to compute $\Prob[X_{s+1}(P_{\gamma}) = 1 \, | \, \pi_{s+1}(1),\pi_{s+1}(2), \ldots, \pi_{s+1}(r)]$, for any positive event $P_{\gamma}$, where the random event is the random choice of the semi-random permutation $\pi_{s+1}$ conditioned on its first $r$ elements already being fixed. Its running time is 
 $$
 O(|Atomic(P_{\gamma})| \cdot (n^2 + nk (k! + \log^2 n))) \, .
 $$
\end{theorem}

\smallskip 

\noindent
{\bf Pessimistic estimator.} Let $P_{\gamma} \in \mathcal{P}$. Denote $\Exp[X_j(P_{\gamma})] = \Prob[X_j(P_{\gamma}) = 1] = \mu_{\gamma j}$ for each $j \in [\ell]$, and $\Exp[X(P_{\gamma})] = \sum_{j=1}^{\ell} \mu_{\gamma j} = \mu_{\gamma}$. By the assumption in Theorem \ref{Thm:Derandomization_2_111}, $\mu_{\gamma j} \geq p_{\gamma}$, for each $j \in [\ell]$. We will now use
Raghavan's proof of the Hoeffding bound, see \cite{Young95}, for any $\delta > 0$, using that $\mu_{\gamma j} \geq p_{\gamma}$ (see more details in Section \ref{sec:derandomization-proofs_2_111}):
\begin{eqnarray*}
  \Prob\left[X(P_{\gamma}) < (1-\delta) \cdot \ell \cdot p_{\gamma} \right]
  &\leq&
  \prod_{j=1}^{\ell} \frac{1-\delta \cdot  \Exp[X_j(P_{\gamma})]}{(1-\delta)^{(1-\delta)p_{\gamma}}}
  <
  \prod_{j=1}^{\ell} \frac{\exp(- \delta \mu_{\gamma j})}{(1-\delta)^{(1-\delta)p_{\gamma}}}
  \leq
  \prod_{j=1}^{\ell} \frac{\exp(- \delta p_{\gamma})}{(1-\delta)^{(1-\delta)p_{\gamma}}} \nonumber \\
  &=&
  \frac{1}{\exp(b(-\delta) \ell p_{\gamma})}
  \,\, < \,\, \frac{1}{\exp(\delta^2 \ell p_{\gamma}/2)} \, ,
\end{eqnarray*} where $b(x) = (1+x) \ln(1+x) - x$, and the last inequality follows by $b(-x) > x^2/2$, see, e.g., \cite{Young95}. Thus, the union bound implies:
\begin{eqnarray}
\Prob\left[\exists P_{\gamma} \in \mathcal{P} : X(P_{\gamma}) < (1-\delta) \cdot \ell \cdot p_{\gamma} \right] \,\, \leq \,\, 
\sum_{\gamma=1}^q \prod_{j=1}^{\ell} \frac{1-\delta \cdot  \Exp[X_j(P_{\gamma})]}{(1-\delta)^{(1-\delta)p_{\gamma}}} \label{Eq:Union_Bound_1_2_111} \, .
\end{eqnarray} 

\noindent
We will derive a pessimistic estimator of this failure probability in (\ref{Eq:Union_Bound_1_2_111}).
Let $\phi_j(P_{\gamma}) = 1$ if $\pi_j$ makes event $P_{\gamma}$ true, and $\phi_j(P_{\gamma}) = 0$ otherwise, and the failure probability (\ref{Eq:Union_Bound_1_2_111}) is at most:
\begin{eqnarray}
  & & \sum_{\gamma=1}^q\prod_{j=1}^{\ell} \frac{1-\delta \cdot \Exp[\phi_j(P_{\gamma})]}{(1-\delta)^{(1-\delta)p_{\gamma}}} \label{eq:first_term_2_111} \\
  &= & \sum_{\gamma=1}^q\left(\prod_{j=1}^{s} \frac{1-\delta \cdot \phi_j(P_{\gamma})}{(1-\delta)^{(1-\delta)p_{\gamma}}}\right) \cdot \left(\frac{1-\delta \cdot \Exp[\phi_{s+1}(P_{\gamma})]}{(1-\delta)^{(1-\delta)p_{\gamma}}}\right) \cdot \left(\frac{1-\delta \cdot \Exp[\phi_j(P_{\gamma})]}{(1-\delta)^{(1-\delta)p_{\gamma}}}\right)^{\ell - s - 1} \label{eq:second_term_2_111} \\
  &\leq& 
  \sum_{\gamma=1}^q\left(\prod_{j=1}^{s} \frac{1-\delta \cdot \phi_j(P_{\gamma})}{(1-\delta)^{(1-\delta)p_{\gamma}}}\right) \cdot \left(\frac{1-\delta \cdot \Exp[\phi_{s+1}(P_{\gamma})]}{(1-\delta)^{(1-\delta)p_{\gamma}}}\right) \cdot \left(\frac{1-\delta \cdot p_{\gamma}}{(1-\delta)^{(1-\delta)p_{\gamma}}}\right)^{\ell - s - 1} \nonumber \\ 
  &=& \, \Phi(\pi_{s+1}(1),\pi_{s+1}(2), \ldots, \pi_{s+1}(r)) \label{Eq:Pessimistic_Est_2_111} \, ,
\end{eqnarray} where equality (\ref{eq:second_term_2_111}) is conditional expectation under: (fixed) permutations $\pi_1,\ldots,\pi_s$ for some $s \in [\ell-1]$, the (semi-random)
permutation $\pi_{s+1}$ currently being chosen, and (fully random) permutations $\pi_{s+2},\ldots,\pi_{\ell}$. The first term (\ref{eq:first_term_2_111}) is less than $\sum_{\gamma=1}^q \exp(-\delta^2 \ell p_{\gamma}/2)$, which is strictly smaller than $1$ for large $\ell$.
  Let us denote
$\Exp[\phi_{s+1}(P_{\gamma})] =
\Exp[\phi_{s+1}(P_{\gamma}) \, | \, \pi_{s+1}(r) = \tau]
= \Prob[X_{s+1}(P_{\gamma}) = 1 \, | \, \pi_{s+1}(1),\pi_{s+1}(2), \ldots, \pi_{s+1}(r-1), \pi_{s+1}(r) = \tau]$, where positions $\pi_{s+1}(1),\pi_{s+1}(2), \ldots, \pi_{s+1}(r)$ were fixed in the semi-random permutation $\pi_{s+1}$, $\pi_{s+1}(r)$ was fixed in particular to $\tau \in [n] \setminus \{\pi_{s+1}(1),\pi_{s+1}(2), \ldots, \pi_{s+1}(r-1)\}$, and it can be computed by using the algorithm from Theorem~\ref{theorem:semi-random-conditional_2_111}.
This gives our pessimistic estimator $\Phi$, when the semi-random permutation $\pi_{s+1}$ is being decided. Observe that $\Phi$ can be rewritten as
\vspace*{-1ex}
\begin{eqnarray}
  \Phi = \sum_{\gamma=1}^q \omega(\ell,s,\gamma) \cdot  \left(\prod_{j=1}^{s} (1-\delta \cdot \phi_j(P_{\gamma}))\right) \cdot (1-\delta \cdot \Exp[\phi_{s+1}(P_{\gamma})]), \, 
  \mbox{ where } \label{Eq:Pessimistic_Est_Obj_111} \\ \,  \omega(\ell,s,\gamma) = 
  \frac{\left(1-\delta p_{\gamma}\right)^{\ell-s+1}}{(1-\delta)^{(1-\delta)p_{\gamma}\ell}} \, . \nonumber
\end{eqnarray}

\noindent
Recall $\pi_{s+1}(r)$ in semi-random permutation was fixed but not final. To make it final, we choose $\pi_{s+1}(r) \in [n] \setminus \{\pi_{s+1}(1),\pi_{s+1}(2), \ldots, \pi_{s+1}(r-1)\}$ that minimizes $\Phi$ in (\ref{Eq:Pessimistic_Est_Obj_111}). Proof of Lemma  \ref{lem:potential_correct_2_111}
can be found in Section \ref{sec:derandomization-proofs_2_111}.

\begin{lemma}\label{lem:potential_correct_2_111}
$\Phi(\pi_{s+1}(1),\pi_{s+1}(2), \ldots, \pi_{s+1}(r))$ is a pessimistic estimator of  failure probability (\ref{Eq:Union_Bound_1_2_111}), if
$\ell \geq \frac{2\log{q}}{\delta^2 p_0}$.
\end{lemma}

\begin{proofof}{Theorem \ref{Thm:Derandomization_2_111}} 
See the precise details of this proof in Section \ref{sec:derandomization-proofs_2_111}. 
\end{proofof}

\begin{algorithm}[ht!]
\SetAlgoVlined

\DontPrintSemicolon
 \KwIn{Positive integers $n$, $k \leq n$, $\ell \geq 2$, such that $\log k \geq 8$.}
 \KwOut{A multi-set $\mathcal{L} \subseteq \Pi_n$ of $\ell$ permutations.}
 
 /* This algorithm uses Function ${\sf Prob}(A)$ defined in Algorithm \ref{algo:Cond_prob_2_111}. */ 
 
 $\pi_1 := (1,2,\ldots,n)$  /* Identity permutation */
 
 $\mathcal{L} := \{\pi_1\}$
 
 Let $\mathcal{P} = \{P_1,\ldots,P_q\}$ be the set of all positive events.  
 
 \For{$\gamma \in \{1,\ldots,q\}$}{$w(P_{\gamma}) := 1-\delta \cdot \phi_1(P_{\gamma})$ \label{Alg:Weight_Init_2}}
 
    \For{$s = 1 \ldots \ell - 1$}{
      \For{$r = 1 \ldots n$}{
         \For{$\gamma \in \{1,\ldots,q\}$}{
           \For{$\tau \in [n] \setminus \{\pi_{s+1}(1),\pi_{s+1}(2), \ldots, \pi_{s+1}(r-1)\}$}{
            \For{$A \in Atomic(P_{\gamma})$}{
             $\Prob[A \, | \, \pi_{s+1}(1), \ldots, \pi_{s+1}(r-1), \pi_{s+1}(r) = \tau] := {\sf Prob}(A)$}
           
             $\Exp[\phi_{s+1}(P_{\gamma}) \, | \, \pi_{s+1}(r) = \tau] := \sum_{A \in Atomic(P_{\gamma})} \Prob[A \, | \, \pi_{s+1}(1), \ldots, \pi_{s+1}(r-1), \pi_{s+1}(r) = \tau]$
            }
           }
           Choose $\pi_{s+1}(r) = \tau$ for $\tau \in [n] \setminus \{\pi_{s+1}(1),\pi_{s+1}(2), \ldots, \pi_{s+1}(r-1)\}$ to minimize
           $\sum_{\gamma=1}^q \omega(\ell,s,\gamma) \cdot w(P_{\gamma}) \cdot (1 - \delta \cdot \Exp[\phi_{s+1}(P_{\gamma}) \, | \, \pi_{s+1}(r) = \tau])$.
       }
      $\mathcal{L} := \mathcal{L} \cup \{\pi_{s+1}\}$
      
      \For{$\gamma \in \{1,\ldots,q\}$}{
        $w(P_{\gamma}) :=  w(P_{\gamma}) \cdot (1-\delta \cdot \phi_{s+1}(P_{\gamma}))$ \label{Alg:Weight_Update_2}
       }
     }
 \Return $\mathcal{L}$
 \caption{Find permutations distribution (for positive events)}
 \label{algo:Find_perm_2_111}
\end{algorithm}

\subsection{Conditional probabilities and proof of Theorem \ref{theorem:semi-random-conditional_2_111}}\label{sec:Cond_Prob_Thm_4_2_111}

We will show how to compute conditional probabilities of the atomic events. Let us first recall a definition of an atomic event. Given $t \in \{2,3,\ldots, 
n\}$, let $\mathcal{B} := B_{1}, B_{2}, \ldots, B_{t}$, be a bucketing of the sequence $(1, \ldots, n)$. Let there be indices $1 = b_{0} < b_{1} < b_{2} < \cdots < b_{t-1} < b_t = n$, $\forall j \in [t-1] : b_j \in [n]$, such that $B_1 = \{1,\ldots,b_{1}\}$, and $B_j = \{b_{j-1}+1,\ldots,b_{j}\}$, for $j \in \{2,3,\ldots, t\}$.

Consider an ordered subset $\sigma = (\sigma_{1}, \ldots, \sigma_{k}) = (\sigma(1), \ldots, \sigma(k))$ of the set $[n]$. If $f : [k] \longrightarrow [t]$ be a non-decreasing mapping, an atomic event for $\sigma$ and $f$ in the probabilistic space $\Omega$ is:
$$A_{\sigma, f} = \{\pi \in \Omega : \forall_{i \in [k]} \pi^{-1}(\sigma_{i}) \in B_{f(i)} \text{ and } \pi^{-1}(\sigma_{1}) < \pi^{-1}(\sigma_{2}) < \ldots < \pi^{-1}(\sigma_{k}) \}.$$

Recall how we generate a random permutation $\pi_j$ by the index random variables $X^1_j, X^2_j, \ldots,$ $ X^n_j$, which generate its elements $\pi_j(1), \pi_j(2), \ldots, \pi_j(n)$ sequentially in this order.
 We naturally extend the 
definition of the random variable $X^{P}_j \in  \{0,1\}$ and function $\phi_j(P) \in \{0,1\}$ to $X^{A_{\sigma, f}}_j$ and $\phi_j(A_{\sigma, f})$ for the atomic event $A_{\sigma, f}$, and random permutation $\pi_j$, $j \in [\ell]$.
  We will define an algorithm to 
compute $\Prob[X^{A_{\sigma, f}}_{s+1} = 1 \, | \, \pi_{s+1}(1),\pi_{s+1}(2), \ldots, \pi_{s+1}(r)]$ for the semi-random permutation $\pi_{s+1}$. Slightly abusing notation we let for $r=0$ to have that $\Prob[X^{A_{\sigma, f}}_{s+1} = 1 \, | \, \pi_{s+1}(1),\pi_{s+1}(2), \ldots, \pi_{s+1}(r)] = \Prob[X^{A_{\sigma, f}}_{s+1} = 1]$. In this case, we will also show below how to compute $\Prob[X^{A_{\sigma, f}}_{s+1} = 1]$ when $\pi_{s+1}$ is fully random.
\\

\noindent
{\bf Proof of Theorem \ref{theorem:semi-random-conditional_2_111}.} We will now present the proof of Theorem \ref{theorem:semi-random-conditional_2_111}. If $r=0$ and $\pi_{s+1}$ is fully random. Note that $B_j' = f^{-1}(\{j\}) \subseteq [k]$ is the set of indices of elements from $\sigma$ that $f$ maps to bucket $j \in [t]$, and  
$b_j' = |B_j'|$ is their number; $b_j'$ can also be zero. Let us also denote $\chi(j',j) = \sum_{i=j'}^{j-1} b_i'$ for $j', j \in \{1,2,\ldots,t\}$ assuming $j' < j$, and $\chi(j',j) = 0$ when $j'=j$. Using Bayes’ formula on conditional probabilities and the definition of event $A_{\sigma, f}$ we have that
\begin{eqnarray}
\Prob[A_{\sigma, f}] = \Prob[X^{A_{\sigma, f}}_{s+1} = 1] = \prod_{j=1}^{t}
\left(\left(\prod_{i=1}^{b_j'} \frac{|B_j| - (i-1)}{n - \chi(1,j) - (i-1)}\right) \cdot \frac{perm(B_j',\sigma)}{(b_j')!}\right) \, , \label{eq:atomic_prob_1}
\end{eqnarray} where $perm(B_j',\sigma)$ denotes the number of all permutations of elements in the set $B_j'$ that agree with their order in permutation $\sigma$. Note that $perm(B_j',\sigma)$ can be computed by enumeration in time $(b_j')! \leq k!$.
 
Assume from now on that $r \geq 1$. Suppose that values $\pi_{s+1}(1),\pi_{s+1}(2), \ldots, \pi_{s+1}(r)$ have already been chosen for some $r \in [n]$, i.e., they all are fixed and final, except that $\pi_{s+1}(r)$ is fixed but not final. The algorithm will be based on an observation that the random process of generating the remaining values $\pi_{s+1}(r+1),\pi_{s+1}(r+2), \ldots, \pi_{s+1}(n)$ can be viewed as choosing u.a.r.~a random permutation of values in set $[n] \setminus \{\pi_{s+1}(1),\pi_{s+1}(2), \ldots, \pi_{s+1}(r)\}$; so this random permutation
has length $n-r$.


\begin{algorithm}[!h]
\SetAlgoVlined
\DontPrintSemicolon
\SetKwFunction{FMain}{${\sf Prob}$}
\SetKwProg{Fn}{Function}{:}{}
  \Fn{\FMain{$A_{\sigma, f}$}}{
   Let $j \in [t]$ be such that $r \in B_j = \{b_{j-1} + 1,\ldots, b_j\}$. \;
   Let $I_j = \{\sigma(i) : i \in f^{-1}(\{j\})\}$ for $j \in [t]$. /* $I_j$ $=$ items from $[n]$ mapped by $f$ to $B_j$ */\;
   \If{$j > 1$ or $r = b_j$}{
       \lIf{$r = b_j$}{$j'' := j$}\lElse{$j'' := j-1$}
       \If{$\exists j' \in [j''] : I_{j'} \not \subseteq
           \{\pi_{s+1}(b_{j'-1}+1),\pi_{s+1}(b_{j'-1}+2),\ldots,\pi_{s+1}(b_{j'})\}$ \label{algo:Cond_prob_2_111_step_1}}{
             $q := \Prob[A_{\sigma, f}] = 0$; \Return $q$}
        \If{$\exists i_1, i_2 \in \bigcup_{j' \in [j'']}       
             I_{j'} \, : \, 
             \pi^{-1}_{s+1}(i_1) < \pi^{-1}_{s+1}(i_2)$
             $\mbox{ and } \sigma^{-1}(i_1) > \sigma^{-1}(i_2)$ \label{algo:Cond_prob_2_111_step_2}}{
             $q := \Prob[A_{\sigma, f}] = 0$; \Return $q$}
       }
      \If{$r = b_j$ \label{algo:Cond_prob_2_111_step_3}}{
        $q:= \Prob[A_{\sigma, f}] = \prod_{\kappa=j+1}^{t}
\left(\left(\prod_{i=1}^{b'_{\kappa}} \frac{|B_{\kappa}| - (i-1)}{n - r - \chi(j+1,\kappa) - (i-1)}\right) \cdot \frac{perm(B'_{\kappa},\sigma)}{(b'_{\kappa})!}\right)$ \label{algo:Cond_prob_2_111_step_4}
       }
      \If{$r < b_j$}{
          Let $J=\{\pi_{s+1}(b_{j-1}+1),\pi_{s+1}(b_{j-1}+2),\ldots,\pi_{s+1}(r)\}$. \;
         \If{$|I_{j} \cap J| + b_j - r < |I_{j}|$ \label{algo:Cond_prob_2_111_cond_a}}{
         $q := \Prob[A_{\sigma, f}] = 0$; \Return $q$ \label{algo:Cond_prob_2_111_step_5}}
         \If{$\exists i_1, i_2 \in I_{j} \cap J \, : \, 
             \pi^{-1}_{s+1}(i_1) < \pi^{-1}_{s+1}(i_2)$
             $\mbox{ and } \sigma^{-1}(i_1) > \sigma^{-1}(i_2) $ \label{algo:Cond_prob_2_111_cond_b}}{
         $q := \Prob[A_{\sigma, f}] = 0$; \Return $q$}
         Let $I = I_{j} \setminus J$, and let $I' = \{\sigma^{-1}(i) : i \in I\}$.\; \label{algo:Cond_prob_2_111_step_6}
          $p_0 := \left(\prod_{i=1}^{|I|} \frac{b_j - r - (i-1)}{n - r - (i-1)}\right) \cdot \frac{perm(I',\sigma)}{|I'|!}$\;
          $q:= \Prob[A_{\sigma, f}] = p_0 \cdot \prod_{\kappa=j+1}^{t}
\left(\left(\prod_{i=1}^{b'_{\kappa}} \frac{|B_{\kappa}| - (i-1)}{n - r - |I'| - \chi(j+1,\kappa) - (i-1)}\right) \cdot \frac{perm(B'_{\kappa},\sigma)}{(b'_{\kappa})!}\right)$ \label{algo:Cond_prob_2_111_step_7}
       }
 \KwRet $\Prob[A_{\sigma, f}] = q$; \;
} 
\caption{Conditional probabilities (Atomic events)}
 \label{algo:Cond_prob_2_111}
\end{algorithm}


\begin{lemma}\label{l:Prob_Algo_2_111} Algorithm \ref{algo:Cond_prob_2_111}, called ${\sf Prob}(A_{\sigma, f})$, correctly computes
$$\Prob[A_{\sigma, f} \, | \, \pi_{s+1}(1),\pi_{s+1}(2), \ldots, \pi_{s+1}(r)]$$ in time $O(n^2 + nk (k! + \log^2 n))$.
\end{lemma}

\begin{proof} 
We will show first the correctness. For simplicity, we will write below $\Prob[A_{\sigma, f}]$ instead of $\Prob[A_{\sigma, f} \, | \, \pi_{s+1}(1),\pi_{s+1}(2), \ldots, \pi_{s+1}(r)]$.
  When computing the conditional 
probability $\Prob[A_{\sigma, f}] = \Prob[A_{\sigma, f} \, | \, \pi_{s+1}(1),\pi_{s+1}(2), \ldots, \pi_{s+1}(r)]$, we will use the formula (\ref{eq:atomic_prob_1}).
  Algorithm 
\ref{algo:Cond_prob_2_111} first finds the bucket $B_j$ which contains the last position $r$ of the semi-random permutation $\pi_{s+1}$. Then, the algorithm finds all buckets $B_{j'}$, $j' \in [j'']$ whose all items in their positions are fixed in $\pi_{s+1}$.

Then, in line \ref{algo:Cond_prob_2_111_step_1} we check if there is any previous bucket $B_{j'}$, $j' \in [j'']$ among buckets with all fixed 
positions in $\pi_{s+1}$, which does not contain the set of items $I_{j'} \subseteq [n]$ that mapping $f$ maps to that bucket. If this is the case then clearly $\pi_{s+1}$ does not fulfil event $A_{\sigma, f}$, so $\Prob[A_{\sigma, f}] = 0$.

Then, in line \ref{algo:Cond_prob_2_111_step_2} we check if there are any two items $i_1, i_2 \in [n]$ that are mapped by $f$ to the union of the fully fixed buckets $\bigcup_{j' \in [j'']} B_{j'}$, whose order in permutation $\pi_{s+1}$ disagrees with their order in permutation $\sigma$. If this is the case then again $\pi_{s+1}$ disagrees with event $A_{\sigma, f}$, so $\Prob[A_{\sigma, f}] = 0$.

If none of the conditions in lines \ref{algo:Cond_prob_2_111_step_1} and  \ref{algo:Cond_prob_2_111_step_2} hold, then the part of 
the permutation $\pi_{s+1}$ with fully fixed buckets fulfills event $A_{\sigma, f}$ for all buckets $B_{j'}, j' \in [j'']$. Thus the probability of event $A_{\sigma, f}$ depends only on the random positions $\pi_{s+1}(r+1), \ldots, \pi_{s+1}(n)$ in $\pi_{s+1}$, and the algorithm continues in line \ref{algo:Cond_prob_2_111_step_3}.

If $r=b_j$ holds in line \ref{algo:Cond_prob_2_111_step_3}, then we know that the event holds for all buckets $B_{j'}, j' \in [j]$ with fixed positions, and in this case buckets $B_{j'}, j' \in \{j+1,j+2\ldots,n\}$ are are all fully random. Therefore, this situation is the same as the fully random permutation $\pi_{s+1}$ where formula (\ref{eq:atomic_prob_1}) applies to all buckets $B_{j'}, j' \in \{j+1,j+2\ldots,n\}$, where we start with bucket $B_{j+1}$ instead of bucket $B_1$ -- see line \ref{algo:Cond_prob_2_111_step_4}.

If $r < b_j$ holds, then bucket $B_j$ is semi-random and we know by the check in the previous lines that event $A_{\sigma, f}$ holds for all buckets $B_{j'}, j' \in [j-1]$ with fully fixed positions. Then the set $J$ contains all items on fixed positions in bucket $B_j$. If $|I_j \cap J| + b_j-r < |I_j|$, then there is no enough room in bucket $B_j$ to accommodate for all items from set $I_j$ that are mapped by $f$ to bucket $B_j$. Therefore in this case event $A_{\sigma, f}$ does not hold for permutation $\pi_{s+1}$ and so $\Prob[A_{\sigma, f}] = 0$ in line \ref{algo:Cond_prob_2_111_step_5}.   
Now, if we are in line \ref{algo:Cond_prob_2_111_step_6}, then we know that condition in line \ref{algo:Cond_prob_2_111_cond_a} does not hold, implying that there is enough space in bucket $B_j$ in its random positions to accommodate for the the items $I_j \setminus J$ mapped by $f$ to bucket $B_j$. We also know that condition in line \ref{algo:Cond_prob_2_111_cond_b} does not hold, which means that event $A_{\sigma, f}$ holds in the fixed positions of bucket $B_j$, implying that its probability depends only on the still random positions $\pi_{s+1}(r+1),\pi_{s+1}(r+2), \ldots, \pi_{s+1}(n)$ in permutation $\pi_{s+1}$.

Therefore, we use now similar ideas to those from in formula (\ref{eq:atomic_prob_1}) and line \ref{algo:Cond_prob_2_111_step_4}, and calculate the probability $p_0$ of the part of event $A_{\sigma, f}$ in the random positions $\pi_{s+1}(r+1),\pi_{s+1}(r+2), \ldots, \pi_{s+1}(b_j)$ in bucket $B_j$; observe that this is the probability that the items from set $I_j \setminus J$ are mapped by $f$ to these random positions in bucket $B_j$, multiplied by the probability $\frac{perm(I',\sigma)}{|I'|!}$ that their order in permutation $\pi_{s+1}$ agrees with the order of their indices (in the set $I' \subseteq [k]$) in permutation $\sigma$.

The final probability of event $A_{\sigma, f}$ computed in line \ref{algo:Cond_prob_2_111_step_7} is by the Bayes' formula on conditional probabilities, equal to the product of $p_0$ and the probability of that event on the remaining buckets $B_{j'}, j' \in \{j+1,j+2\ldots,n\}$ with fully random positions (we use here the same formula as in line \ref{algo:Cond_prob_2_111_step_4}).
 
We now show how to implement the algorithm. One kind of operations are operations on subsets of set $[n]$, set membership and intersection, which can easily be done in time $O(n)$. The other kind of operations in computing probabilities are divisions of numbers from set $[n]$ and multiplications of resulting rational expressions. Each of these arithmetic operations can be performed in time $O(\log^2(n))$.

The conditions in lines \ref{algo:Cond_prob_2_111_step_1}, \ref{algo:Cond_prob_2_111_step_2}, \ref{algo:Cond_prob_2_111_cond_a} and \ref{algo:Cond_prob_2_111_cond_b} can each be easily checked in time $O(n^2)$. The formula in line \ref{algo:Cond_prob_2_111_step_4} and in line \ref{algo:Cond_prob_2_111_step_7} can be computed in time $O(nk k! \log^2(n))$ as follows. The outer product has at most $n$ terms. The inner product has at most $b'_{\kappa} \leq k$ products of fractions; each of these fractions is a fraction of two numbers from $[n]$, thus each number having $\log n$ bits. This means that each of these fractions can be computed in time $O(\log^2 n)$ and their products in time $O( b'_{\kappa} \cdot \log^2 n) = O(k \cdot \log^2 n)$. The at most $n$ fractions $\frac{perm(I',\sigma)}{|I'|!}$ each can be computed in time $O(k \cdot k!) + O(\log^2 (k!)) = O(k \cdot k!) + O(k^2 \log^2 k)) = O(k \cdot k!)$, because we compute $perm(I',\sigma)$ by complete enumeration of all $k!$ permutations, and division $\frac{perm(I',\sigma)}{|I'|!}$ is computed in time $O(\log^2(k!))$. Therefore the total time is $O(n^2 + nk (k! + \log^2 n))$.
\end{proof}

\noindent
To finish the proof of  Theorem \ref{theorem:semi-random-conditional_2_111}, we use Algorithm \ref{algo:Cond_prob_2_111} from Lemma \ref{l:Prob_Algo_2_111} for each event $A \in Atomic(P)$.

\subsection{Missing details in proof of Theorem \ref{Thm:Derandomization_2_111} }\label{sec:derandomization-proofs_2_111}

To be precise, the adversary assigned value $v(u)$ (the $u$-th largest adversarial value, $u \geq 1$) to the position $j_u$ in his/her permutation and the random permutation $\pi \in \Pi_n$ places this value at the position $\pi^{-1}(j_u)$, for each $u \in \{1,2,\ldots,k\}$.\\

\noindent
{\bf Pessimistic estimator.} Let $P_{\gamma} \in \mathcal{P} = \{P_1,\ldots,P_q\}$, $\gamma \in [q]$. Recall that $X(P_{\gamma}) =
X_1(P_{\gamma}) + \cdots + X_{\ell}(P_{\gamma})$ and denote $\Exp[X_j(P_{\gamma})] = \Prob[X_j(P_{\gamma}) = 1] = \mu_{\gamma j}$ for each $j \in [\ell]$, and $\Exp[X(P_{\gamma})] = \sum_{j=1}^{\ell} \mu_{\gamma j} = \mu_{\gamma}$. By the assumption in Theorem \ref{Thm:Derandomization_2_111}, $\mu_j \geq p$, for each $j \in [\ell]$. We will now use
Raghavan's proof of the Hoeffding bound, see \cite{Young95}, for any $\delta > 0$, using that $\mu_j \geq p$:
\begin{eqnarray*}
	\Prob\left[X(P_{\gamma}) < (1-\delta) \cdot \mu_{\gamma}\right]
	&=&
    \Prob\left[\prod_{j=1}^{\ell} \frac{(1-\delta)^{X_j(P_{\gamma})}}{(1-\delta)^{(1-\delta)\mu_{\gamma j}}} \geq 1\right] \\
    &\leq& 
    \Exp\left[\prod_{j=1}^{\ell} \frac{1-\delta \cdot  X_j(P_{\gamma})}{(1-\delta)^{(1-\delta)\mu_{\gamma j}}}\right] 
    =
    \prod_{j=1}^{\ell} \frac{1-\delta \cdot  \Exp[X_j(P_{\gamma})]}{(1-\delta)^{(1-\delta)\mu_{\gamma j}}} \\
    &<&
    \prod_{j=1}^{\ell} \frac{\exp(- \delta \mu_{\gamma j})}{(1-\delta)^{(1-\delta)\mu_{\gamma j}}} 
    =
    \frac{1}{\exp(b(-\delta) \mu_{\gamma})}\, ,
\end{eqnarray*} where $b(x) = (1+x) \ln(1+x) - x$, and the second step uses Bernoulli's inequality $(1+x)^r \leq 1 + rx$, that holds for $0 \leq r \leq 1$ and $x \geq -1$, and
Markov's inequality, and the last inequality uses $1-x \leq \exp(-x)$, which holds for $x \geq 0$ and is strict if $x \not = 0$. 

By $\mu_{\gamma j} \geq p_{\gamma}$, for each $j \in [\ell]$, we can further upper bound the last line of Raghavan's proof to obtain $\frac{1}{\exp(b(-\delta) \mu_{\gamma})} \leq \frac{1}{\exp(b(-\delta) \ell p_{\gamma})}$. Theorem \ref{theorem:Chernoff_Positive_Events} guarantees existence of the multi set $\mathcal{L}$ of permutations by bounding
$\Prob[X(P_{\gamma}) < (1-\delta) \cdot p_{\gamma}  \ell]  \leq  \exp(- \delta^2 p_{\gamma} \ell/2)$, see (\ref{eqn:Chernoff_Hoeffding_111}). Now, repeating the Raghavan's proof of the Chernoff-Hoeffding bound, cf. \cite{Young95}, with each $\mu_{\gamma j}$ replaced by $p_{\gamma}$ implies that
\begin{eqnarray}
  \Prob\left[X(P_{\gamma}) < (1-\delta) \cdot p_{\gamma} \ell \right]
  &\leq&
  \prod_{j=1}^{\ell} \frac{1-\delta \cdot  \Exp[X_j(P_{\gamma})]}{(1-\delta)^{(1-\delta)p_{\gamma}}} \label{Eq:Raghavan-1_2} \\
  &<&
  \prod_{j=1}^{\ell} \frac{\exp(- \delta \mu_{\gamma j})}{(1-\delta)^{(1-\delta)p_{\gamma}}} \nonumber 
  \leq
  \prod_{j=1}^{\ell} \frac{\exp(- \delta p_{\gamma})}{(1-\delta)^{(1-\delta)p_{\gamma}}} \nonumber \\
  &=&
  \frac{1}{\exp(b(-\delta) \ell p_{\gamma})}
  \,\, < \,\, \frac{1}{\exp(\delta^2 \ell p_{\gamma}/2)} \label{Eq:Raghavan-2_2} \, ,
\end{eqnarray} where the last inequality follows by a well known fact that $b(-x) > x^2/2$, see, e.g., \cite{Young95}. By this argument and by the union bound we obtain that (note that $\mathcal{P} =\{P_1,\ldots,P_q\}$):
\begin{eqnarray}
\Prob\left[\exists P_{\gamma} \in \mathcal{P} : X(P_{\gamma}) < (1-\delta) \cdot p_{\gamma} \ell \right] \,\, \leq \,\, 
\sum_{\gamma = 1}^q \prod_{j=1}^{\ell} \frac{1-\delta \cdot  \Exp[X_j(P_{\gamma})]}{(1-\delta)^{(1-\delta)p_{\gamma}}} \label{Eq:Union_Bound_2} \, .
\end{eqnarray} 

Let us define a function $\phi_j(P_{\gamma})$ as $\phi_j(P_{\gamma})=1$ if $\pi_j$ makes event $P_{\gamma}$ true, and $\phi_j(P_{\gamma}) = 0$ otherwise. The above proof upper bounds the probability of failure by the expected value of
$$
  \sum_{\gamma=1}^q \prod_{j=1}^{\ell} \frac{1-\delta \cdot \phi_j(P_{\gamma})}{(1-\delta)^{(1-\delta)p_{\gamma}}} \, ,
$$ the expectation of which is less than $\sum_{\gamma=1}^q \exp(-\delta^2 \ell p_{\gamma}/2)$, which is strictly smaller than $1$ for appropriately chosen large $\ell$.

 Suppose that we have so far chosen the (fixed) permutations $\pi_1,\ldots,\pi_s$ for some $s \in \{1,2,\ldots,\ell-1\}$, the (semi-random)
permutation $\pi_{s+1}$ is currently being chosen, and the remaining (fully random) permutations, if any, are $\pi_{s+2},\ldots,\pi_{\ell}$. The conditional expectation is then
\begin{eqnarray}
 & & \sum_{\gamma=1}^q\left(\prod_{j=1}^{s} \frac{1-\delta \cdot \phi_j(P_{\gamma})}{(1-\delta)^{(1-\delta)p_{\gamma}}}\right) \cdot \left(\frac{1-\delta \cdot \Exp[\phi_{s+1}(P_{\gamma})]}{(1-\delta)^{(1-\delta)p_{\gamma}}}\right) \cdot \left(\frac{1-\delta \cdot \Exp[\phi_j(P_{\gamma})]}{(1-\delta)^{(1-\delta)p_{\gamma}}}\right)^{\ell - s - 1} \nonumber \\
  &\leq& 
  \sum_{\gamma=1}^q\left(\prod_{j=1}^{s} \frac{1-\delta \cdot \phi_j(P_{\gamma})}{(1-\delta)^{(1-\delta)p_{\gamma}}}\right) \cdot \left(\frac{1-\delta \cdot \Exp[\phi_{s+1}(P_{\gamma})]}{(1-\delta)^{(1-\delta)p_{\gamma}}}\right) \cdot \left(\frac{1-\delta \cdot p_{\gamma}}{(1-\delta)^{(1-\delta)p_{\gamma}}}\right)^{\ell - s - 1} \nonumber \\ 
  &=& \, \Phi(\pi_{s+1}(1),\pi_{s+1}(2), \ldots, \pi_{s+1}(r)) \label{Eq:Pessimistic_Est-2_2} \, ,
\end{eqnarray} where in the inequality, we used that $\Exp[\phi_j(P_{\gamma})] \geq p_{\gamma}$. Note, that 
\begin{eqnarray*}
\Exp[\phi_{s+1}(P_{\gamma})] &=&
\Exp[\phi_{s+1}(P_{\gamma}) \, | \, \pi_{s+1}(r) = \tau] \\
&=& \Prob[X_{s+1}(P_{\gamma}) = 1 \, | \, \pi_{s+1}(1),\pi_{s+1}(2), \ldots, \pi_{s+1}(r-1), \pi_{s+1}(r) = \tau] \, ,
\end{eqnarray*} 
where positions $\pi_{s+1}(1),\pi_{s+1}(2), \ldots, \pi_{s+1}(r)$ have already been fixed in the semi-random permutation $\pi_{s+1}$, $\pi_{s+1}(r)$ has been fixed in particular to $\tau \in [n] \setminus \{\pi_{s+1}(1),\pi_{s+1}(2), \ldots, \pi_{s+1}(r-1)\}$, and this value can be computed by using the algorithm from Theorem \ref{theorem:semi-random-conditional_2_111}.
This gives the pessimistic estimator $\Phi$ of the failure probability in (\ref{Eq:Union_Bound_2}) for our derandomization.


We can rewrite our pessimistic estimator $\Phi$ as follows
$$
  \sum_{\gamma=1}^q \omega(\ell,s,\gamma) \cdot  \left(\prod_{j=1}^{s} \left(1-\delta \cdot \phi_j(P_{\gamma})\right)\right) \cdot \left(1-\delta \cdot \Exp[\phi_{s+1}(P_{\gamma})]\right) \, , \mbox{ where } \,
  \omega(\ell,s,\gamma) = 
  \frac{\left(1-\delta p_{\gamma}\right)^{\ell-s+1}}{(1-\delta)^{(1-\delta)p_{\gamma}\ell}} \, .
$$
Recall that the value of $\pi_{s+1}(r)$ in the semi-random permutation was fixed but not final. To make it fixed and final, we simply choose the value $\pi_{s+1}(r) \in [n] \setminus \{\pi_{s+1}(1),\pi_{s+1}(2), \ldots, \pi_{s+1}(r-1)\}$ that minimizes this last expression
\begin{eqnarray}
  \sum_{\gamma=1}^q \omega(\ell,s,\gamma) \cdot  \left(\prod_{j=1}^{s} \left(1-\delta \cdot \phi_j(P_{\gamma})\right)\right) \cdot \left(1-\delta \cdot \Exp[\phi_{s+1}(P_{\gamma})]\right) \, . \label{Eq:Pessimistic_Est_Obj-2_2}
\end{eqnarray}


\begin{proofof}{Lemma \ref{lem:potential_correct_2_111}}
  This follows from the following three properties: (a) it is an upper bound on the conditional probability of failure;  
  (b) it is initially strictly less than $1$; (c) some new value of the next index variable in the partially fixed semi-random permutation 
  $\pi_{s+1}$ can always be chosen without increasing it.
  
  Property (a) follows from (\ref{Eq:Raghavan-1_2}) and (\ref{Eq:Union_Bound_2}). To prove (b) we see by (\ref{Eq:Raghavan-2_2}) and (\ref{Eq:Union_Bound_2}) that
  $$
  \Prob\left[\exists P_{\gamma} \in \mathcal{P} : X(P_{\gamma}) < (1-\delta) \cdot p_{\gamma} \ell \right] < \sum_{\gamma=1}^q \exp(- \delta^2 \ell p_{\gamma}){\gamma}/2) \, .
  $$ Note that $\ell \geq \frac{2\log{q}}{\delta^2 p_0}$ implies that  $\sum_{\gamma=1}^q \exp(- \delta^2 \ell p_{\gamma}){\gamma}/2) \leq 1$, see the proof of Theorem \ref{theorem:Chernoff_Positive_Events}. (a) and (b) follow by the above arguments and by the assumption about $\ell$.
  
  Part (c) follows because $\Phi$ is an expected value conditioned on the choices made so far. For the precise argument let us observe that
\begin{eqnarray*}
  & & \Prob[X_{s+1}(P_{\gamma}) = 1 \, | \, \pi_{s+1}(1),\pi_{s+1}(2), \ldots, \pi_{s+1}(r-1)] \\
  &=& \sum_{\tau \in T}  \frac{1}{n-r+1} \cdot \Prob[X_{s+1}(P_{\gamma}) = 1 \, | \, \pi_{s+1}(1),\pi_{s+1}(2), \ldots, \pi_{s+1}(r-1), \pi_{s+1}(r) = \tau] \, ,
\end{eqnarray*} where $T = [n] \setminus \{\pi_{s+1}(1),\pi_{s+1}(2), \ldots, \pi_{s+1}(r-1)\}$. Then by (\ref{Eq:Pessimistic_Est-2_2}) we obtain
\begin{eqnarray*}
& &
\Phi(\pi_{s+1}(1),\pi_{s+1}(2), \ldots, \pi_{s+1}(r-1)) \\
&=&
\sum_{\tau \in T}  \frac{1}{n-r+1} \cdot \Phi(\pi_{s+1}(1),\pi_{s+1}(2), \ldots, \pi_{s+1}(r-1), \pi(r) = \tau) \\
&\geq& 
\min \{ \Phi(\pi_{s+1}(1),\pi_{s+1}(2), \ldots, \pi_{s+1}(r-1), \pi(r) = \tau) \,\, : \,\, \tau \in T \} \, ,
\end{eqnarray*} which implies part (c). This finishes the proof of Lemma \ref{lem:potential_correct_2_111}.
\end{proofof}

\begin{proofof}{Theorem \ref{Thm:Derandomization_2_111}}
The computation of the conditional probabilities ${\sf Prob}(A)$, for any atomic event $A$, by Algorithm \ref{algo:Cond_prob_2_111} is correct by Theorem \ref{theorem:semi-random-conditional_2_111}. Algorithm \ref{algo:Find_perm_2_111} is a direct translation of the optimization of the pessimistic estimator $\Phi$. In particular, observe that the correctness of the weight initialization in Line \ref{Alg:Weight_Init_2} of Algorithm \ref{algo:Find_perm_2_111}, and of weight updates in Line \ref{Alg:Weight_Update_2}, follow from the form of the pessimistic estimator objective function in (\ref{Eq:Pessimistic_Est_Obj-2_2}).

The value of the pessimistic estimator $\Phi$ is strictly smaller than $1$ at the beginning and in each step, it is not increased by properties of the pessimistic estimator (Lemma \ref{lem:potential_correct_2_111}). Moreover, at the last step all values of all $\ell$ permutations will be fixed, that is, there will be no randomness in the computation of $\Phi$. Observe that $\Phi$ is an upper bound on the expected number of the positive events from $\mathcal{P}$ that are not well-covered. So at the end of the derandomization process the number of such positive events will be $0$, implying that all these positive events from $\mathcal{P}$ will be well-covered, as desired.

Using Theorem \ref{theorem:semi-random-conditional_2_111} it is straightforward to count the number of operations in Algorithm \ref{algo:Find_perm_2_111} as follows
$$
  O\left(  
    \ell \cdot \left( n \cdot \left( q \cdot \left(n \cdot |Atomic(P)|^2 \cdot (n^2 + nk (k! + \log^2 n))
                                             \right) 
                                     + n
                              \right) 
                      + qn 
               \right)  
  \right) =
$$
$$
  = O\left( \ell n^2 q \cdot |Atomic(P)|^2 \cdot (n^2 + nk (k! + \log^2 n)) \right)
  = O\left( \ell n^3 q \cdot t^{2k} \cdot (k!)^2 \cdot \left(n + k k! + k \log^2 n\right) \right) \, ,
$$ where we used that $|Atomic(P)| \leq t^k \cdot k!$. This concludes the proof of Theorem \ref{Thm:Derandomization_2_111}.
\end{proofof}

\section{Dimension reduction}
\label{section:dim_reduction_2}

A set $\mathcal{G}$ of functions $g : [n] \rightarrow [\ell]$ is called a \emph{dimension-reduction set} with parameters $(n,\ell, d)$  if it satisfies the following two conditions:

\vspace*{2mm}
\noindent
(1) the number of functions that have the same value on any element of the domain is bounded:\\
$\forall_{i,j \in [n], i \neq j} : |\{g \in \mathcal{G} : g(i) = g(j) \}| \le d;$

\vspace*{2mm}
\noindent
(2) for each function, the elements of the domain are almost uniformly partitioned into the elements of the image:
$\forall_{i \in [\ell] ,g \in \mathcal{G}} : \frac{n}{\ell} \le |g^{-1}(i)| \le \frac{n}{\ell} + o(\ell).$

\vspace*{2mm}
The dimension-reduction set of functions is key in our approach to find low-entropy probability distribution that guarantees a high probability of positive events occurrence. When applied once, it reduces the size of permutations needed to be considered for optimal success probability from $n$-element to $\ell$-element.
The above conditions (1) and (2) are 
to ensure that the found set of $\ell$-element permutations can be reversed into $n$-element permutations without much loss of the occurrence probability of the atomic events (the probabilities of positive events are later reconstructed from a sum of the probabilities of atomic events). Kesselheim, Kleinberg and Niazadeh~\cite{KesselheimKN15} were first to use this type of reduction in context of secretary problems. Our refinement is adding the new condition -- see condition $(2)$ above -- which significantly strengthens the reduction for threshold algorithms. More specifically it allows to preserve the bucket structure of atomic events and has huge consequences for our constructions of low entropy distributions. In particular condition $(2)$ is crucial in proving bounds on the competitive ratio. Before constructing such families, let us describe how a dimension-reduction functions can be used to lift probabilities of atomic events from distribution of lower-dimensional permutations to higher-dimensional permutations.

\subsection{A polynomial time construction of a dimension-reduction set}

We show a general pattern for constructing a set of functions that reduce the dimension of permutations from $n$ to $q < n$ for which we use refined Reed-Solomon codes. 

\ignore{OLD LEMMA 3:
\begin{lemma} \label{lem:Reed_Solomon_Construction}
There exists a set $\mathcal{G}$ of functions $g : [n] \longrightarrow [q]$, for some prime integer $q \geq 2$, such that for any two distinct indices $i, j \in [n]$, $i \not = j$, we have $$|\{g \in \mathcal{G} : g(i) = g(j)\}| \leq d \,\,\,\,\, \mbox{and} \,\,\,\,\, \forall {q' \in [q]} : |g^{-1}(q')| \in \left\{\myfloor{n/q}, \myfloor{n/q} + 1\right\},$$ where $1 \le d < q$ is an integer such that $n \le q^{d + 1}$. Moreover, $|\mathcal{G}| = q$ and set $\mathcal{G}$ can be constructed in deterministic polynomial time in $n, q, d$.
\end{lemma}
}

\begin{lemma} \label{lem:Reed_Solomon_Construction}
Let $n$ and $d$ be positive integers and $q \geq 2$ be any prime integer such that $1 \leq d < q$ and $n \leq q^{d+1}$. There exists a set $\mathcal{G}$ of functions $g : [n] \longrightarrow [q]$, such that for any two distinct indices $i, j \in [n]$, $i \not = j$, we have $$|\{g \in \mathcal{G} : g(i) = g(j)\}| \leq d \,\,\,\,\, \mbox{and} \,\,\,\,\, \forall {q' \in [q]} : |g^{-1}(q')| \in \left\{\myfloor{n/q}, \myfloor{n/q} + 1\right\} \, .$$ Moreover, $|\mathcal{G}| = q$ and set $\mathcal{G}$ can be constructed in deterministic polynomial time in $n, q, d$.
\end{lemma}
\begin{proof} Let us take any finite field $\F$ of size $q \geq 2$. It is known that $q$ must be of the following form: $q = p^r$, where $p$ is any prime number and $r \geq 1$ is any integer; this has been proved by  Galois, see \cite[Chapter 19]{Stewart_Book}. 
 We will do our construction assuming that $\F = \F_q$ is the Galois field, where $q$ is a prime number. 
 
 Let us take the prime $q$ and the integer $d$ such that $1 \leq d < q$ and $q^{d+1} \geq n$. We want to take here the smallest such prime number and an appropriate smallest $d$ such that $q^{d+1} \geq n$. 
 
 Let us now consider the ring $\F[x]$ of univariate polynomials over the field $\F$ of degree $d$. The number of such polynomials is exactly $|\F[x]| = q^{d+1}$. By the field $\F_q$ we chose, we have that $\F_q = \{0,1,\ldots,q-1\}$. We will now define the following $q^{d+1} \times q$ matrix $M = (M_{i,q'})_{i \in [q^{d+1}], q' \in \{0,1\ldots,q-1\}}$ whose rows correspond to polynomials from $\F[x]$ and columns -- to elements of the field $\F_q$.
 
 Let now $\G \subset \F[x]$ be the set of all polynomials from $\F[x]$ with the free term equal to $0$, that is, all polynomials of the form $\sum_{i=1}^d a_i x^i \in \F[x]$, where all coefficients $a_i \in \F_q$, listed in {\em any} fixed order: $\G =\{g_1(x), g_2(x),\ldots,g_{q^d}(x)\}$.
 To define matrix $M$ we will list all polynomials from $\F[x]$ in the following order $\F[x] = \{f_1(x), f_2(x),\ldots,f_{q^{d+1}}(x)\}$, defined as follows. The first $q$ polynomials $f_1(x), f_2(x),\ldots,f_q(x)$ are $f_i(x) = g_i(x) + i-1$ for $i \in \{1,\ldots,q\}$; note that here $i-1 \in \F_q$. The next $q$ polynomials $f_{q+1}(x), f_{q+2}(x),\ldots,f_{2q}(x)$ are $f_{q+i}(x) = g_{q+i}(x) + i-1$ for $i \in \{1,\ldots,q\}$, and so on. In general, to define polynomials $f_{qj+1}(x), f_{qj+2}(x),\ldots,f_{qj+q}(x)$, we have $f_{qj+i}(x) = g_{qj+i}(x) + i-1$ for $i \in \{1,\ldots,q\}$, for any $j \in \{0,1,\ldots,q^d - 1\}$.
 
 We are now ready to define matrix $M$: $M_{i,q'} = f_i(q')$ for any $i \in [q^{d+1}], q' \in \{0,1\ldots,q-1\}$. From matrix $M$ we define the set of functions $\mathcal{G}$ by taking precisely $n$ first rows of matrix $M$ (recall that $q^{d+1} \geq n$) and letting the columns of this truncated matrix define functions in the set $\mathcal{G}$. More formally, $\mathcal{G} = \{h_{q'} : q' \in \{0,1\ldots,q-1\}\}$, where each function 
 $h_{q'} : [n] \longrightarrow [q]$ for each $q' \in \{0,1\ldots,q-1\}$ is defined as $h_{q'}(i) = f_i(q')$ for $i \in \{1,2,\ldots,n\}$.
 
 We will now prove that $|h_{q'}^{-1}(q'')| \in \left\{\myfloor{\frac{n}{q}}, \myfloor{\frac{n}{q}} + 1\right\}$ for each function $h_{q'} \in \mathcal{G}$ and for each $q'' \in \{0,1\ldots,q-1\}\}$. Let us focus on column $q'$ of matrix $M$. Intuitively the property that we want to prove follows from the fact that when this column is partitioned into $q^{d+1} / q$ ``blocks" of $q$ consecutive elements, each such block is a permutation of the set $\{0,1\ldots,q-1\}$ of elements from the field $\F_q$. More formally, the $j$th such ``block" for $j \in \{0,1,\ldots,q^d - 1\}$ contains the elements $f_{qj+i}(q')$ for all $i \in \{1,\ldots,q\}$. But by our construction we have that $f_{qj+i}(q') = g_{qj+i}(q') + i-1$ for $i \in \{1,\ldots,q\}$. Here, $g_{qj+i}(q') \in \F_q$ is a fixed element from the Galois field $\F_q$ and elements $f_{qj+i}(q')$ for $i \in \{1,\ldots,q\}$ of the ``block" are obtained by adding all other elements $i-1$ from the field $\F_q$ to $g_{qj+i}(q') \in \F_q$. This, by properties of the field $\F_q$ imply that $f_{qj+i}(q')$ for $i \in \{1,\ldots,q\}$ are a permutation of the set $\{0,1,\ldots,q - 1\}$.\\
 
 \noindent
 {\em Claim.} For any given $j \in \F_q = \{0,1,\ldots,q-1\}$ the values $j+i$, for $i \in \{0,1,\ldots,q-1\}$, where the addition is in the field $\F_q$ modulo $q$, are a permutation of the set $\{0,1,\ldots,q-1\}$, that is, $\{j+i : i \in \{0,1,\ldots,q-1\}\} = \{0,1,\ldots,q-1\}$. \\
 
 \noindent
\begin{proof}
In this proof we assume that addition and substraction are in the field $\F_{q}$. 
The multiset $\{ j + i : i \in \{0,1,\ldots,q-1\} \} \subseteq \F_{q}$ consists of $q$ values, thus it suffices to show that all values from the multiset are distinct. Assume contrary the there exists two different elements $i$, $i' \in \F_{q}$ such that $j + i = j + i'$. It follows that $i' - i = 0$. This cannot be true since $|i'|, |i| < q$ and $i'$ and $i$ are different.    
\end{proof}
 The property that $|h_{q'}^{-1}(q'')| \in \left\{\myfloor{\frac{n}{q}}, \myfloor{\frac{n}{q}} + 1\right\}$ now follows from the fact that in the definition of the function $h_{q'}$ all the initial ``blocks" $\{ f_{qj+i}(q') : i \in \{1,\ldots,q\}\}$ for $j \in \{0,1,\ldots,\myfloor{\frac{n}{q}} - 1\}$ are fully used, and the last ``block" $\{ f_{qj+i}(q') : i \in \{1,\ldots,q\}\}$ for $j = \myfloor{\frac{n}{q}}\}$ is only partially used. 
 
 Finally, we will prove now that $|\{g \in \mathcal{G} : g(i) = g(j)\}| \leq d$. This simply follows form the fact that for any two polynomials $g, h \in \F[x]$, they can assume the same values on at most $d$, their degree, number of elements from the field $\F_q = \{0,1,\ldots,q - 1\}$. This last property is true because the polynomial $g(x)-h(x)$ has degree $d$ and therefore it has at most $d$ zeros in the field $\F[x]$.
 
 Let us finally observe that the total number of polynomials, $q^{d+1}$, in the field $\F[x]$ can be exponential in $n$. However, this construction can easily be implemented in polynomial time in $n,q,d$, because we only need the initial $n$ of these polynomials. Thus we can simply disregard the remaining $q^{d+1} - n$ polynomials. This completes the proof of the lemma.
\end{proof}

We can obtain the corollary below by choosing the prime $q$ in Lemma~\ref{lem:Reed_Solomon_Construction} as the largest prime such that $q < 2 \log(n)$. Then, by the properties of primes, we have that $\log(n) \leq q < 2 \log(n)$, and this implies that $q^{d+1} > n$ for $d = \Theta(q)$. This gives that $q$ is at least $\Omega(\log(n))$ as required.

\begin{corollary} \label{cor:dim-red-1}
Observe that setting $q \in \Omega(\log{n})$, $d \in \Theta(q)$ in Lemma~\ref{lem:Reed_Solomon_Construction} results in a dimension-reduction set of functions $\mathcal{G}$ with parameters $(n,q, \sqrt{q})$. Moreover, set $\mathcal{G}$ has size $q$ and as long as $q \in O(n)$, it can be computed in polynomial time in $n$.
\end{corollary}

\section{Applications of abstract derandomization and dimension reduction }\label{sec:Applications}
In this section, we show how to apply our generic framework of computing in polynomial-time distributions that serve well threshold algorithms. We then show two applications of this construction and obtain (free-order) algorithms for the multiple-choice secretary problem and the classic secretary problem, with almost-tight competitive ratios.
The construction starts from choosing a parameter $\ell$, usually exponentially smaller than the number of elements $n$, and defining a family of atomic events with respect to a specific bucketing on the set of $\ell$-element permutations. The goal of the bucketing is to capture the algorithm's behavior with respect to different checkpoints and its choice depends on the choice of algorithm. Bucketing can also help defining the algorithm's thresholds. Then we show, that algorithm's success can be described solely by events from the atomic family. Despite the small value of $\ell$, the family of atomic events is usually too rich to support the construction of a low-entropy distribution that preserves the probabilities of atomic events. Thus our idea is to group atomic events into \textit{positive} events, that have larger probabilities. Then, we use the derandomization framework introduced in Section~\ref{sec:abstract_derand} to deterministically construct such low-entropy distribution that supports probabilities of positive events. Doing so on $\ell$-element permutations, makes it computable in polynomial time in $n$, if $\ell$ is sufficiently small. To lift the constructed distribution to full-size $n$ permutations, we use as a black box the dimension-reduction technique introduced in Section~\ref{section:dim_reduction_2}. This all together makes a concise construction of the proper $n$-element permutations distribution, which the algorithm uses to decide in which order to process the elements.

\subsection{Multiple-choice free order secretary problem}\label{sec:application_k_secretary}

As described above, we will start with describing the algorithm and the construction of a smaller $\ell$-element permutations space, for $k < \ell < n$. Then we will apply the lifting technique that uses a dimension-reduction set to obtain a distribution on $n$-element permutations. We use here a modified version of an algorithm with multiple 
thresholds from the survey by Gupta and Singla \cite{GuptaSingla}:

\begin{algorithm}[ht!]
\SetAlgoVlined

\DontPrintSemicolon
 \KwIn{Integers $\ell \geq 2$, $k \leq \ell$, sequence of $\ell$ items each with an adversarial value, and $\pi$ s.t.~$\pi \sim \Pi_{\ell}$.}
 \KwOut{Selected $k$ items from the input sequence.}
 
   Set $\delta := \sqrt{\frac{\log k }{k}}$.\;
   
   Consider the 
   $\ell$ items in the order given by the random permutation $\pi$.\;
   
   Denote $\ell_j := 2^j \delta \ell$ and ignore the first $\ell_0 = \delta \ell$ items.\;
      
   \For{$j \in [0, \log 1/\delta)$, phase $j$ runs on arrivals in window $W_j := (\ell_j,\ell_{j+1}]$}{
      Let $k_j := (k/\ell)\ell_j$ and let $\varepsilon_j := \sqrt{3 \delta/2^j}$.\; 
      Set threshold $\tau_j$ to be the $(1-\varepsilon_j)k_j$th-largest value among the 
      first $\ell_j$ items.\;
      Choose any item in window $W_j$ with value above $\tau_j$ (until budget $k$ is exhausted).\;
   }
 \caption{Multiple-choice secretary algorithm with adaptive 
thresholds and permutation distribution $\mathcal{D} \in \{\Pi_{\ell}, \mathcal{L}_{\ell}\}$.}
 \label{algo:k_secr_algo_1}
\end{algorithm}

\subsubsection{Defining positive events and probabilistic analysis} Here 
we show a lower bound on the measure of each positive event in the space $\Omega_{\ell}$. Let $\hat{S} = \{j_1,\ldots, j_{k}\}$, called a $k$-tuple, be an ordered subset $\{j_1,\ldots, j_{k}\} \subseteq [\ell]$ of $k$ indices. $\hat{S}$ models the positions in the adversarial permutation of the $k$ largest adversarial values $v(1),v(2), \ldots,v(k)$. Let $\mathcal{K}$ be the set of all such $k$-tuples.

Suppose that a random permutation $\pi$ is chosen in the probabilistic space $\Omega$.

We first define some auxiliary events. Let $H_j$ ($L_j$, resp.) be the event that $\tau_j$ is not too low, i.e., is high enough, (not too high, resp.), for $j \in [0, \log 1/\delta)$. More precisely, we define
$$
  H_j = \{\tau_j \geq \min \{v(i) : i=1,...,k\} = v(k)\}, 
$$ and 
$\neg H_j = \{\mbox{less than } (1-\varepsilon_j)k_j \mbox{ items from } \hat{S} \mbox{  fall in the first } \ell_j \mbox{ items in } \pi\}$. We also define 
$$
  L_j = \{\tau_j \leq v((1-2 \varepsilon_j)k)\},
$$ and 
$$
\neg L_j = \{\mbox{more than } (1-\varepsilon_j)k_j \mbox{ items with values } v(1),v(2),\ldots, v((1-2 \varepsilon_j)k)\mbox{ fall in first } \ell_j \mbox{ items in } \pi\}.
$$ Event $C_i$ means that item $j_i$ with value $v(i)$, for $i \in \{1,2,\ldots, (1-2\varepsilon_0)k\}$, will be chosen by the above algorithm. Similarly, for any $j \in \{0,1,\ldots, \log(1/\delta) - 1\}$ and any $i \in \{(1-2\varepsilon_j)k,\ldots, (1-2\varepsilon_{j+1})k\}$, event $C_i$ means that item $j_i$ with value $v(i)$ is chosen by the algorithm. We note that
$$
  C_i = \{\mbox{item } j_i \mbox{ with value } v(i) \mbox{ arrives after position } \ell_0 \mbox{ in } \pi\}, \,\, i \in \{1,2,\ldots, (1-2\varepsilon_0)k\}, 
$$
$$
  C_i = \{\mbox{item } j_i \mbox{ with value } v(i) \mbox{ arrives after position } \ell_{j+1} \mbox{ in } \pi\}, \,\, i \in \{(1-2\varepsilon_j)k,\ldots, (1-2\varepsilon_{j+1})k\} \, ,
$$ and for any $j \in [0,\log 1/\delta)$.

We can now define a {\em positive event} corresponding to the $k$-tuple $\hat{S}$ as the event
$$
  P_{\hat{S},i} = \left(\bigcap_{j \in [0,\log 1/\delta]} \left(L_j \cap H_j\right)\right) \cap C_i \, ,
$$ for any $i \in \{1,2,\ldots, (1-2\varepsilon_0)k\}$, and for any $j \in [0,\log 1/\delta)$ and any 
$i \in \{(1-2\varepsilon_j)k,\ldots, (1-2\varepsilon_{j+1})k\}$. Note that here $i \in \{1,2,\ldots,(1-\delta)k\}$.

We will now show a lower bound on the probability of a positive event in the probabilistic space $\Omega$. Towards this aim we will follow the analysis from the survey by Gupta and Singla \cite{GuptaSingla}. In this analysis they apply Chernoff bounds to a family of partly correlated random variables. Chernoff bounds are in principle applicable to a family of mutually independent random variables, but they show some alternative ways how one might avoid this issue, leaving the details of the argument to the reader. We suggest another way based on {\em negative association} of these random variables, and present a self-contained proof.

We will need the following technical lemma, where we exploit the fact that the indicator random variables which indicate if indices fall in an interval in a random permutation are {\em negatively associated}, see, e.g., \cite{D_Wajc_2017}.

\begin{lemma}\label{lemma:Chernoff-per}
Let $X$ denote a random variable that counts the number of values from the set $\{1,2, \ldots, a\}$, for an integer number $ 1 \le a << n$, that are on positions smaller than the checkpoint $m$ in a random permutation $\sigma \sim \Pi_n$. Define $\mu = \E(X) = a\frac{m}{n}$. 
Then, we obtain that
$$
  \Prob(X \geq (1+\eta)\mu) \leq \exp(-\eta^2 \mu/3), \,\, \mbox{for any} \,\, \eta > 0\, ,
$$ and
$$
  \Prob(X \leq (1-\eta)\mu) \leq \exp(-\eta^2 \mu/2), \,\, \mbox{for any} \,\, \eta \in (0,1) \, .
$$
\end{lemma}

\begin{proof} 
For a number $i$ in the set $\{1,2, \ldots, a\}$ consider an indicator random variable $X_{i}$ equal to $1$ if the position of the number $i$ is in the first $m$ positions of a random permutation $\sigma$, and equal to $0$ otherwise. We have that $X = \sum_{i=1}^{a} X_{i}$. Using standard techniques, for instance Lemma $8$ and Lemma $9ii)$ from \cite{D_Wajc_2017}, we obtain that random variables $X_{1}, \ldots, X_{n}$ are negatively associated (NA) and we can apply the Chernoff concentration bound to their mean. Observe here, that $\E(X) = \E(\sum_{i=1}^{a} X_{i}) = \mu$. Therefore, by Theorem~5 in \cite{D_Wajc_2017}, we have that
$$
  \Prob(X \geq (1+\eta)\mu) \leq \left(\frac{\exp(\eta)}{(1+\eta)^{(1+\eta)}}\right)^{\mu} \,\, \mbox{for any} \,\, \eta > 0\, ,
$$ and
$$
  \Prob(X \leq (1-\eta)\mu) \leq \left(\frac{\exp(-\eta)}{(1-\eta)^{(1-\eta)}}\right)^{\mu} \,\, \mbox{for any} \,\, \eta \in (0,1) \, .
$$
By a well known bound, shown for instance in \cite[page 5]{M_Goemans_2015}, we have that $\left(\frac{\exp(\eta)}{(1+\eta)^{(1+\eta)}}\right)^{\mu} \leq \exp(\frac{-\eta^2 \mu}{2+\eta}) < \exp(-\eta^2 \mu/3)$, where the last inequality follows by $\eta < 1$. Similarly, it is known that $\left(\frac{\exp(-\eta)}{(1-\eta)^{(1-\eta)}}\right)^{\mu} \leq \exp(-\eta^2 \mu/2)$ \cite{HagerupR90}, which together with the above finishes the proof.
\end{proof}


\begin{theorem}\label{thm:gupta-singla-analysis}
 The order-adaptive algorithm in Algorithm \ref{algo:k_secr_algo_1} for the multiple-choice secretary problem has an expected value of at least $\left(1-\sqrt{\frac{\log k}{k}}\right) \cdot v^*$, where $v^*= v(1) + v(2) + \ldots + v(k)$.
\end{theorem}

\begin{proof} 
 We will follow
the proof given in \cite{GuptaSingla} but use our Lemma \ref{lemma:Chernoff-per} to justify the application of the Chernoff bound there.
 
    We first prove that the thresholds 
$\tau_j$ are not too low, so that we will never run out of budget $k$ in the algorithm. Formally, we will prove that 
 $$
    \Prob[H_j] \geq 1 - 1/poly(k) \, .
 $$ We will show it by proving
 $$
    \Prob[\neg H_j] \leq 1/poly(k) \, ,
 $$ by applying Lemma \ref{lemma:Chernoff-per}. Recalling that
 $$
 \neg H_j = \{\mbox{less than } (1-\varepsilon_j)k_j \mbox{ items from } \hat{S} \mbox{  fall in the first } \ell_j \mbox{ items in } \pi\} \, ,
 $$ we see that we can apply this lemma with $a=k$, $X_i = 1$ if number $i \in \{1,2,\ldots,k\}$ falls in the first $m=\ell_j$ positions of the random permutation $\pi$; note that set $\{1,2,\ldots,k\}$ models the indices $\hat{S} = \{j_1,\ldots, j_{k}\} \subseteq [n]$ of the $k$ largest adversarial values $\{v(1),v(2),\ldots,v(k)\}$ in the permutation $\pi$. Then $X = \sum_{i=1}^a X_i$ and $\Exp[X] = km/\ell = k_j$, $\eta = \varepsilon_j$ and Lemma \ref{lemma:Chernoff-per} implies
 $$
   \Prob[\neg H_j] \leq \exp(-\varepsilon_j^2 k_j/2) \leq (1/k)^{3/2} \, .
 $$ We will prove now that thresholds $\tau_j$ are not too high. More precisely we will prove that 
 $$
    \Prob[L_j] \geq 1 - 1/poly(k) \, .
 $$ We will show it by proving
 $$
    \Prob[\neg L_j] \leq 1/poly(k) \, ,
 $$ by again applying Lemma \ref{lemma:Chernoff-per}. Recalling that
$$
\neg L_j = \{\mbox{more than } (1-\varepsilon_j)k_j \mbox{ elements with values } v(1),v(2),\ldots, v((1-2 \varepsilon_j)k)
$$
$$
\mbox{ fall in the first } \ell_j \mbox{ items in } \pi\} \, ,
$$ we see that we can apply Lemma \ref{lemma:Chernoff-per} with $a=(1-2 \varepsilon_j)k$, $X_i = 1$ if number $i \in \{1,2,\ldots,(1-2 \varepsilon_j)k\}$ falls in the first $m=\ell_j$ positions of the random permutation $\pi$; note that set $\{1,2,\ldots,(1-2 \varepsilon_j)k\}$ models the indices $\{j_1,\ldots, j_{(1-2 \varepsilon_j)k}\} \subseteq \hat{S}$ of the $(1-2 \varepsilon_j)k$ largest adversarial values $\{v(1),v(2),\ldots,v((1-2 \varepsilon_j)k)\}$ in the permutation $\pi$. Then $X = \sum_{i=1}^a X_i$ and $\Exp[X] = (1-2 \varepsilon_j)k \cdot (\ell_j/\ell) = (1-2 \varepsilon_j)k_j$, $\eta = \varepsilon_j/(1 - 2\varepsilon_j)$, and Lemma \ref{lemma:Chernoff-per} implies
 $$
   \Prob[\neg L_j] \leq \exp(-\eta^2 k_j/3) \leq \exp(-\varepsilon_j^2 k_j/3) \leq 1/k \, .
 $$ From the above reasoning we can take $poly(k) = \min\{k, k^{3/2}\} = k$. We now take the union bound, which shows that
 $$
   \Prob\left[\bigcup_{j \in [0,\log(1/\delta)]} \left((\neg L_j) \cup (\neg H_j)\right)\right] \leq 2\log(1/\delta)/poly(k) \, ,
 $$ so we have that 
 $$
   \Prob[B] \geq 1 - 2\log(1/\delta)/poly(k) \geq 1 - 1/poly'(k), \,\,\,\, \mbox{ where } \,\,\, B = \bigcap_{j \in [0,\log(1/\delta)]} \left(L_j \cap H_j\right) \, ,
 $$ where $poly'(k) = k/\log(k)$.
 
 We will next compute the probability of the events $C_i$, recalling that
 $$
  C_i = \{\mbox{item } j_i \mbox{ with value } v(i) \mbox{ arrives after position } \ell_0 \mbox{ in } \pi\}, \,\, i \in \{1,2,\ldots, (1-2\varepsilon_0)k\}, 
$$
$$
  C_i = \{\mbox{item } j_i \mbox{ with value } v(i) \mbox{ arrives after position } \ell_{j+1} \mbox{ in } \pi\}, \,\, i \in \{(1-2\varepsilon_j)k,\ldots, (1-2\varepsilon_{j+1})k\} \, ,
$$ and for any $j \in [0,\log 1/\delta)$.

Let first $i \in \{1,2,\ldots, (1-2\varepsilon_0)k\}$ and then, conditioning on the event $B$, we obtain that
\begin{eqnarray}\label{eqn:positive_prob_1}
  \Prob\left[P_{\hat{S},i}\right] = \Prob[B \cap C_i] = \Prob[B] \cdot \Prob[C_i \,|\, B]
  \geq (1-\log(k)/k) \cdot (1-\delta) \, ,
\end{eqnarray} since by event $B$ no threshold is too high and we 
never run out of budget $k$.

Let now 
$i \in \{(1-2\varepsilon_j)k,\ldots, (1-2\varepsilon_{j+1})k\}$ for some $j \in [0,\log 1/\delta)$. Then by the same conditioning on the event $B$ we have that 
\begin{eqnarray}\label{eqn:positive_prob_2}
  \Prob\left[P_{\hat{S},i}\right] = \Prob[B \cap C_i] = \Prob[B] \cdot \Prob[C_i \,|\, B]
  \geq (1-\log(k)/k) \cdot (1-2^{j+1} \cdot \delta) \, .
\end{eqnarray} By (\ref{eqn:positive_prob_1}) and (\ref{eqn:positive_prob_2}) we now finally obtain that the expected value of the Algorithm \ref{algo:k_secr_algo_1} is at least
$$
  \sum_{i=1}^{(1-2\varepsilon_0)k} v(i)(1-\log(k)/k)(1-\delta)
  + \sum_{j=0}^{\log(1/\delta) - 1} \sum_{i=(1-2\varepsilon_j)k}^{(1-2\varepsilon_{j+1})k}
  v(i)(1-\log(k)/k)(1-2^{j+1}\delta) \, .
$$ This expression is at least 
$$
   (1-\log(k)/k) \cdot \left(v^*(1-\delta) - \frac{v^*}{k} \cdot \left(\sum_{j=0}^{\log(1/\delta) - 1} 2 k \varepsilon_{j+1} 2^{j+1} \delta \right)\right) \, ,
$$ because the negative terms are maximized when the top $k$ items are all equal to $\frac{v^*}{k}$. We can simplify this formula to finally obtain $v^* (1-\log(k)/k)(1-O(\delta))$.
\end{proof}

\subsubsection{Decomposing positive event into atomic events} Here we express each positive event as a union of disjoint atomic events, including an algorithm which provides such decomposition. Additionally, we also specify an algorithm to compute conditional probabilities.

Let us define the following bucketing $\mathcal{B}_1$: $B_1 = \{1,2,\ldots,\ell_0\}$,
$B_j = \{\ell_{j-2} + 1,\ell_{j-2} + 2,\ldots,\ell_{j-1}\}$, for $j=2,3,\ldots,\log(1/\delta) + 1$. We will now define all atomic events that define a given positive event $P_i$ for some $i$. The bucketing has $t = \log(1/\delta) + 1$ buckets, so we will use the mappings $f : [k] \longrightarrow [t]$.

We are given an adversarial sequence $\hat{S} = \{j_1,\ldots,j_k\}$, where $j_i \in [\ell]$ is the position of $i$th highest value $v(i)$ in the adversary order. To express event $\left(\bigcap_{j \in [0,\log 1/\delta]} \left(L_j \cap H_j\right)\right)$ by atomic events let us define a mapping $f : [k] \longrightarrow [t]$ such that 
$$
  (1): \,\, \forall j \in [0,\log 1/\delta] : f \mbox{ maps more than }
  (1-\varepsilon_j)k_j \mbox{ items from } \{j_1,\ldots,j_k\} \mbox{ into the first } j+1 \mbox{ buckets} 
$$
$$
  \mbox{AND } f \mbox{ maps less than }
  (1-\varepsilon_j)k_j \mbox{ items from } \{j_1,\ldots,j_{(1-2\varepsilon_j)k}\} \mbox{ into the first } j+1 \mbox{ buckets} \, .
$$ Now to model event $C_i$ for  
$i \in \{(1-2\varepsilon_j)k,\ldots, (1-2\varepsilon_{j+1})k\}$
and some $j \in \{-1\} \cup [0,\log 1/\delta]$, the mapping $f$ has to additionally fulfil that 
$$
(2): \,\, f(j_i) > j+2, \mbox{i.e., item } j_i \mbox{ arrives after time } \ell_{j+1} \, ,
$$ where we additionally define $\varepsilon_{-1} = (1-1/k)/2$, so that $(1-2\varepsilon_{-1})k = 1$. It is easy to see that such mapping $f$ exists. Namely, suppose that $f$ maps $s \leq (1-\varepsilon_j)k_j$ items from $\{j_1,\ldots,j_{(1-2\varepsilon_j)k}\}$ into the first $j+1$ buckets, that is, until position $\ell_j$ in $\pi$. This means that we have at least $(1-2\varepsilon_j)k-s$ remaining items from $\{j_1,\ldots,j_{(1-2\varepsilon_j)k}\}$ together with $k - (1-2\varepsilon_j)k$ items $\{j_{(1-2\varepsilon_j)k+1},\ldots,j_k\}$ that (with those $s$ items included) together should be at least $(1-\varepsilon_j)k_j = (1-\varepsilon_j)k \ell_j/\ell$ items to be mapped by $f$ in the first $\ell_j$ positions (i.e., first $j+1$ buckets). For this to be possible we must have:
$$
  (1-2\varepsilon_j)k-s
  + k - (1-2\varepsilon_j)k + s \geq (1-\varepsilon_j)k_j \, \, \Leftrightarrow \,\, k \geq (1-\varepsilon_j)k_j \, ,
$$ and the last inequality holds because the largest value of $(1-\varepsilon_j)k_j$ is achieved for $j = \log(1/\delta)$ and it is equal to $k - \sqrt{3 k \log k} < k$, which holds for $k$ being at least some constant. This shows property (1), and property (2) can hold because $k \geq (1-\varepsilon_j)k_j + 1$, when $k$ is at least large enough constant. 

Let $\mathcal{F}_{C_i}$, for $i \in \{1,2,\ldots,(1-\delta)k\}$, be the family of all mappings $f : [k] \longrightarrow [t]$ that fulfil conditions (1) and (2) defined above. Given any mapping $f \in \mathcal{F}_{C_i}$, we define a set of permutations $\Sigma_f$ of the sequence $\hat{S} = (j_1,\ldots,j_k)$ that are consistent with $f$, i.e.,   
$\sigma \in \Sigma_f$ if 
$\forall i, i' \in [k], i \not = i' : f(i) < f(i') \iff \sigma^{-1}(j_{i}) < \sigma^{-1}(j_{i'})$. Now we can express the positive event $P_{\hat{S}, i}$, for any $i \in \{1,2,\ldots,(1-\delta)k\}$, as
$$
  P_{\hat{S}, i} = \bigcup_{f \in \mathcal{F}_{C_i}} \bigcup_{\sigma \in \Sigma_f} A_{\sigma, f} \, .
$$

Finally, we define the family of positive events $\mathcal{P}_{\text{G-S}}$ containing all events that are of interest of the described above multiple-choice secretary algorithm in Algorithm \ref{algo:k_secr_algo_1}:
$$\mathcal{P}_{\text{G-S}} = \{ P_{\hat{S}, i} \}_{\hat{S} \subseteq [\ell], \hspace{1mm} 1 \le i \le (1-\delta)k} \, .$$

\begin{lemma}\label{lem:mult-computation-time}
For any $k$-tuple $\hat{S}$ and any positive event $P_{\hat{S},i}$, $i \in \{1,2,\ldots,(1-\delta)k\}$, for the Algorithm \ref{algo:k_secr_algo_1} for the multiple-choice secretary problem, we can compute the set $Atomic(P_{\hat{S},i})$ of all atomic events defining $P_{\hat{S},i}$ in time $O(t^k \cdot k! \cdot \ell)$, where $t = \log(1/\delta)$.
\end{lemma}

\begin{proof}
We have defined above how to decompose the event $P_{\hat{S},i}$ into a union of disjoint atomic events. It is easy to see that the total number of all atomic events is at most $t^k \cdot k!$. The algorithm simply enumerates all $t^k \cdot k!$ atomic events and for each event checks in time $O(\ell)$ if this event fulfills $P_{\hat{S},i}$.
\end{proof}

\subsubsection{Derandomization of positive events via concentration bounds} 

By employing the derandomization technique from Section~\ref{sec:abstract_derand}, we also claim the existence of a small multi-set of $\ell$-element permutations the uniform distribution over $\mathcal{L}_{\ell}$ preserves the probabilities of positive events $P_{\hat{S},i}$.

\begin{lemma}\label{lem:mult-construction}
There exists an $\ell$-element permutations multi-set $\mathcal{L}_{\ell}$ with entropy $O(\log{\ell})$ of the uniform distribution on $\mathcal{L}_{\ell}$ and such that
$$\Prob_{\pi \sim \mathcal{L}_{\ell}}\left[\pi \in P_{\hat{S},i}\right] \ge \left(1-\frac{1}{\sqrt{k}}\right) \Prob_{\pi \sim \Pi_{\ell}}\left[\pi \in P_{\hat{S},i}\right],
$$
for any positive event $P_{\hat{S}, i} \in \mathcal{P}_{\text{G-S}}$. The multi-set $\mathcal{L}_{\ell}$ can be computed in $O(\ell^{10}\cdot t^{2k} \cdot (k!)^{3})$ time.
\end{lemma}
\begin{proof}
By the analysis of Theorem~\ref{thm:gupta-singla-analysis}, specifically formulas (\ref{eqn:positive_prob_1}) and (\ref{eqn:positive_prob_2}),  we have that for any $P_{\hat{S},i} \in \mathcal{P}_{\text{G-S}}$ it holds
$$\Prob_{\pi \sim \Pi_{\ell}}[\pi \in P_{\hat{S},i}] \ge (1 - \log k/k) \cdot (1 - 2^{j + 1}\cdot \delta'),$$
for some parameters $j, \delta'$. However, when used we always have that $j < \log(1 / \delta') - 1$ which implies $1 - 2^{j + 1}\cdot \delta' > \frac{1}{2}$ and further that
the probability $\Prob_{\pi \sim \Pi_{\ell}}[\pi \in P_{\hat{S},i}]$ is bounded below for some constant $p > \frac{1}{10}$, if $k$ is larger than $5$.
Thus, we can apply Theorem~\ref{Thm:Derandomization_2_111} with $\delta := \frac{1}{\sqrt{k}}$ and obtain a multi-set set $\mathcal{L}_{\ell}$ of $\ell$-element permutations with the following properties:
\\ \noindent \textit{a)} Its size is at most:
$$2\frac{\log (|\mathcal{P}_{\text{G-S}}|)}{\delta^2 p} \le O\left( \log\bigg({\ell\choose k}k!\cdot k\bigg) \cdot k\right) \le O\left(k\log(\ell) \cdot k\right) \le O\left(k^{2}\log(\ell)\right).$$
Thus it follows that the entropy of the uniform distribution over $\mathcal{L}_{\ell}$ is at most $O(\log{k} + \log\log{\ell}) = O(\log{\ell})$ given that $k < \ell$.
\\ \noindent \textit{b)} The set $\mathcal{L}_{\ell}$ is computable in time:
$$O(k^{2}\log(\ell) \cdot \ell^3 \cdot k\log(\ell) \cdot t^{2k} \cdot (k!)^2 \cdot (\ell + k k! + k \log^2(\ell)))$$ 
$$= O(k^{3}\log^2(\ell) \cdot \ell^3 \cdot t^{2k} \cdot (k!)^2 \cdot \ell k k!) = O(\ell^{10}\cdot t^{2k} \cdot (k!)^{3}),$$
where the last equality follows from Lemma~\ref{lem:mult-computation-time}.
\\ \noindent \textit{c)} It holds that:
$$\Prob_{\pi \sim \mathcal{L}_{\ell}}\left[\pi \in P_{\hat{S},i}\right] \ge \left(1-\frac{1}{\sqrt{k}}\right) \Prob_{\pi \sim \Pi_{\ell}}\left[\pi \in P_{\hat{S},i}\right],
$$ thus the lemma follows.
\end{proof}

\subsubsection{Lifting lower-dimension permutations distribution satisfying positive events}

Next, we show how, by applying a dimension-reduction set of functions defined in Section~\ref{section:dim_reduction_2}, one can turn a set $\mathcal{L}_{\ell}$ of $\ell$-element permutations to a set $\mathcal{L}_{n}$ of $n$-element permutations such that the competitive ratio of Gupta-Singla algorithm in Algorithm \ref{algo:k_secr_algo_1} executed on the uniform distribution over $\mathcal{L}_{\ell}$ is carried to the competitive ratio of the Gupta-Singla algorithm executed on the uniform distribution over $\mathcal{L}_{n}$. 

\begin{lemma}\label{lem:mult-lifting}
Let $ALG_{\text{G-S}}(\pi)$ denote the output of Algorithm \ref{algo:k_secr_algo_1} on the permutation $\pi$. Assuming that $\ell^{2} < \frac{n}{\ell}$, one can compute a multi-set of $n$-element permutations $\mathcal{L}_{n}$ such that
$$\E_{\pi \sim \mathcal{L}_{n}}(ALG_{\text{G-S}}(\pi)) > \bigg(1 - \frac{k^2}{\sqrt{\ell}}\bigg)\bigg(1 - \frac{1}{\sqrt{k}}\bigg)\bigg(1 - \sqrt{\frac{\log{k}}{k}}\bigg) \cdot v^{\ast},$$
where $v^*= v(1) + v(2) + \ldots + v(k)$ is the sum of $k$ largest adversarial elements.
Given $\mathcal{L}_{\ell}$ we can construct $\mathcal{L}_{n}$ in $O(n \cdot k^2\log(\ell))$ time and the entropy of the uniform distribution over $\mathcal{L}_{n}$ is $O(\log{\ell} + \log{|\mathcal{L}_{\ell}|})$.
\end{lemma}
\begin{proof}
Consider a dimension-reduction set of functions $\mathcal{G}$ given by Corollary~\ref{cor:dim-red-1} with parameters $(n, \ell, \sqrt{\ell})$. Note, that the size of set $\mathcal{G}$ is $O(poly\text{ } \ell)$. Recall, that set $\mathcal{L}_{\ell}$ is given in Lemma~\ref{lem:mult-construction}.
For a given function $g \in \mathcal{G}$ and a permutation $\pi \in \mathcal{L}_{\ell}$ we denote by $\pi \circ g : [n] \rightarrow [n]$ any permutation $\sigma$ over set $[n]$ satisfying the following:
$\forall_{i,j \in [n], i \neq j}$ if $\pi^{-1}(g(i)) < \pi^{-1}(g(j))$ then $\sigma^{-1}(i) < \sigma^{-1}(j)$.
The aforementioned formal definition has the following natural explanation. The function $g \in \mathcal{G}, g : [n] \rightarrow [\ell]$ may be interpreted as an assignment of each element from set $[n]$ to one of $\ell$ blocks. Next, permutation $\pi \in \mathcal{L}_{\ell}$ determines the order of these blocks. The final permutation $\sigma$ is obtained by listing the elements from the blocks in the order given by $\pi$. The order of elements inside the blocks is irrelevant. 
The set $\mathcal{L}_{n}$ of $n$-element permutations is defined as $\mathcal{L}_{n} = \{ \pi \circ g : \pi \in \mathcal{L}_{\ell}, g \in \mathcal{G}\}$, and its size is $|\mathcal{L}_{\ell}| \cdot |\mathcal{G}| = O(poly \text{ }(\ell) \cdot |\mathcal{L}_{\ell}|)$. It is easy to observe that $\mathcal{L}_{n}$ can be computed in $O(poly\text{ }(\ell) \cdot |\mathcal{L}_{\ell}|)$ time and the entropy of the uniform distribution over this set is $O(\log{\ell} + \log{|\mathcal{L}_{\ell}|})$.

In the remaining part, we show that uniform distribution over the set $\mathcal{L}_{n}$ guarantees the proper competitive ratio. Consider any $k$-tuple $\hat{S} = (j_{1}, \ldots, j_k) \subseteq [n]$ denoting the positions of the $k$ largest adversarial elements in an $n$-element adversarial permutation. If the multiple-choice secretary algorithm is executed on the permutation $\sigma$, one can associate to this random experiment the following interpretation: first we draw u.a.r a function $g$ from $\mathcal{G}$ and then draw u.a.r a permutation $\pi$ from $\mathcal{L}_{\ell}$. 
Observe, that for a random function $g \in \mathcal{G}$ the probability that a pair of fixed indices $j_{i},j_{i'}$ is distributed to the same block is at most $\frac{d}{\ell} = \frac{\sqrt{\ell}}{\ell}$, by Property $(1)$ of the dimension-reduction set $\mathcal{G}$. Then, by the union bound argument, we conclude that the probability that all elements of the $k$-tuple $\hat{S}$ are assigned to different blocks is at least $1 - \frac{k^2}{\sqrt{\ell}}$. Conditioned on this event, the image of $\hat{S}$ under the function $g$ is a $k$-tuple of elements from $[\ell]$ denoted $\hat{S}' = (j'_{1}, \ldots, j'_{k})$. Therefore, by Lemma~\ref{lem:mult-construction}, the multiple-choice secretary algorithm when executed on a random permutation $\pi$ from $\mathcal{L}_{\ell}$ chooses the $i$-th largest adversarial element with probability at least 
$$\left(1-\frac{1}{\sqrt{k}}\right) \Prob_{\pi \sim \Pi_{\ell}}\left[\pi \in P_{\hat{S},i}\right] \, .
$$ To argue that this competitive ratio carries to the $n$-element permutation $\pi \circ g$, we observe first that the size of each block is $[\frac{n}{\ell}, \frac{n}{\ell} + o(\ell)]$, by Property $(2)$ of the dimension-reduction set $\mathcal{G}$. Also, we required that $\ell^{2} < \frac{n}{\ell}$. 
Therefore, any bucket $B_{j} = \{\ell_{j-2} + 1, \ldots, \ell_{j-1} -1 \}$\footnote{Note, that we removed here last element from the $j$-th bucket. This allows to carry these buckets to proper positions in an $n$-element permutation, and since we removed at most $\frac{t}{\ell} < \frac{1}{\sqrt{\ell}}$ fraction of elements (buckets) from the consideration it does not affect the final probability by more than $\frac{1}{\sqrt{\ell}}$ additive error.}, for $1 \le j \le t = \log(1/ \delta) + 1$ from the bucketing $\mathcal{B}_{1}$ translates to a continuous interval of positions 
$$\left\{(\ell_{j-2} + 1)\frac{n}{\ell}, \ldots, (\ell_{j-1} - 1)\frac{n}{\ell} \right\} = \left\{2^{j-2}\delta \cdot n + \frac{n}{\ell}, \ldots, 2^{j-1}\delta \cdot n - \frac{n}{\ell} + 2^{j-1}\delta\ell\cdot o(\ell) + o(\ell) \right\}$$ 
$$\subseteq  \{2^{j-2}\delta \cdot n + 1, \ldots, 2^{j-1}\delta \cdot n \} \subseteq \{n_{j-2} + 1, \ldots, n_{j-1}\}$$ of
the $n$-element permutation $\pi \circ g$, where $n_{j} := 2^{j}\delta n$. Here, we used the fact that $\ell^{2} < \frac{n}{\ell}$, which gives the first sets' inclusion. 
The last set is exactly the set of positions that fall between $j$-th and $(j+1)$-th checkpoints when the Gupta-Singla algorithm is executed on the full $n$-element permutation $\pi \circ g$. Since this observation holds for all adversarial elements from the $k$ largest elements $\hat{S}'$, it follows then that the positions of these elements with respect to checkpoints $(n_{j})_{j \in [\log(1/ \delta) + 1]}$ in the $n$-element permutation $\pi \circ g$ are the same as positions of them in the $\ell$-element permutation $\pi$ with respect to    
checkpoints
$(\ell_{j})_{j \in [\log(1/ \delta) + 1]}$. Consequently, whenever the multiple-choice secretary algorithm chooses $i$-th adversarial element to the returned sum when executed on $\pi$, the same algorithm also chooses the $i$-th adversarial element when executed on the permutation $\pi \circ g$, conditioned on the fact that $g$ is injective. This conditioning holds with probability $\big(1 - \frac{k^2}{\sqrt{\ell}}\big)$, thus the lemma follows from the above discussion and Lemma~\ref{lem:mult-construction} combined with Theorem~\ref{thm:gupta-singla-analysis}. 
\end{proof}

In the main theorem we use the above lemmas with the appropriate choice of parameters $k$, and $\ell$. 
\begin{theorem}\label{thm:k_secretary_main}
For any $k < \log{n}/\log \log n$, there exists a multi-set of $n$-element permutations $\mathcal{L}_{n}$ such that Algorithm \ref{algo:k_secr_algo_1} achieves
$$1 - 4\sqrt{\frac{\log{k}}{k}} $$
expected competitive ratio for the free order multiple-choice secretary problem, when the adversarial elements are presented in the order chosen uniformly from $\mathcal{L}_{n}$. The set is computable in time $O(\text{poly }(n))$ and the uniform distribution over $\mathcal{L}_{n}$ has the optimal $O(\log{k}) = O(\log\log{n})$ entropy.
\end{theorem}
\begin{proof}
Set $k$, $k < \log{n} / \log \log n$, $\ell := k^{8}$, then the theorem follows by applying directly Lemma~\ref{lem:mult-construction} and~\ref{lem:mult-lifting} to this choice of parameters. The total running time follows because if $k < \log n / \log \log n$, the running time of the construction by Lemma~\ref{lem:mult-construction} is $O(\ell^{10}\cdot t^{2k} \cdot (k!)^{3}) = poly(n)$, because the dominant time here is $k! \leq k^k = poly(n)$.
\end{proof}

\subsection{Classical free order secretary problem}\label{sec:application_1_secretary}


Consider the classic secretary algorithm with checkpoint on the position $\frac{n}{e}$\footnote{For simplicity we omit rounding notation $\lfloor \frac{n}{e} \rfloor$.} for the free order $1$-secretary problem. The algorithm looks at all elements arriving before the checkpoint, finds the maximum one and selects the first element arriving after the checkpoint that is larger than the found maximum. It is well known, that if the algorithm uses a uniform random order to process the elements, then the probability of picking the maximum element is at least $\frac{1}{e} + \Omega(\frac{1}{n})$ (see Section~\ref{section:lb_classic_secr}, Proposition~\ref{Thm:optimum_expansion}.). In this section we show how by applying the techniques crafted in previous sections together with a new analysis of the performance guarantee of the classic algorithm, one can obtain an optimal-entropy permutations distribution that achieves the competitive ratio $\frac{1}{e} - O(\frac{\log\log^2{n}}{\log^{1/2}{n}})$, which in this case is the success probability. \\

\subsubsection{Defining positive events and decomposing them into atomic events}

Following the approach introduced in the beginning of the section, we start by constructing a permutations distribution of much smaller dimension $\ell = O(\log{n})$ such that the classic secretary algorithm executed on this distribution achieves $\frac{1}{e} - O(\frac{\log\log^2{n}}{\log^{1/2}{n}})$ competitive ratio. 

Let us fix a parameter $\ell < n$ (the precise relation between $\ell$ and $n$ will be defined later), a bucketing $\mathcal{B}_{1} = \{1, \ldots, \frac{\ell}{e} - 1\}, \{\frac{\ell}{e} + 1, \ldots, n\}$ and some integer parameter $k < \ell$ to be defined later. Let $\mathcal{A}_{k, \mathcal{B}_{1}}$ be the set of atomic events defined on the uniform probabilistic space $\Omega_{\ell}$ of $\ell$-element permutations with respect to the parameter $k$ and bucketing $\mathcal{B}_{1}$. Next, we define the family of positive events, each consisting of a disjoint atomic events from the family $\mathcal{A}_{k, \mathcal{B}_{1}}$ that captures successful events of the classic secretary algorithm. 

Consider a $k$-tuple $\sigma = (\sigma(1), \ldots, \sigma(k)) \subseteq [\ell]$, interpreted as the positions of the $k$ largest elements in the adversarial order. Let $T_{i}$, for $2 \le i \le k$ be a set of these $k$-element permutations of the $k$-tuple $\sigma$ such that $\sigma(i)$ is presented on the first position and for any other $2 \le i' \le i-1$ it holds that $\sigma(i')$ appears after element $\sigma(i)$. Additionally, let $F_{i} \subseteq \{1, 2\}^{k}$, for $2 \le i \le k$, be a set of all these non-decreasing functions $f_{i}$ that satisfy $f_i(1) = 1$ and $f_i(i) = 2$.
Then we define a positive event 
$$P_{i} = \bigcup_{\pi \in T_{i}, f \in F_{i}} A_{\pi, f_{i}},
$$ where recall that $A_{\pi, f_{i}}$ is the atomic event in the space of all $n$-element permutations. We also define
$$P_{\sigma} = \bigcup_{2 \le i \le k} P_{i}.
$$ To explain the motivation behind the above construction, note that $P_{i}$ is the set of these events, or equivalently $\ell$-element permutations, such that the $i$-th greatest adversarial value appears in the first bucket, the $1$-st greatest value, a.k.a.~greatest, appears in the second bucket and all values in between these two appear after the $1$-st greatest value in the second bucket. Clearly, if this event happens then the classic secretary algorithm picks the largest adversarial value. Under this reasoning, the set of events $P_{\sigma}$ captures all possibilities of picking the largest adversarial value when algorithm is committed to tracking only $k$ largest adversarial values and the indices of the $k$ largest adversarial values are given by $\sigma$. Observe also, that any two atomic events included in $P_{\sigma}$ are disjoint because they describe different arrangements of the same set of elements $\sigma$. Note also, that for an atomic event $A_{\pi, f_{i}}$ determining whether it belongs to $P_{i}, 2 \le i \le k$ is computable in time depending on $\ell$, and $k$ only. 

\begin{lemma}\label{lem:one-positive-comp-time}
Given any positive event $P_i$, $i \in \{2,\ldots, k\}$ we can compute the set $Atomic(P_i)$ of all atomic events defining $P_i$ in time $O(2^k \cdot k! \cdot \ell)$.
\end{lemma}
\begin{proof}
The brute-force algorithm iterating over all possible permutations of $\sigma = (\sigma(1), \ldots, \sigma(k))$ and all functions $f_{i} : [k] \rightarrow \{1, 2\}$ and checking whether the atomic event associated with the choice of the permutation and the function $f_{i}$ satisfies the mentioned before condition and works in time $O(2^k \cdot k! \cdot \ell)$.
\end{proof}

\subsubsection{Probabilistic analysis}

Next step is to show that considering only $k$ largest adversarial values is in fact enough to obtain a good competitive ratio.

\begin{lemma}\label{lem:one-sec-uniform-lower-bound}
If $\pi$ is chosen uniformly at random from  the set $\Pi_{\ell}$ of all $\ell$-element permutation, then:
$$\Prob_{\pi\sim \Pi_{\ell}}\left[P_{\sigma}\right] \ge \frac{1}{e} + \frac{1}{\ell} - \frac{1}{e\cdot k} \cdot \left( 1 - \frac{1}{e} \right)^{k}.$$
\end{lemma}
\begin{proof}
Let $\hat{\sigma} := (\hat{\sigma}(1), \hat{\sigma}(2), \ldots, \hat{\sigma}(k), \hat{\sigma}(k+1), \ldots, \hat{\sigma}(\ell))$ be an arbitrary extension of the $k$-tuple $\sigma$ to a full $\ell$-element permutations. Consider an event $S_{\ell}$ corresponding to the success of the classic secretary algorithm executed on the uniform distribution over the set of $\ell$-element permutations if the adversarial order is $\hat{\sigma}$. By Proposition~\ref{Thm:optimum_expansion} (Part 1), we have that
$$\Prob_{\pi\sim \Pi_{\ell}}\left[S_{\ell}\right] \ge \frac{1}{e} + \frac{1}{\ell}.$$
On the other hand, we can decompose the event $S_{\ell}$ as follows $S_{\ell} = \cup_{2 \le i \le \ell} E_{i}$, where event $E_{i}$ corresponds to the fact that the algorithm selects the largest value given that the largest value observed in the first bucket $B_{1}$ is $i$-th largest in the whole sequence. Precisely,
$$E_{i} := \{ \pi \in \Pi_{\ell} : \pi^{-1}(\hat{\sigma}(i)) \in B_{1}, \pi^{-1}(\hat{\sigma}(1)) \in B_{2} \wedge \forall_{2 \le i' \le i-1} \pi^{-1}(\hat{\sigma}(i')) > \pi^{-1}(\hat{\sigma}(i)) \}.$$
From the Bayes' formula on conditional probability we obtain that
\begin{equation}\label{line:bayes}
\Prob_{\pi\sim \Pi_{\ell}}[E_i] = \frac{(1 / e)\ell}{\ell} \cdot \left( \prod_{j=1}^{i-1} \frac{(\ell - (1 / e)\ell) - (j - 1)}{(\ell - 1) - (j - 1)} \right) \cdot \frac{(i-2)!}{(i-1)!} \, .    
\end{equation}
In the above, the first factor corresponds to the probability that  $\pi^{-1}(\hat{\sigma}(i)) \in B_{1}$, the second factor corresponds to the probability that numbers $\hat{\sigma}(1), \ldots, \hat{\sigma}(j-1)$ are mapped to the second bucket $B_{2}$, while the last factor the probability that, conditioned on the previous events, $\pi^{-1}(\hat{\sigma}(1)) \le \pi^{-1}(\hat{\sigma}(i'))$, for $2 \le i' \le i - 1$.

On the one hand, we have that $\Prob_{\pi\sim \Pi_{\ell}}[P_{\sigma}] = \sum_{2 \le i \le k} \Prob_{\pi\sim \Pi_{\ell}}[E_{i}]$. On the other hand, by the above decomposition (\ref{line:bayes}), we get
\[
\Prob_{\pi\sim \Pi_{\ell}}[S_{\ell}]
= 
\sum\limits_{2 \le i \le \ell - (1/e)\ell} \Prob_{\pi\sim \Pi_{\ell}}[E_{i}] 
\le 
\sum\limits_{2 \le i \le k } \Prob_{\pi\sim \Pi_{\ell}}[E_{i}] + \frac{1}{e} \left( \frac{\ell - (1/e)\ell}{\ell - 1} \right)^{k} \cdot \frac{1}{k} \, ,
\mbox{ and then}
\]
\vspace*{-2ex}
\begin{eqnarray}
\Prob_{\pi\sim \Pi_{\ell}}[P_{\sigma}] = \sum_{i=2}^{k} \Prob_{\pi\sim \Pi_{\ell}}[E_i] \,\, \geq \,\, \Prob_{\pi\sim \Pi_{\ell}}[S_{\ell}] - \frac{1}{e\cdot k} \cdot \left( 1 - \frac{1}{e} \right)^{k}. \label{Eq:Rho_Lower_Bound}
\end{eqnarray}
This proves the lemma.
\end{proof}
Finally, we define the positive family that captures successful events for the classic secretary algorithm executed on $\ell$-element permutations regardless of the positions on which $k$ largest values appear in the adversarial order.  
$$\mathcal{P}_{\ell-\text{classic}} := \{ P_{\sigma} \}_{\sigma = (\sigma(1), \ldots, \sigma(k)) \hspace{1mm} \subseteq \hspace{1mm} [\ell]}.$$

\subsubsection{Derandomization of positive events via concentration bounds} 

Because probabilities $\Prob_{\pi\sim \Pi_{\ell}}[P_{\sigma}]$ are bounded from below by a constant, when drawing a permutation from the uniform distribution. $\pi\sim \Pi_{\ell}$, we can use our derandomization technique introduced in Section~\ref{sec:abstract_derand} and obtain the following.

\begin{lemma}\label{lem:one-construction-small}
There exists a multi-set $\mathcal{L}_{\ell}$ of  $\ell$-element permutations such that the uniform distribution on $\mathcal{L}_{\ell}$ has entropy $O(\log{\ell})$, and
$$\Prob_{\pi \sim \mathcal{L}_{\ell}}(\pi \in P_{\sigma}) \ge \frac{1}{e} - \frac{1}{\ell} - \frac{1}{e\cdot k} \cdot \left( 1 - \frac{1}{e} \right)^{k},$$
for any positive event $P_{\sigma} \in \mathcal{P}_{\ell-\text{classic}}$. The distribution can be computed in $O(\ell^{k+6} \cdot (k!)^5)$ time.
\end{lemma}
\begin{proof}
Consider Theorem~\ref{Thm:Derandomization_2_111} applied to the atomic events $\mathcal{A}_{k, \mathcal{B}_{1}}$, the family of positive events $\mathcal{P}_{\ell-\text{classic}}$ and parameters $\ell$ and $k < \ell$. Lemma~\ref{lem:one-sec-uniform-lower-bound} and Lemma~\ref{lem:one-positive-comp-time} ensure that necessary conditions for Theorem~\ref{Thm:Derandomization_2_111} hold and in consequence there exists a mulit-set $\mathcal{L}_{\ell}$ (identified in the proof with the uniform distribution over the set) such that for every $P_{\sigma} \in \mathcal{P}_{\ell-\text{classic}}$ it holds 
$$\Prob_{\pi \sim \mathcal{L}_{\ell}}(\pi \in P_{\sigma}) \ge (1-\delta) \bigg(\frac{1}{e} + \frac{1}{\ell} - \frac{1}{e\cdot k} \cdot \left( 1 - \frac{1}{e} \right)^{k} \bigg),$$
which for $\delta := \frac{1}{\ell}$ implies the claimed inequality in the lemma statement. The size of $\mathcal{L}_{\ell}$ can be upper bounded as follows:
$$|\mathcal{L}_{\ell}| \le \frac{2\log(|\mathcal{P}_{\ell-\text{classic}}|)}{\delta^{2} p} \le 2\ell p^{-1}\log\bigg( {\ell \choose k} \cdot k!\bigg)$$
$$\le 4\ell (k\cdot \log{\ell} + \log^{2}(k)),$$
where the last inequality holds because Lemma~\ref{lem:one-sec-uniform-lower-bound} ensures that $p > \frac{1}{2}$. It follows then that the entropy of $\mathcal{L}_{\ell}$ is $O(\log \ell)$, since $k < \ell$. Finally, using Theorem \ref{Thm:Derandomization_2_111}, we calculate that the running time for computing the multi-set $\mathcal{L}_{\ell}$ is the following
$$O\left(|\mathcal{L}_{\ell}| \cdot \ell^4 \cdot |\mathcal{P}_{\ell-\text{classic}}| \cdot 2^{2k} 
\cdot (k!)^3 \cdot k\right) = O(\ell^{k+6} \cdot (k!)^5) \, ,
$$
assuming that $k$ and $\ell$ are larger than an absolute large enough constant. This completes the proof of the lemma.
\end{proof}

\subsubsection{Lifting lower-dimension permutations distribution satisfying positive events}

Assume now, that we are given a set $\mathcal{L}_{\ell}$ of $\ell$-element permutations such that the classic secretary algorithm achieves $\frac{1}{e}-\epsilon$ competitive ratio, for an $0 < \epsilon < \frac{1}{e}$, when the adversarial elements are presented in the order given by an uniform permutation from $\mathcal{L}_{\ell}$. In this part, we apply a dimension-reduction set of functions defined in Section~\ref{section:dim_reduction_2}, and show how to lift the set $\mathcal{L}_{\ell}$ to a set of $n$-element permutations $\mathcal{L}_{n}$, for $n > l$. This implies that the classic secretary algorithm executed on a uniform permutation chosen from $\mathcal{L}_{n}$ achieves the competitive ratio of the smaller-dimension distribution. The proof is similar of the analog lemma for the multiple-choice secretary problem.


\begin{lemma}\label{lem:one-lifting}
Denote $ALG_{\text{classic}}(\pi')$ the event that the classic algorithm achieves success on the random permutation $\pi' \sim \mathcal{L}_{n}$. Assuming that $\ell^{2} < \frac{n}{\ell}$, there exists a multi-set of $n$-element permutations $\mathcal{L}_{n}$ such that
$$\Prob_{\pi' \sim \mathcal{L}_{n}}(ALG_{\text{classic}}(\pi')) > \bigg(1 - \frac{k^2}{\sqrt{\ell}}\bigg)\bigg(\frac{1}{e} - \frac{1}{\ell} - \frac{1}{e\cdot k} \cdot \left( 1 - \frac{1}{e} \right)^{k}\bigg).$$
Moreover, given the multi-set $\mathcal{L}_{\ell}$, the set $\mathcal{L}_{n}$ can be constructed in $O(n \cdot |\mathcal{L}_{\ell}|)$ time and the entropy of the uniform distribution on $\mathcal{L}_{n}$ is $O(\log\log{n} + \log{|\mathcal{L}_{\ell}|})$.
\end{lemma}
\begin{proof}
Consider a dimension-reduction set of functions $\mathcal{G}$ given by Corollary~\ref{cor:dim-red-1} with parameters $(n, \ell, \sqrt{\ell})$. Note, that the size of set $\mathcal{G}$ is $O(poly (\ell))$. Recall that set $\mathcal{L}_{\ell}$ is given in Lemma \ref{lem:one-construction-small}.  
For a given function $g \in \mathcal{G}$ and a permutation $\pi \in \mathcal{L}_{\ell}$ we denote by $\pi \circ g : [n] \rightarrow [n]$ any permutation $\sigma$ over set $[n]$ satisfying the following:
$\forall_{i,j \in [n], i \neq j}$ if $\pi^{-1}(g(i)) < \pi^{-1}(g(j))$ then $\sigma^{-1}(i) < \sigma^{-1}(j)$.
The aforementioned formal definition has the following natural explanation. The function $g \in \mathcal{G}, g : [n] \rightarrow [\ell]$ may be interpreted as an assignment of each element from set $[n]$ to one of $\ell$ blocks. Next, permutation $\pi \in \mathcal{L}_{\ell}$ determines the order of these blocks. The final permutation is obtained by listing the elements from the blocks in the order given by $\pi$. The order of elements inside the blocks is irrelevant. 
The set $\mathcal{L}_{n}$ of $n$-element permutations is defined as $\mathcal{L}_{n} = \{ \pi \circ g : \pi \in \mathcal{L}_{\ell}, g \in \mathcal{G}\}$, and its size is $|\mathcal{L}_{\ell}| \cdot |\mathcal{G}| = O(poly \log (\ell) \cdot |\mathcal{L}_{\ell}|)$. It is easy to observe that $\mathcal{L}_{n}$ can be computed in $O(poly (\ell) \cdot |\mathcal{L}_{\ell}|)$ time and the entropy of the uniform distribution over this set is $O(\log{n} + \log{|\mathcal{L}_{\ell}|})$.

In the remaining part, we show that uniform distribution over the set $\mathcal{L}_{n}$ guarantees the proper competitive ratio. Consider any $k$-tuple $\sigma = (\sigma(1), \ldots, \sigma(k)) \subseteq [n]$ denoting the position of $k$ largest adversarial elements in the $n$-element adversarial permutation. If the classic secretary algorithm is executed on  permutation $\sigma$, one can associate to this random experiment the following interpretation: first we draw u.a.r a function $g$ from $\mathcal{G}$ and then draw u.a.r a permutation $\pi$ from $\mathcal{L}_{\ell}$. 
Observe, that for a random function $g \in \mathcal{G}$ the probability that a pair of fixed indices $\sigma(i),\sigma(j)$ is distributed to the same block is at most $\frac{d}{\ell} = \frac{\sqrt{\ell}}{\ell}$, by Property $(1)$ of the dimension-reduction set $\mathcal{G}$. Then, by the union bound argument, we conclude that the probability that all elements of the $k$-tuple $\sigma$ are assigned to different blocks is at least $1 - \frac{k^2}{\sqrt{\ell}}$. Conditioned on this event, the image of $\sigma$ under the function $g$ is a $k$-tuple of elements from $[\ell]$ denoted $\sigma' = (\sigma'(1), \ldots, \sigma'(k))$. Therefore, by Lemma~\ref{lem:one-construction-small}, the classic algorithm (with checkpoint $\ell / e$) when executed on a random permutation $\pi$ from $\mathcal{L}_{\ell}$ picks the largest element with probability at least $$\frac{1}{e} - \frac{1}{\ell} - \frac{1}{e\cdot k} \cdot \left( 1 - \frac{1}{e} \right)^{k} \, .$$
To argue that this competitive ratio carries to $n$-element permutation $\pi \circ g$, we observe first that size of each block is $[\frac{n}{\ell}, \frac{n}{\ell} + o(\ell)]$, by Property $(2)$ of the dimension-reduction set $\mathcal{G}$. Also, we required that $\ell^{2} < \frac{n}{\ell}$. Therefore, the first bucket from the bucketing $\mathcal{B}_{1} = \{1, \ldots,  \frac{\ell}{e} - 1\}, \{\frac{\ell}{e} + 1, \ldots, \ell \}$ associated with the multi-set of permutation $\mathcal{L}_{n}$ translates to first $(\ell / e - 1) \cdot (\frac{n}{\ell} + o(\ell)) = \frac{n}{e} - \frac{n}{\ell} + \frac{\ell o(\ell)}{e} < \frac{n}{e}$ positions when the permutation $\pi \circ g$ is considered. It follows that the positions of elements from $k$-tuple $\sigma$ in the permutation $\pi \circ g$ with respect to the checkpoint $\frac{n}{e}$ are the same as positions of elements from $k$-tuple $\sigma'$ in the permutation $\pi$ with respect to the checkpoint $\frac{\ell}{e}$. Also, their relative order conveys from permutation $\pi$ to permutation $\pi \circ g$. Thus, it follows that whenever the classic algorithm is successful on the permutation $\pi$ it is also successful  on the permutation $\pi \circ g$ conditioned on the fact that $g$ is injective. This conditioning holds with probability $\big(1 - \frac{k^2}{\sqrt{\ell}}\big)$, thus the lemma follows by using Lemma \ref{lem:one-construction-small}.
\end{proof}

The main theorem employs the techniques introduced earlier in this section with the appropriate choice of parameters $k$, and $\ell$. 
\begin{theorem}\label{thm:1_secretary}
There exists a multi-set of $n$-element permutations $\mathcal{L}_{n}$ such that the the wait-and-pick algorithm with checkpoint $\lfloor n/e \rfloor$ achieves
$$\frac{1}{e} - \frac{3\log\log^2{n}}{e\log^{1/2}{n}} $$
success probability for the free order $1$-secretary problem, when the algorithm uses the order chosen uniformly from $\mathcal{L}_{n}$. The set $\mathcal{L}_{n}$ is computable in time $O(\text{poly } (n))$ and the uniform distribution on this set has $O(\log\log{n})$ entropy.
\end{theorem}
\begin{proof}
Set $\ell := \log{n}$, and $k := \log\log{n}$, then the theorem follows by applying directly Lemma~\ref{lem:one-construction-small} and~\ref{lem:one-lifting} to this choice of parameters. By Theorem~\ref{Thm:Derandomization_2_111} the running time is at most $O(\log \log n \cdot \ell^{5+k} \cdot (k!)^5)$ and it is $poly(n)$ because $\ell^{k} = O(poly(n))$ and $k! =O(poly(n))$. 
\end{proof}

\section{Lower bounds for the $k$-secretary problem}
\label{sec:lower-bounds}


\subsection{Entropy lower bound for $k=O(\log^a n)$, for some constant $a\in (0,1)$}
\label{sec:lower-general}


Our proof of a lower bound on the entropy of any $k$-secretary algorithm achieving ratio $1-\epsilon$, for a given $\epsilon\in (0,1)$, stated in Theorem~\ref{thm:lower-general}, generalizes the proof for the (classic) secretary problem in \cite{KesselheimKN15}. 
This generalization is in two ways: first, we reduce the problem of selecting the largest value to the $k$-secretary problem of achieving ratio $1-\epsilon$, by considering a special class of hard assignments of values.
Second, when analyzing the former problem, we have to accommodate the fact that a our algorithm aiming at selecting the largest value can pick $k$ elements, while the classic adversarial algorithm can pick only one element. Below is an overview of the lower bound analysis.

We consider a subset of permutations, $\Pi\subseteq \Pi_n$, of size $\ell$ on which the distribution is concentrated enough (see Lemma~\ref{lem:entropy-support} proved in~\cite{KesselheimKN15}). Next, we fix
a semitone sequence $(x_1,\ldots,x_s)$ w.r.t. $\Pi$ of length $s=\frac{\log n}{\ell+1}$ and consider a specific class of hard assignments of values, defined later.
A semitone sequence with respect to $\pi$, introduced in~\cite{KesselheimKN15}, is defined recursively as follows: an empty sequence is semitone with respect to any permutation
$\pi$, and a sequence $(x_1, \ldots , x_s)$ is semitone w.r.t. $\pi$ if $\pi(x_s) \in\{ \min_{i\in [s]} \pi(x_i), \max_{i\in [s]} \pi(x_i)\}$ and
$(x_1, \ldots, x_{s-1})$ is semitone w.r.t. $\pi$. 
It has been showed that for any given set $\Pi$ of 
$\ell$ permutations of $[n]$, there
is always a sequence of length $s=\frac{\log n}{\ell+1}$ that is semitone with respect to all $\ell$ permutations.

Let 
$V^*=\{1,\frac{k}{1-\epsilon},(\frac{k}{1-\epsilon})^2,\ldots,(\frac{k}{1-\epsilon})^{n-1}\}$. 
An assignment is {\em hard} if the values of the semitone sequence form a permutation of some 
subset of $V^*$
while elements not belonging to the semitone sequence have value $\frac{1-\epsilon}{k}$.
Note that values allocated by hard assignment to elements not in the semitone system are negligible, in the sense that the sum of any $k$ of them is $1-\epsilon$ while the sum of $k$ largest values in the whole system is much bigger than $k$. 
Intuitively, every $k$-secretary algorithm achieving ratio $1-\epsilon$ must select largest value in hard assignments (which is in the semitone sequence) with probability at least $1-\epsilon$ -- this requires analysis of how efficient are deterministic online algorithms selecting $k$ out of $s$ values in finding the maximum value on certain random distribution of hard assignments (see Lemma~\ref{lem:lower-random-adv}) and applying Yao's principle to get an upper bound on the probability of success on any randomized algorithm against hard assignments (see Lemma \ref{lem:lower-deterministic-adv}).

For the purpose of this proof, let us fix $k\le \log^a n$ for some constant $a\in (0,1)$, and parameter $\epsilon \in (0,1)$ (which could be a function of $n,k$).

\begin{lemma}
\label{lem:lower-random-adv}
Consider a set of $\ell<\log n - 1$ permutations $\Pi\subseteq \Pi_n$ and a semitone sequence $(x_1,\ldots,x_s)$ w.r.t. set $\Pi$ of length $s=\frac{\log n}{\ell+1}<\log n$.
Consider any deterministic online algorithm that for any given $\pi\in \Pi$ aims at selecting the largest value, using at most $k$ picks, 
against the following distribution of hard assignments. 

Let $V=V^*$.
We proceed recursively: 
$v(x_s)$ is the middle element of $V$, and 
we apply the recursive procedure u.a.r.: 
(i) on sequence $(x_1,\ldots,x_{s-1})$ and new set $V$ containing $|V|/2$ {\em smallest} elements in $V$ with probability $\frac{1}{s}$ (i.e., $v(x_s)$ is larger than values of the remaining elements with probability $1/s$), and 
(ii) on sequence $(x_1,\ldots,x_{s-1})$ and new set $V$ containing $|V|/2$ {\em largest} elements in $V$ with probability $\frac{s-1}{s}$ (i.e., $v(x_s)$ is smaller than values of the remaining elements with probability $(s-1)/s$).

\ignore{
First, $z$ is selected from $V^*$ u.a.r.
Let $V=V_z$ be the pool of values to be assigned to the semitone sequence 1-1.
Then, we proceed recursively: $v(x_s)=\max V$ with probability $\frac{1}{s}$ and $v(x_s)=\min V$ with probability $\frac{s-1}{s}$, while the assignment of the remaining values from $V\setminus \{v(x_s)\}$ to $(x_1,\ldots,x_{s-1})$ is done recursively and independently.
}


Then, for any $\pi\in\Pi$, the algorithm selects the maximum value with probability at most~$\frac{k}{s}$.
\end{lemma}

\begin{proof}
We start from observing that the hard assignments produced in the formulation of the lemma are disjoint -- it follows directly by the fact that set $V$ of available values is an interval in $V^*$ and it shrinks by half each step; the number of steps $s<\log n$, so in each recursive step set $V$ is non-empty.

In the remainder we prove the sought probability.
Let $A_t^i$, for $1\le t\le s$ and $0\le i\le k$, be the event that the algorithm picks at most $i$ values from $v(x_1),\ldots,v(x_t)$.
Let $B_t$ be the probability that the algorithm picks the largest of values $v(x_1),\ldots,v(x_t)$, in one of its picks.
Let $C_t$ be the probability that the algorithm picks value $v(x_t)$.
We prove, by induction on lexicographic pair $(t,i)$, that $\Pr{B_t|A_t^i}\le \frac{i}{t}$.
Surely, the beginning of the inductive proof for any pair $(t,i=t)$ is correct: $\Pr{B_t|A_t^t}\le 1$.
%
%

Consider an inductive step for $i< t\le s$.
Since, by the definition of semitone sequence $(x_1,\ldots,x_s)$, element $x_t$ could be either before all elements $x_1,\ldots,x_{t-1}$ or after all elements $x_1,\ldots,x_{t-1}$ in permutation $\pi$, we need to analyze both of these cases:

\vspace*{1ex}
\noindent
{\bf Case 1: $\pi(x_t)<\pi(x_1),\ldots,\pi(x_{t-1})$.}
Consider the algorithm when it receives the value of $x_t$. It has not seen the values of elements $x_1,\ldots,x_{t-1}$ yet. Assume that the algorithm already picked $k-i$ values before processing element $x_t$.
Note that, due to the definition of the hard assignment in the formulation of the lemma, the knowledge of values occurring by element $x_t$ only informs the algorithm about set $V$ from which the adversary draws values for sequence $(x_1,\ldots,x_{t-1})$; thus this choice of values is independent, for any fixed prefix of values until the occurrence of element $x_t$. We use this property when deriving the probabilities in this considered case.

We consider two conditional sub-cases, depending on whether either $C_t$ or $\neg C_t$ holds, starting from the former:
\[
\Pr{B_t|A_t^i \& C_t}
=
\frac{1}{t} + \frac{t-1}{t} \cdot \Pr{B_{t-1}|A_{t-1}^{i-1} \& C_t}
=
\frac{1}{t} + \frac{t-1}{t} \cdot \Pr{B_{t-1}|A_{t-1}^{i-1}}
=
\frac{1}{t} + \frac{t-1}{t} \cdot \frac{i-1}{t-1}
=
\frac{i}{t}
\ ,
\]
where 
\begin{itemize}
    \item 
the first equation comes from the fact that $val(x_t)$ is the largest among $v(x_1),\ldots,v(x_t)$ with probability $\frac{1}{t}$ (and it contributes to the formula because of the assumption $C_t$ that algorithm picks $v(x_t)$) and $v(x_t)$ is not the largest among $v(x_1),\ldots,v(x_t)$ with probability $\frac{t-1}{t}$ (in which case the largest value must be picked within the first $v(x_1),\ldots,v(x_{t-1})$ using $i-1$ picks), and
\item
the second equation comes from the fact that $B_{t-1}$ and $C_t$ are independent, and
\item
the last equation holds by inductive assumption for $(t-1,i-1)$.
\end{itemize}
In the complementary condition $\neg C_t$ we have:
\[
\Pr{B_t|A_t^i \& \neg C_t}
=
\frac{t-1}{t} \cdot \Pr{B_{t-1}|A_{t-1}^{i} \& \neg C_t}
=
\frac{t-1}{t} \cdot \Pr{B_{t-1}|A_{t-1}^{i}}
=
\frac{t-1}{t} \cdot \frac{i}{t-1}
=
\frac{i}{t}
\ ,
\]
where 
\begin{itemize}
\item 
the first equation follows because if the algorithm does not pick $v(x_t)$ then the largest of values  $v(x_1),\ldots,v(x_t)$ must be within $v(x_1),\ldots,v(x_{t-1})$ and not $v(x_t)$ (the latter happens with probability $\frac{t-1}{t}$), and
\item
the second equation comes from the fact that $B_{t-1}$ and $\neg C_t$ are independent, and
\item
the last equation holds by inductive assumption for $(t-1,i)$.
\end{itemize}
Hence,
\[
\Pr{B_t|A_t^i}
=
\Pr{B_t|A_t^i \& C_t}\cdot \Pr{C_t}
+
\Prob[B_t|A_t^i \& \neg C_t]\cdot \Pr{\neg C_t}
= \frac{i}{t} \cdot \left( \Pr{C_t} + \Pr{\neg C_t}\right)
=
\frac{i}{t}
\ .
\]
It concludes the analysis of Case 1.

\vspace*{1ex}
\noindent
{\bf Case 2: $\pi(x_t)>\pi(x_1),\ldots,\pi(x_{t-1})$.}
Consider the algorithm when it receives the value of $x_t$. It has already seen the values of elements $x_1,\ldots,x_{t-1}$; therefore, we can only argue about conditional event on the success in picking the largest value among $v(x_1),\ldots,v(x_{t-1})$, i.e., event $B_{t-1}$.

Consider four conditional cases, depending on whether either of $C_t,\neg C_t$ holds and whether either of $B_{t-1},\neg B_{t-1}$ holds, starting from sub-case $B_{t-1}\& C_t$:
\[
\Pr{B_t|A_t^i \& B_{t-1} \& C_t}
=
\frac{\Pr{B_t\& B_{t-1}|A_t^i \& C_t}}{\Pr{B_{t-1}|A_t^i \& C_t}}
=
\frac{\Pr{B_t\& B_{t-1}|A_t^i \& C_t}}{\Pr{B_{t-1}|A_{t-1}^{i-1}}}
=
1
\ ,
\]
since the algorithm already selected the largest value among $v(x_1),\ldots,v(x_{t-1})$ (by $B_{t-1}$) and now it also selects $v(x_t)$ (by $C_t$).
We also used the observation $A_t^i\& C_t = A_{t-1}^{i-1}$.
Next sub-case, when the conditions $\neg B_{t-1}\& C_t$ hold, implies:
\[
\Pr{B_t|A_t^i \& \neg B_{t-1} \& C_t}
=
\frac{\Pr{B_t\& \neg B_{t-1}|A_t^i \& C_t}}{\Pr{\neg B_{t-1}|A_t^i \& C_t}}
=
\frac{\Pr{B_t\& \neg B_{t-1}|A_t^i \& C_t}}{1-\Pr{B_{t-1}|A_{t-1}^{i-1}}}
=
\frac{1}{t}
\ ,
\]
because when the maximum value among $v(x_1),\ldots,v(x_{t-1})$ was not selected (by $\neg B_{t-1}$) the possibility that the selected (by $C_t$) $v(x_t)$ is the largest among $v(x_1),\ldots,v(x_{t})$ is $\frac{1}{t}$, by definition of values $v(\cdot)$.
As in the previous sub-case, we used $A_t^i\& C_t = A_{t-1}^{i-1}$.
When we put the above two sub-cases together, for $B_{t-1}\& C_t$ and $\neg B_{t-1}\& C_t$, we get:
\[
\Pr{B_t\& B_{t-1}|A_t^i \& C_t} + \Pr{B_t\& \neg B_{t-1}|A_t^i \& C_t}
=
\Pr{B_{t-1}|A_{t-1}^{i-1}} \cdot 1 + \left(1-\Pr{B_{t-1}|A_{t-1}^{i-1}}\right) \cdot \frac{1}{t}
=
\]
\[
=
\frac{i-1}{t-1} + \left(1-\frac{i-1}{t-1}\right) \cdot \frac{1}{t}
=
\frac{(i-1)t+(t-i)}{(t-1)t}
=
\frac{(t-1)i}{(t-1)t}
=
\frac{i}{t}
\ ,
\]
where the first equation comes from the previous sub-cases, the second is by inductive assumption, and others are by simple arithmetic.

We now consider two remaining sub-cases, starting from $B_{t-1}\& \neg C_t$: 
\[
\Pr{B_t|A_t^i \& B_{t-1} \& \neg C_t}
=
\frac{\Pr{B_t\& B_{t-1}|A_t^i \& \neg C_t}}{\Pr{B_{t-1}|A_t^i \& \neg C_t}}
=
\frac{\Pr{B_t\& B_{t-1}|A_t^i \& \neg C_t}}{\Pr{B_{t-1}|A_{t-1}^i}}
=
\frac{t-1}{t}
\ ,
\]
since the algorithm already selected the largest value among $v(x_1),\ldots,v(x_{t-1})$ (by $B_{t-1}$) and now it also selects $v(x_t)$ (by $C_t$).
We also used the observation $A_t^i\& \neg C_t = A_{t-1}^{i}$.
Next sub-case, when the conditions $\neg B_{t-1}\& \neg C_t$ hold, implies:
\[
\Pr{B_t|A_t^i \& \neg B_{t-1} \& \neg C_t}
=
\frac{\Pr{B_t\& \neg B_{t-1}|A_t^i \& \neg C_t}}{\Pr{\neg B_{t-1}|A_t^i \& \neg C_t}}
=
\frac{\Pr{B_t\& \neg B_{t-1}|A_t^i \& \neg C_t}}{1-\Pr{B_{t-1}|A_{t-1}^i}}
=
0
\ ,
\]
because when the maximum value among $x_1,\ldots,x_{t-1}$ was not selected (by $\neg B_{t-1}$) the possibility that the selected (by $C_t$) $v(x_t)$ is the largest among $v(x_1),\ldots,v(x_{t})$ is $\frac{1}{t}$, by definition of values $v(\cdot)$.
We also used the observation $A_t^i\& \neg C_t = A_{t-1}^{i}$.
When we put the last two sub-cases together, for $B_{t-1}\& \neg C_t$ and $\neg B_{t-1}\& \neg C_t$, we get:
\[
\Pr{B_t\& B_{t-1}|A_t^i \& \neg C_t} + \Pr{B_t\& \neg B_{t-1}|A_t^i \& \neg C_t}
=
\Pr{B_{t-1}|A_{t-1}^i} + \left(1-\Pr{B_{t-1}|A_{t-1}^i}\right) \cdot \frac{1}{t}
=
\]
\[
=
\frac{i-1}{t-1} + \left(1-\frac{i-1}{t-1}\right) \cdot \frac{1}{t}
=
\frac{(i-1)t+(t-i)}{(t-1)t}
=
\frac{(t-1)i}{(t-1)t}
=
\frac{i}{t}
\ ,
\]
where the first equation comes from the previous sub-cases, the second is by inductive assumption, and others are by simple arithmetic.

Hence, similarly as in Case 1, we have
\[
\Pr{B_t|A_t^i}
=
\Pr{B_t|A_t^i \& C_t}\cdot \Pr{C_t}
+
\Pr{B_t|A_t^i \& \neg C_t}\cdot \Pr{\neg C_t}
= \frac{i}{t} \cdot \left( \Pr{C_t} + \Pr{\neg C_t}\right)
=
\frac{i}{t}
\ .
\]
It concludes the analysis of Case 2, and also the inductive proof.

It follows that 
$\Pr{B_s|A_t^k} = \frac{k}{s}$, and since $\Pr{A_s^k} = 1$ (as the algorithm does $k$ picks in the whole semitone sequence), we get $\Pr{B_s}=\frac{k}{s}$.
\end{proof}

Applying Yao's principle~\cite{Yao77} to Lemma~\ref{lem:lower-random-adv}, we get:

\begin{lemma}
\label{lem:lower-deterministic-adv}
Fix any $\epsilon \in (0,1)$.
For any set $\Pi\subseteq \Pi_n$ of an $\ell<\log n -1$ permutations and any probabilistic distribution on it, and for any online algorithm using $k$ picks to select the maximum value, 
there is an adversarial (worst-case) hard assignment of values to elements in $[n]$ such that: 

(i) the maximum assigned value is unique and bigger by factor at least $\frac{k}{1-\epsilon}$ from other used values,

(ii) highest $s=\frac{\log n}{\ell+1}$ values are at least $1$, while the remaining ones are $\frac{1-\epsilon}{k}$,

(iii) the algorithm selects the maximum allocated value with probability at most $\frac{k}{s}$.
%
\end{lemma}

\begin{proof}
Consider a semitone sequence w.r.t. the set of permutations $\Pi$, of length $s=\frac{\log n}{\ell+1}$ (it exists as shown in \cite{KesselheimKN15}), and restrict for now to this sub-sequence of the whole $n$-value sequence.
Consider any online algorithm that ignores elements that are not in this sub-sequence.
We apply Yao's principle~\cite{Yao77} to Lemma~\ref{lem:lower-random-adv}: the latter computes a lower bound on the cost (probability of selecting largest value) of a deterministic $k$-secretary algorithm, for inputs being hard assignments selected from distribution in Lemma~\ref{lem:lower-random-adv}.
By Yao's principle, there is a deterministic (worst-case) adversarial hard assignment 
values from set $V^*\cup \{\frac{1-\epsilon}{k}\}$ 
such that for any (even randomized) algorithm and probabilistic distribution on $\Pi$, 
the probability of the algorithm to select the largest of the assigned values with at most $k$ picks is at most $\frac{k}{s}$.
The hard assignment satisfies, by definition, also the first two conditions in the lemma statement.
\ignore{
Now, we consider the whole sequence of $n$ values.
The adversary assigns value $\frac{1-\epsilon}{k}$ to all elements not from the semitone sequence, and uses the assignment to the semitone sequence described in the previous paragraph. Clearly, this assignment satisfies conditions (i) and (ii) of the lemma. If this algorithm selected the largest value (in the whole sequence) with probability larger than $\frac{k}{s}$, and thus violate condition (iii), then, since non-semitone elements are clearly recognized by their values, we could restrict the algorithm to the semitone sequence and obtain contradiction with

Since $1/k$ is smaller by factor at least $k$ from any element in $V$, the algorithm can clearly reject all elements not from the semitone sequence (if not, it gets less than $k$ choices left for elements in the semitone sequence, while gaining not more than a value of a single element in the semitone sequence, more precisely, $1$).
}
\end{proof}




We can extend Lemma~\ref{lem:lower-deterministic-adv} to any distribution on a set $\Pi$ of permutations of $[n]$ with an entropy $H$, by using the following lemma from~\cite{KesselheimKN15}, in order to obtain the final proof of Theorem~\ref{thm:lower-general} (re-stated below).

\begin{lemma}[\cite{KesselheimKN15}]
\label{lem:entropy-support}
Let $\pi$ be drawn from a finite set $\Pi_n$ by a distribution of entropy $H$. Then, for any $\ell \ge 4$,
there is a set $\Pi \subseteq \Pi_n$, $|\Pi| \le \ell$, such that $\Pr{\pi\in\Pi} \ge 1 - \frac{8H}{\log(\ell - 3)}$.
\end{lemma}

\noindent
{\bf Theorem~\ref{thm:lower-general} }
{\em Assume $k\le \log^a n$ for some constant $a\in (0,1)$.
Then, any algorithm (even fully randomized) solving $k$-secretary problem while drawing permutations from some distribution on $\Pi_n$ with an entropy $H\le \frac{1-\epsilon}{9} \log\log n$, cannot achieve the expected ratio of at least $1-\epsilon$, for any $\epsilon\in (0,1)$ and sufficiently large $n$.}

\vspace{1ex}
\begin{proof}
Let us fix $\ell=\sqrt{\frac{\log n}{k}}-1$. 
By Lemma~\ref{lem:entropy-support}, there is a set $\Pi\subseteq \Pi_n$ of size at most $\ell$ such that $\Pr{\pi\in\Pi} \ge 1 - \frac{8H}{\log(\ell - 3)}$. Let $s=\frac{\log n}{\ell +1}$ be the length of a semitone sequence w.r.t. $\Pi$. 

By Lemma~\ref{lem:lower-deterministic-adv} applied to the conditional distribution on set $\Pi$, there is an adversarial hard assignment of values such that the probability of selecting the largest value is at most $\frac{k}{s}$.
Summing up the events and using Lemma~\ref{lem:entropy-support}, the probability of the algorithm selecting the largest value is at most 
\[
\frac{k}{s}\cdot 1 + \frac{8H}{\log(\ell - 3)}
=
\frac{k\cdot (\ell+1)}{\log n} + \frac{8H}{\log(\ell - 3)}
=
\sqrt{\frac{k}{\log n}} + \frac{8H}{\log(\sqrt{\frac{\log n}{k}} - 4)}
\ ,
\]
which is smaller than $1-\epsilon$, for any $\epsilon\in (0,1)$, for sufficiently large $n$, because $k\le \log^a n$, where $a\in (0,1)$ is a constant, and $H\le \frac{1-\epsilon}{9} \log\log n$.

To complete the proof, recall that, by the definition of hard assignments and statements (i) and (ii) in Lemma~\ref{lem:lower-deterministic-adv}, the maximum value selected from $V$ is unique, and is bigger from other values by factor at least $\frac{k}{1-\epsilon}$, therefore the event of selecting $k$ values with ratio $1-\epsilon$ is a sub-event of the considered event of selecting largest value. Thus, the probability of the former is upper bounded by the probability of the latter and so the former cannot be achieved as well.
\end{proof}

\subsection{Entropy lower bound for wait-and-pick algorithms} 
\label{sec:lower-wait-and-pick}

Assume that there is a deterministic wait-and-pick algorithm for the $k$-secretary problem with competitive ratio $1-\epsilon$. Let $m$ be the 
checkpoint position and $\tau$ be any statistic (i.e., the algorithm selects $\tau$-th largest element among the first $m$ elements as chosen element, and then from the elements after position $m$, it selects every element greater than or equal to the statistic). Our analysis works for any statistics. Let $\ell$ be the number of permutations, from which the order is chosen uniformly at random. 
We prove that no wait-and-pick algorithm achieves simultaneously a inverse-polynomially (in $k$) small error $\epsilon$ and entropy asymptotically smaller than $\log k$. More precisely, we re-state and prove Theorem~\ref{thm:lower}.

\vspace*{3ex}
\noindent
{\bf Theorem~\ref{thm:lower} }
{\em Any wait-and-pick algorithm solving $k$-secretarial problem with competitive ratio of at least $(1-\epsilon)$ requires entropy $\Omega(\min\{\log 1/\epsilon,\log \frac{n}{2k}\})$.}

\vspace*{3ex}
\begin{proof}
W.l.o.g.~and for simplicity, assume $1/\epsilon$ is an integer.
Let $G=(V,W,E)$ be a bipartite graph, where $V$ is the set of $n$ elements, $W$ corresponds to the set of $\ell$ permutations, and a neighborhood of $i\in W$ is defined as the set of elements (in $V$) which are on the left hand side of checkpoint $m$ in the $i$-th permutation; clearly $|E|=\ell\cdot m$. Let $d$ denote an average degree of a node in $V$, i.e., $d=\frac{|E|}{n}=\frac{\ell \cdot m}{n}$. 

Consider first the case when $m\ge k$.
We prove that $\ell \ge 1/\epsilon$. 
Consider a different strategy of the adversary: it processes elements $i\in W$ one by one, and selects $\epsilon \cdot k$ neighbors of element $i$ that has not been selected before to set $K$. This is continued until set $K$ has $k$ elements or all elements in $W$ has been considered.
Note that if during the above construction the current set $K$ has at most $k(1-\epsilon)$ elements, the adversary can find $\epsilon \cdot k$ neighbors of the currently considered $i\in W$ that are different from elements in $K$ and thus can be added to $K$, by assumption $m\ge k$.
If the construction stops because $K$ has $k$ elements, it means that $\ell\ge 1/\epsilon$, because $1/\epsilon$ elements in $W$ have had to be processed in the construction.
If the construction stops because all the elements in $W$ have been processed but $K$ is of size smaller than $k$, it means that $|W|=\ell<1/\epsilon$; however, if we top up the set $K$ by arbitrary elements in $V$ so that the resulting $K$ is of size $k$, no matter what permutation is selected the algorithm misses at least $\epsilon \cdot k$ elements in $K$, and thus its value is smaller by factor less than $(1-\epsilon)$ from the optimum. and we get a contradiction.
Thus, we proved $\ell \ge 1/\epsilon$, and thus the entropy needed is at least $\log 1/\epsilon$, which for optimal algorithms with $\epsilon=\Theta(k^{-1/2})$ gives entropy $\Theta(\log k)$.

Consider now the complementary case when $m<k$. The following has to hold: $\ell\cdot (m+k)\ge n$. This is because 
in the opposite case the adversary could allocate value $1$ to an element which does not occur in the first $m+k$ positions of any of the $\ell$ permutations, and value $\frac{1-\epsilon}{k}$ to all other elements -- in such scenario, the algorithm would pick the first $k$ elements after the checkpoint position $m$ (as it sees, and thus chooses, the same value all this time -- it follows straight from the definition of wait-and-pick checkpoint),
for any of the $\ell$ permutations, obtaining the total value of $1-\epsilon$, while the optimum is clearly $1+(k-1)\cdot\frac{1-\epsilon}{k}>1$ contradicting the competitive ratio $1-\epsilon$ of the algorithm.
It follows from the equation $\ell\cdot (m+k)\ge n$ that $\ell \ge \frac{n}{m+k} > \frac{n}{2k}$, and thus the entropy is $\Omega(\log\frac{n}{2k})$.

To summarize both cases, the entropy is $\Omega(\min\{\log 1/\epsilon,\log \frac{n}{2k}\})$.
\ignore{
\dk{Ponizsze na razie sie nie stosuje}
Assume, for the sake of contradiction, that for any constant $\eta>0$, $\ell=o(k^\eta)$ and $\epsilon=\Theta(k^{-\eta})$.
By a pigeonhole principle, there is a set of elements $K\subseteq V$ such that $|K|=k$ and an average degree of an element in $K$ is at least $d$.
Then, by the uniformity of distribution assumption, we have that ??? and consequently, each permutation must have at most $\epsilon \cdot k$ elements in from $K$ on their left hand side (considering adversary who allocates the same value to all elements in $K$ and arbitrary small value to others), thus in total the number of occurrences of elements from $K$ on left hand sides on the considered $\ell$ permutations is at most $\epsilon \cdot k \cdot \ell$.
On the other hand, the sum of degrees of elements in $k$ is at least $k\cdot d \ge k\cdot \frac{\ell \cdot m}{n}$.
They together imply $\epsilon \ge \frac{m}{n}$, and consequently, $m\le n\epsilon$.

We would like to show that for $k$ sufficiently large but still sub-logarithmic, the entropy is $\Omega(\log\log n)$.
}
\end{proof}

In particular, it follows from Theorem~\ref{thm:lower} 
that for $k$ such that $k$ is super-polylogarithmic and sub-$\frac{n}{\polylog n}$, the entropy of competitive ratio-optimal algorithms is $\omega(\log\log n)$. Moreover, if $k$ is within range of some polynomials of $n$ of degrees smaller than $1$, the entropy is $\Omega(\log n)$. 

\subsection{$\Omega(\log\log n + (\log k)^2)$ entropy of previous solutions}
\label{sec:previous-suboptimality}

All previous solutions but~\cite{KesselheimKN15} used uniform distributions on the set of all permutations of $[n]$, which requires large entropy $\Theta(n\log n)$.\footnote{%
Some of them also used additional randomness, but with negligible entropy $o(n\log n)$.}
In~\cite{KesselheimKN15}, the $k$-secretary algorithm uses $\Theta(\log\log n)$ entropy to choose a permutation u.a.r. from a given set, however, it also uses recursively additional entropy to choose the number of blocks $q'$.
It starts with $q'$ being polynomial in $k$, and in a recursive call it selects a new $q'$ from the binomial distribution $Binom(q',1/2)$. It continues until $q'$ becomes $1$.
Below we estimate from below the total entropy needed for this random process.

Let $X_i$, for $i=1,\ldots,\tau$, denote the values of $q'$ selected in subsequent recursive calls, where $\tau$ is the first such that $X_\tau=1$. We have $X_1=Binom(q',1/2)$ and recursively, $X_{i+1}=Binom(X_i,1/2)$.
We need to estimate the joint entropy $\cH(X_1,\ldots,X_\tau)$ from below.
Joint entropy can be expressed using conditional entropy as follows:
\begin{equation}
\label{eq:conditional-entropy}
\cH(X_1,\ldots,X_\tau) =
\cH(X_1) + \sum_{i=2}^\tau \cH(X_i|X_{i-1},\ldots,X_1) 
\ .
\end{equation}
By the property of $Binom(q',1/2)$ and the fact that $q'$ is a polynomial on $k$, its entropy $\cH(X_1)=\Theta(\log q')=\Theta(\log k)$.
We have:
\[
\cH(X_i|X_{i-1},.\ldots,X_1) 
=
\sum_{q_i\ge \ldots \ge q_{i-1}} \Pr{X_1=q_1,\ldots,X_{i-1}=q_{i-1}} \cdot \cH(X_i|X_1=q_1,\ldots,X_{i-1}=q_{i-1})
\]
\[
=
\sum_{q_i\ge \ldots \ge q_{i-1}} \Pr{X_1=q_1,\ldots,X_{i-1}=q_{i-1}} \cdot \cH(X_i|X_{i-1}=q_{i-1})
\]
\[
=
\Theta\left(\Pr{X_1\in (\frac{1}{3}q',\frac{2}{3}q'), X_2\in (\frac{1}{3}X_{1},\frac{2}{3}X_{1})\ldots,X_{i-1}\in (\frac{1}{3}X_{i-2},\frac{2}{3}X_{i-2})}\right) \cdot
\cH\left(X_i\Big| X_{i-1}\in (\frac{1}{3^{i-1}}q',\frac{2^{i-1}}{3^{i-1}}q')\right)
\ ,
\]
where the first equation is the definition of conditional entropy, second follows from the fact that once $q_{i-1}$ is fixed, the variable $X_1$ does not depend on the preceding $q_{i-2},\ldots,q_1$, and the final asymptotics follows from applying Chernoff bound to each $X_1,\ldots,X_{i-1}$ and taking the union bound.
Therefore, for $i\le \frac{1}{2}\log_3 q'$, we have 
\[
\cH(X_i|X_{i-1},.\ldots,X_1) 
=
(1-o(1)) \cdot \cH(X_i|X_{i-1}\in\Theta(\text{poly}(k)))
=
\Theta(\log k)
\ .
\]
Consequently, putting all the above into Equation~(\ref{eq:conditional-entropy}), we get
\[
\cH(X_1,\ldots,X_\tau) 
=
\Theta(\log_3 k \cdot \log k) 
= 
\Theta(\log^2 k)
\ .
\]
The above proof leads to the following.

\begin{proposition}\label{prop:Kessel_Large_Entropy}
 The randomized $k$-secretary algorithm of Kesselheim, Kleinberg and Niazadeh \cite{KesselheimKN15} uses randomization that has a total entropy $\Omega(\log\log n + (\log k)^2)$, where entropy $\log\log n$ corresponds to the distribution from which it samples a random order, and entropy $(\log k)^2$ corresponds to the internal random bits of the algorithm.
\end{proposition}

Our algorithm shaves off the additive $\Theta(\log^2 k)$ from the formula for all $k$ up to nearly $\log n$. 

\section{Lower bounds and characterization for the classical secretary problem}\label{section:lb_classic_secr}

Given a wait-and-pick algorithm for the classical secretary ($1$-secretary) problem, we will denote its checkpoint by $m_0$ (we will reserve $m$ to be used as a variable  checkpoint in the analysis). 

We will first understand the optimal success probability of the best secretary algorithms. Let $f(k,m)=\frac{m}{k} (H_{k-1}-H_{m-1})$, where $H_k$ is the $k$-th harmonic number, $H_k = 1 + \frac{1}{2} + \frac{1}{3} + \ldots + \frac{1}{k}$. It is easy to prove that $f(n,m_0)$ is the exact success probability of the wait-and-pick algorithm with checkpoint $m_0$ when random order is given by choosing u.a.r.~a permutation $\pi \in \Pi_n$, see \cite{GuptaSingla}.

\begin{lemma}\label{lemma:f_expansion_1}
The following asymptotic behavior holds, if $k \rightarrow \infty$ and 
$j\le \sqrt{k}$ is such that $m=k/e+j$ is an integer in $[k]$:
\[
f\left(k, \frac{k}{e} + j\right) = \frac{1}{e} - \left(\frac{1}{2e} - \frac{1}{2} + \frac{e j^2}{2k}  \right) \frac{1}{k} + \Theta\left( \left( \frac{1}{k} \right)^{3/2} \right)
\ .
\]
\end{lemma}

\begin{proof}
 To prove this expansion we extend the harmonic function $H_n$ to real numbers. Namely, for any real number $x \in \reals$ we use the well known definition:
\[
  H_x = \psi(x+1) + \gamma \ ,
\] 
where $\psi$ is the digamma function and  $\gamma$ is the Euler-Mascheroni
constant. Digamma function is just the derivative of the logarithm of the 
gamma function $\Gamma(x)$, pioneered by Euler, Gauss and Weierstrass. Both functions are important and widely studies in real and complex analysis.

For our purpose, it suffices to use the following inequalities that hold for any real $x > 0 $ (see Theorem~5 in \cite{Gordon1994}):
$$
  \ln(x) - \frac{1}{2x} - \frac{1}{12x^2} + \frac{1}{120(x + 1/8)^4} \,\, < \,\,  \psi(x) \,\, < \,\,  \ln(x) - \frac{1}{2x} - \frac{1}{12x^2} + \frac{1}{120x^4} \, .
$$ Now we use these estimates for $f(k,m)$ for $\psi(k)$ and $\psi(m)$ with $m = k/e + j$:
\begin{eqnarray*}
  & & f(k,m) = \frac{m}{k} (\psi(k) - \psi(m)) = \\
  & &        = \frac{m}{k} \left(\ln(k) - \frac{1}{2k} - \frac{1}{12k^2} + \frac{1}{120(k + \theta(k))^4}  
                        - \ln(m) + \frac{1}{2m} + \frac{1}{12m^2} - \frac{1}{120(m + \theta(m))^4} \right) \\
  & &        = \frac{m}{k} \left(1 + \ln\left(\frac{k}{e m}\right) - \frac{1}{2k} + \frac{1}{2m} - \frac{1}{12k^2} + \frac{1}{12m^2} + \frac{1}{120(k + \theta(k))^4} - \frac{1}{120(m + \theta(m))^4} \right) \, ,
\end{eqnarray*} where $\theta(x) \in (0,1/8)$. Now, taking into account that 
$m\in [k]$,
we can suppress the low order terms under $\Theta\left( \frac{1}{k^2} \right)$ to obtain
\begin{eqnarray*}
  f(k,m) &=& \frac{m}{k} \left(1 + \ln\left(\frac{k}{e m}\right) - \frac{1}{2k} + \frac{1}{2m}\right) +  \Theta\left( \frac{1}{m^2} \right) \\
         &=& \frac{m}{k} - \frac{m}{k} \ln\left(\frac{e m}{k}\right) - \frac{m}{2k^2} + \frac{1}{2k} +  \Theta\left( \frac{1}{m^2} \right) \\
         &=& \frac{1}{e} + \frac{j}{k} - \left(\frac{1}{e} + \frac{j}{k}\right) \ln\left(\frac{k + j e}{k}\right) - \frac{k/e + j}{2k^2} +    
             \frac{1}{2k} +  \Theta\left( \frac{1}{k^2} \right) \\
         &=& \frac{1}{e} + \frac{j}{k} - \left(\frac{1}{e} + \frac{j}{k}\right) \ln\left(\frac{k + j e}{k}\right) - \frac{1}{2ek} + \frac{1}{2k} +  \Theta\left( \frac{1}{k^{3/2}} \right)  \, .
\end{eqnarray*} We will now use the following well known Taylor expansion
\begin{eqnarray*}
 \ln\left(\frac{k + j e}{k}\right) &=& \left(\frac{k + j e}{k} - 1 \right) - \frac{1}{2} \left(\frac{k + j e}{k} - 1 \right)^2 + \frac{1}{3}    
                                       \left(\frac{k + j e}{k} - 1 \right)^3 - \ldots \\
                                &=& \frac{j e}{k} - \frac{1}{2} \left(\frac{j e}{k}\right)^2 + \frac{1}{3} \left(\frac{j e}{k}\right)^3 - \ldots 
  = \frac{j e}{k} - \frac{1}{2} \left(\frac{j e}{k}\right)^2 + \Theta\left( \frac{1}{k^{3/2}} \right)\, .
\end{eqnarray*} Using this expansion, we can continue from above as follows
\begin{eqnarray*}
   f(k,m) &=& \frac{1}{e} + \frac{j}{k} - \left(\frac{1}{e} + \frac{j}{k}\right) \ln\left(\frac{k + j e}{k}\right) - \frac{1}{2ek} + \frac{1}{2k} +  \Theta\left( \frac{1}{k^{3/2}} \right) \\
   &=& \frac{1}{e} + \frac{j}{k} - \left(\frac{1}{e} + \frac{j}{k}\right) \left(\frac{j e}{k} - \frac{1}{2} \left(\frac{j e}{k}\right)^2\right) - \frac{1}{2ek} + \frac{1}{2k} +  \Theta\left( \frac{1}{k^{3/2}} \right) \\
   &=& \frac{1}{e} + \frac{1}{2e} \left(\frac{j e}{k}\right)^2 - \frac{j^2 e}{k^2} + \frac{j^3 e^2}{k^3} - \frac{1}{2ek} + \frac{1}{2k} +  \Theta\left( \frac{1}{k^{3/2}} \right) \\
   &=& \frac{1}{e} + \frac{1}{2e} \left(\frac{j e}{k}\right)^2 - \frac{j^2 e}{k^2} - \frac{1}{2ek} + \frac{1}{2k} +  \Theta\left( \frac{1}{k^{3/2}} \right) \\
   &=& \frac{1}{e} - \frac{j^2 e}{2 k^2} - \frac{1}{2ek} + \frac{1}{2k} +  \Theta\left( \frac{1}{k^{3/2}} \right) 
   =
   \frac{1}{e} - \left(\frac{1}{2e} + \frac{e j^2}{2k} - \frac{1}{2} \right) \frac{1}{k} + \Theta\left( \frac{1}{k^{3/2}} \right) \ .
\end{eqnarray*}
\end{proof}

We will now precisely characterize the maximum of function $f$. Recall that,
$f(k,m)=\frac{m}{k} (H_{k-1}-H_{m-1})$, and note that $1\le m\le k$. We have the discrete derivative~of~$f$:
$
h(m) = f(k,m+1)-f(k,m) =
\frac{1}{k} (H_{k-1}-H_m-1)
\ ,
$
which is positive for $m\le m_0$
and negative otherwise, for some 
$m_0=\max\{m>0: H_{k-1}-H_m-1>0\}$.

\begin{lemma}\label{l:deriv_bounds_1}
There exists an absolute constant $c > 1$ such that for any integer $k \geq c$, we have that $h\left(\lfloor \frac{k}{e} \rfloor - 1\right) > 0$ and
$h\left(\lfloor \frac{k}{e} \rfloor + 1 \right) < 0$. Moreover, function $f(k,\cdot)$ achieves its maximum for $m \in \{\lfloor \frac{k}{e} \rfloor, \lfloor \frac{k}{e} \rfloor + 1\}$, and is monotonically increasing for smaller values of $m$ and monotonically decreasing for larger values of $m$.
\end{lemma}

\begin{proof}
We first argue that function $f(n,\cdot)$ has exactly one local maximum, which is also global maximum.
To see it, observe that function $h(m)$ is positive until $H_m+1$ gets bigger than $H_{k-1}$, which occurs for a single value $m$ (as we consider function $h$ for discrete arguments) and remains negative afterwards.
Thus, function $f(n,\cdot)$ is monotonically increasing until that point, and decreasing afterwards.
Hence, it has only one local maximum, which is also global maximum. 

It remains to argue that the abovementioned argument $m$ in which function $f(n,\cdot)$ achieves maximum is in $\{\lfloor \frac{k}{e} \rfloor, \lfloor \frac{k}{e} \rfloor + 1\}$.
We will make use of the following known inequalities.

\begin{lemma}\label{l:harmonic_bounds} 
The following bounds hold for the harmonic and logarithmic functions:
\begin{enumerate}
\item[(1)] $\frac{1}{2(x+1)} < H_x -\ln x -\gamma < \frac{1}{2x}$ ,
\item[(2)] $\frac{1}{24(x+1)^2} < H_x -\ln (x+1/2) -\gamma < \frac{1}{24x^2}$ ,
\item[(3)] $\frac{x}{1+x} \leq \ln(1+x) \leq x$, which holds for $x > -1$.
\end{enumerate}
\end{lemma}

\noindent
Using the first bound $(1)$ from Lemma \ref{l:harmonic_bounds}, we obtain the following:
\begin{eqnarray*}
k \cdot h\left(\left\lfloor \frac{k}{e} \right\rfloor - 1\right)
&=& 
H_{k-1} - H_{\left\lfloor \frac{k}{e} \right\rfloor - 1} - 1
\\
&>& 
\ln(k-1) + \frac{1}{2k} - \ln\left(\left\lfloor \frac{k}{e} \right\rfloor - 1\right) - \frac{1}{2(\left\lfloor \frac{k}{e} \right\rfloor - 1)} - 1
\\
&=&
\ln\left(\frac{e(k-1)}{e(\lfloor \frac{k}{e} \rfloor - 1 )}\right) + \frac{1}{2k} - \frac{1}{2(\left\lfloor \frac{k}{e} \right\rfloor - 1)} - 1
\\
&=&
\ln\left(\frac{k-1}{e(\lfloor \frac{k}{e} \rfloor - 1 )}\right) + \frac{1}{2k} - \frac{1}{2(\left\lfloor \frac{k}{e} \right\rfloor - 1)}
\\
& & \mbox{Rewriting inequality (3) in Lemma \ref{l:harmonic_bounds} as } \ln(y) \geq 1-1/y, \mbox{ we obtain:}
\\
&\geq&
1 - \frac{e(\lfloor \frac{k}{e} \rfloor - 1 )}{k-1} + \frac{1}{2k} - \frac{1}{2(\left\lfloor \frac{k}{e} \right\rfloor - 1)}
\\
&>&
1 - \frac{e(\lfloor \frac{k}{e} \rfloor - 1 )}{k-1} - \frac{1}{2(\left\lfloor \frac{k}{e} \right\rfloor - 1)}
\,\,\,\, > \,\,\, 0
\ ,
\end{eqnarray*} where the last inequality holds because it is equivalent to
$$
  2(k-1)(\lfloor k/e \rfloor - 1 ) > 2e(\lfloor k/e \rfloor - 1 )^2 + k-1 \Leftrightarrow 2k \lfloor k/e \rfloor  + (4e-2) \lfloor k/e \rfloor > 2e(\lfloor k/e \rfloor)^2 +3k +2e-3
$$
$$
  \Leftarrow 2k \lfloor k/e \rfloor > 2e(\lfloor k/e \rfloor)^2 \, \mbox{ and } \, 
  (4e-2) \lfloor k/e \rfloor \geq 3k +2e-3,
$$
$$
  \Leftarrow k/e > \lfloor k/e \rfloor \, \mbox{ and } \, 
  (4e-2)(k/e - 1) \geq 3k +2e-3,
$$ where the first inequality is obvious and the second holds for $k = \Omega(1)$. 

For the second part we again use the first bound $(1)$ from Lemma \ref{l:harmonic_bounds}, to obtain:
 
\begin{eqnarray*}
k \cdot h\left(\left\lfloor \frac{k}{e} \right\rfloor + 1\right)
&=& 
H_{k-1} - H_{\left\lfloor \frac{k}{e} \right\rfloor + 1} - 1
\\
&<& 
\ln(k-1) + \frac{1}{2(k-1)} - \ln\left(\left\lfloor \frac{k}{e} \right\rfloor + 1\right) - \frac{1}{2(\left\lfloor \frac{k}{e} \right\rfloor + 2)} - 1
\\
&=&
\ln\left(\frac{e(k-1)}{e(\lfloor \frac{k}{e} \rfloor + 1 )}\right) + \frac{1}{2(k-1)} - \frac{1}{2(\left\lfloor \frac{k}{e} \right\rfloor + 2)} - 1
\\
&=&
\ln\left(\frac{k-1}{e(\lfloor \frac{k}{e} \rfloor + 1 )}\right) + \frac{1}{2(k-1)} - \frac{1}{2(\left\lfloor \frac{k}{e} \right\rfloor + 2)}
\\
&<&
0
\ ,
\end{eqnarray*} where the last inequality holds because it follows from
$$
  \ln\left(\frac{k-1}{e(\lfloor \frac{k}{e} \rfloor + 1 )}\right) < 0 \, \mbox{ and } \, 
  \frac{1}{2(k-1)} - \frac{1}{2(\left\lfloor \frac{k}{e} \right\rfloor + 2)}  < 0
$$
$$
  \Leftrightarrow k/e < \lfloor k/e \rfloor + 1 + 1/e \, \mbox{ and } \, 
  \lfloor k/e \rfloor < k - 3,
$$
$$
  \Leftarrow  k/e < \lfloor k/e \rfloor + 1 + 1/e \, \mbox{ and } \, 
  \lfloor k/e \rfloor \leq k/e \leq k - 3,
$$ where the first inequality is obvious and the second holds for $k \geq \frac{e-1}{3e}$.

The second part of the lemma that function $f$ achieves its maximum for $m \in \{\lfloor \frac{k}{e} \rfloor, \lfloor \frac{k}{e} \rfloor + 1\}$ follows directly from the first part of that lemma and from the definition of the discrete derivative $h(\cdot)$.
\end{proof}

\noindent Proposition \ref{Thm:optimum_expansion} below  shows a characterization of the optimal success probability $OPT_n$ of secretary algorithms.

\begin{proposition}\label{Thm:optimum_expansion}
\

\begin{enumerate}
\item\label{Thm:optimum_expansion_1} The optimal success probability of the best secretarial algorithm for the problem with $n$ items which uses a uniform random order from $\Pi_n$ is $OPT_n = 1/e + c_0/n + \Theta((1/n)^{3/2})$, where $c_0 = 1/2 - 1/(2e)$.
 
 \item\label{Thm:optimum_expansion_2} The success probability of any secretarial algorithm for the problem with $n$ items which uses any probabilistic distribution on $\Pi_n$ is at most $OPT_n = 1/e + c_0/n + \Theta((1/n)^{3/2})$.
 
 \item\label{Thm:optimum_expansion_3} There exists an infinite sequence of integers $n_1 < n_2 < n_3 < \ldots$, such that the success probability of any deterministic secretarial algorithm for the problem with $n \in \{n_1,n_2, n_3,\ldots\}$ items which uses any uniform probabilistic distribution on $\Pi_n$ with support $\ell < n$ is strictly smaller than $1/e$.
\end{enumerate}
\end{proposition}

\begin{proof}
\noindent
Part \ref{Thm:optimum_expansion_1}. Gilbert and Mosteller~\cite{GilbertM66} proved
that under maximum entropy, the probability of success is maximized by wait-and-pick algorithm with some checkpoint. Another important property, used in many papers (c.f., Gupta and Singla \cite{GuptaSingla}), is that function $f(n,m)$ describes the probability of success of the wait-and-pick algorithm with  checkpoint $m$.

Consider wait-and-pick algorithms with checkpoint $m\in [n-1]$.
\ignore{
We first argue that for $j$ such that $|j|>\sqrt{n}$, function $f$ does not reach its global maximum. Indeed, for any $m\in [n-1]$ consider the difference
\[
f(n,m+1)-f(n,m) = \frac{1}{n} \cdot \left(H_{n-1}-H_m-1 \right)
\ ,
\]
which is positive until $H_m+1$ gets bigger than $H_{k-1}$, and remains negative afterwards.
Thus, there is only one local maximum, which is also global maximum, of function $f$ wrt variable $m$.
Lemma~\ref{lemma:f_expansion_1} for $k=n$ implies that 
$f\left(n, \lceil\frac{n}{e}-\sqrt{n}\right)\rceil \le f\left(n, \lceil\frac{n}{e}\rceil\right)$ and $f\left(n, \lceil\frac{n}{e}\rceil\right) \ge f\left(n, \lfloor\frac{n}{e}+\sqrt{n}\rfloor\right)$, therefore this unique local maximum is reached for some $|j|\le \sqrt{n}$.

Because $\frac{n}{e} - 1 \leq \left\lfloor \frac{n}{e} \right\rfloor$ and $\left\lfloor \frac{n}{e} \right\rfloor + 1 \leq \frac{n}{e} + 1$, by Lemma \ref{l:deriv_bounds_1}, this unique local, and also global, maximum of function $f(n,\cdot)$ is achieved for some $j$ such that $|j| \leq 1$. 
}
By Lemma \ref{l:deriv_bounds_1}, function $f(n,\cdot)$ achieves is maximum for checkpoint
$m\in \left\{\left\lfloor \frac{n}{e} \right\rfloor,\left\lfloor \frac{n}{e} \right\rfloor + 1\right\}$, and
by Lemma \ref{lemma:f_expansion_1}, taken for $k=n$, it could be seen that for any admissible value of $j$ (i.e., such that $n/e+j$ is an integer and $|j|\le 1$, thus also for $j\in \left\{\left\lfloor \frac{n}{e} \right\rfloor -n/e,\left\lfloor \frac{n}{e} \right\rfloor + 1-n/e\right\}$ for which $f(n,m)$ achieves its maximum), and for $c_0 = 1/2 - 1/(2e)$:
$
f\left(n, \frac{n}{e}+j\right)  
= \frac{1}{e} + \frac{c_0}{n}  + \Theta\left( \left( \frac{1}{n} \right)^{3/2} \right) \, .
$
\ignore{
By Lemma \ref{lemma:f_expansion_1}, taken for $k=n$, it could be seen that for any admissible value of $j$ (i.e., such that $n/e+j$ is an integer and $|j|\le \sqrt{n}$), 
\[
f\left(n, \frac{n}{e}+j\right)  
\le \frac{1}{e} + \frac{c_0}{n}  + \Theta\left( \left( \frac{1}{n} \right)^{3/2} \right) \, ,
\] 
where $c_0 = 1/2 - 1/(2e)$. 
This bound is met, up to $\Theta(n^{-3/2})$, for constant values of $j$, thus the above upper bound on $f$ is actually the $OPT_n$.}%
%

\noindent
Part \ref{Thm:optimum_expansion_2}. Consider a probabilistic distribution on set $\Pi_n$, which for every permutation $\pi\in\Pi_n$ assigns probability $p_\pi$ of being selected.
Suppose that the permutation selected by the adversary is $\sigma \in \Pi_n$. Given a permutation $\pi\in\Pi_n$ selected by the algorithm, let $\chi(\pi,\sigma) = 1$ if the algorithm is successful on the adversarial permutation $\sigma$ and its selected permutation $\pi$, and $\chi(\pi,\sigma) = 0$ otherwise.

Given a specific adversarial choice $\sigma \in \Pi_n$, the total weight of permutations resulting in success of the secretarial algorithm is
$
\sum_{\pi\in\Pi_n} p_{\pi} \cdot \chi(\sigma,\pi) \, .
$

Suppose now that the adversary selects its permutation $\sigma$ uniformly at random from $\Pi_n$. The expected total weight of permutations resulting in success of the secretarial algorithm is
$
\sum_{\sigma \in \Pi_n} q_{\sigma} \cdot \left(\sum_{\pi\in\Pi_n} p_{\pi} \cdot \chi(\sigma,\pi)\right)
$, where $q_{\sigma} = 1/n!$ for each $\sigma \in \Pi_n$.
The above sum can be rewritten as follows:
$ 
\sum_{\sigma \in \Pi_n} q_{\sigma} \cdot \left(\sum_{\pi\in\Pi_n} p_{\pi} \cdot \chi(\sigma,\pi)\right) =
 \sum_{\pi\in\Pi_n} p_{\pi} \cdot \left(\sum_{\sigma \in \Pi_n} q_{\sigma} \cdot \chi(\sigma,\pi)\right) \, ,
$ 
and now we can treat permutation $\pi$ as fixed and adversarial, and permutation $\sigma$ as chosen by the algorithm uniformly at random from $\Pi_n$, we have by Part \ref{Thm:optimum_expansion_1} that
$
\sum_{\sigma \in \Pi_n} q_{\sigma} \cdot \chi(\sigma,\pi) =
OPT_n
$. This implies that the expected total weight of permutations resulting in success of the secretarial algorithm is at most
$ 
\sum_{\pi\in\Pi_n} p_{\pi} \cdot \chi(\sigma,\pi) \leq \sum_{\pi\in\Pi_n} p_{\pi} \cdot OPT_n = OPT_n \, .
$ 
Therefore, there exists a permutation $\sigma\in\Pi_n$ realizing this adversarial goal. Thus it is impossible that there is a secretarial algorithm that for any adversarial permutation $\sigma\in\Pi_n$ has success probability $> OPT_n$.

\noindent
Part \ref{Thm:optimum_expansion_3}. Let $\ell_i = 10^i$ and $n_i = 10 \ell_i$ for $i \in \nats_{\geq 1}$. Let us take the infinite decimal expansion of $1/e = 0.367879441171442 ...$ and define as $d_i > 1$ the integer that is build from the first $i$ digits in this decimal expansion after the decimal point, that is, $d_1 = 3$, $d_2 = 36$, $d_3 = 367$, and so on. The sequence $d_i/\ell_i$ has the following properties: $\lim_{i \rightarrow +\infty} d_i/\ell_i = 1/e$, for each $i = 1,2,...$ we have that $d_i/\ell_i < 1/e < (d_i + 1)/\ell_i$ and, moreover, 
$j/\ell_i \not \in [1/e, 1/e + 1/n_i]$ for all $j \in \{0,1,2,\ldots, \ell_i\}$.

Let us now take any $n = n_i$ for some (large enough) $i \in \nats_{\geq 1}$ and consider the secretary problem with $n = n_i$ items. Consider also any deterministic secretarial algorithm for this problem that uses any uniform probability distribution on the set $\Pi_{n_i}$ with support $\ell_i$. By Part \ref{Thm:optimum_expansion_2} the success probability of this algorithm using this probability distribution is at most $OPT_{n_i} = 1/e + c_0/n_i + \Theta((1/n_i)^{3/2})$. Because the algorithm is deterministic, all possible probabilities in this probability distribution belong to the set $\{j/\ell_i : j \in \{0,1,2,\ldots, \ell_i\}\}$. We observe now that $j/\ell_i \not \in [1/e, 1/e + c_0/n_i + \Theta((1/n_i)^{3/2})]$ for $j \in \{0,1,2,\ldots, \ell_i\}$. This 
fact holds by the construction and by the fact that constant $c_0 \in (0,1)$, and we may also need to assume that $i \in \nats_{\geq 1}$ is taken to be large enough to deal with the term $\Theta((1/n_i)^{3/2})$. Thus the success probability of this algorithm is strictly below $1/e$.
\end{proof}

\section{Discussion and conclusions}\label{Sec:Concl}

Our Theorem~\ref{thm:lower-general} is a non-trivial extension using the concept of semitone sequences, proposed in \cite{KesselheimKN15}. One of the technical challenges we had to overcome is that $k$ picks of values by the algorithm, instead of one, create additional
dependencies in probabilistic part of the proof.

By Theorem~\ref{thm:lower-general}, entropy $\Omega(\log\log n)$ is necessary for any algorithm to achieve even a constant competitive ratio $1-\epsilon$, for $k=O(\log^a n)$, $a<1$, for $k$-secretary problem, implying that our upper bound in Theorem~\ref{thm:k_secretary_main_result} is tight.
Theorem~\ref{thm:lower} 
implies that entropy $\Omega(\log\log n)$ is necessary for any wait-and-pick algorithm to get a close-to-optimal competitive ratio $1-\Omega(\frac{1}{k^a})$, for {\em any} $k<\frac{n}{2}$, where $a\le \frac{1}{2}$.  
Moreover, the entropy $\Omega(\log k)$ is necessary, which could be $\Omega(\log n)$~for~$k = poly(n)$.

Theorem \ref{thm:k_secretary_main_result} implies an (almost)\footnote{The competitive ratio is optimal up to a factor of $\sqrt{\log k}$, see \cite{kleinberg2005multiple,GuptaSingla,AgrawalWY14} for a matching lower bound.} optimal competitive ratio $1-O(\sqrt{\log k/k})$, with minimal entropy $O(\log \log n)$, when $k < \log n/\log \log n$, 
for the $k$-secretary problem in the non-uniform arrival model. 
This improves over~\cite{KesselheimKN15}, who achieved competitive ratio $(1-O(1/k^{1/3})-o(1))$ using entropy $O(\log \log n)$ and constructing the distribution for $k=O((\log\log\log n)^{\epsilon})$ only. Our construction improves the range of working values $k$ exponentially.
Optimality of the entropy follows by our new lower bounds in Theorem \ref{thm:lower-general} and \ref{thm:lower}. Such lower bounds were not known before for the $k$-secretary problem, for $k>1$.
For the $1$-secretary problem, Theorem~\ref{thm:1_secretary_results} improves, over doubly-exponentially, on the additive error to $\frac{1}{e}$ of 
$\omega(\frac{1}{(\log\log\log(n))^{c}})$ proposed in ~\cite{KesselheimKN15,KesselheimKN15-arxiv}, which holds for any positive constant $c < 1$, by achieving a polynomial-time construction with the additive error $\Theta(\frac{\log\log{n}}{\log^{1/2}{n}})$.

 We have already applied our new derandomization techniques to two different secretarial problems: $1$-secretary and multiple-choice secretary. Further promising candidates are the online bipartite matching and matroid secretary problems, 
 see \cite{GuptaSingla}.

\addcontentsline{toc}{section}{References}

\bibliographystyle{plainurl}

\bibliography{bibliography}

\begin{thebibliography}{10}

\bibitem{abolhassani2017beating}
Melika Abolhassani, Soheil Ehsani, Hossein Esfandiari, MohammadTaghi HajiAghayi, Robert Kleinberg, and Brendan Lucier.
\newblock Beating 1-1/e for ordered prophets.
\newblock In {\em Proceedings of the 49th Annual ACM SIGACT Symposium on Theory of Computing}, pages 61--71. ACM, 2017.

\bibitem{DBLP:conf/sigecom/0001SZ20}
Shipra Agrawal, Jay Sethuraman, and Xingyu Zhang.
\newblock On optimal ordering in the optimal stopping problem.
\newblock In {\em {EC} '20: The 21st {ACM} Conference on Economics and Computation, Virtual Event, Hungary, July 13-17, 2020}, pages 187--188. {ACM}, 2020.

\bibitem{AgrawalWY14}
Shipra Agrawal, Zizhuo Wang, and Yinyu Ye.
\newblock A dynamic near-optimal algorithm for online linear programming.
\newblock {\em Operations Research}, 62(4):876--890, 2014.

\bibitem{alaei2014bayesian}
Saeed Alaei.
\newblock Bayesian combinatorial auctions: Expanding single buyer mechanisms to many buyers.
\newblock {\em SIAM Journal on Computing}, 43(2):930--972, 2014.

\bibitem{AHL13}
Saeed Alaei, MohammadTaghi Hajiaghayi, and Vahid Liaghat.
\newblock The online stochastic generalized assignment problem.
\newblock In {\em Proceedings of the 16th International Workshop on Approximation Algorithms for Combinatorial Optimization Problems (APPROX)}, pages 11--25, 2013.

\bibitem{antoniadis2020secretary}
Antonios Antoniadis, Themis Gouleakis, Pieter Kleer, and Pavel Kolev.
\newblock Secretary and online matching problems with machine learned advice.
\newblock {\em Advances in Neural Information Processing Systems}, 33:7933--7944, 2020.

\bibitem{ArsenisDK21}
Makis Arsenis, Odysseas Drosis, and Robert Kleinberg.
\newblock Constrained-order prophet inequalities.
\newblock In {\em Proc.~2021 {ACM-SIAM} Symposium on Discrete Algorithms, {SODA} 2021}, pages 2034--2046, 2021.

\bibitem{assadi2019secretary}
Sepehr Assadi, Eric Balkanski, and Renato Leme.
\newblock Secretary ranking with minimal inversions.
\newblock {\em Advances in Neural Information Processing Systems}, 32, 2019.

\bibitem{ACK18}
Yossi Azar, Ashish Chiplunkar, and Haim Kaplan.
\newblock Prophet secretary: Surpassing the 1-1/e barrier.
\newblock In {\em Proceedings of the 2018 {ACM} Conference on Economics and Computation, Ithaca, NY, USA, June 18-22, 2018}, pages 303--318, 2018.

\bibitem{BabaioffIKK08}
Moshe Babaioff, Nicole Immorlica, David Kempe, and Robert Kleinberg.
\newblock Online auctions and generalized secretary problems.
\newblock {\em SIGecom Exch.}, 7(2), 2008.

\bibitem{babaioff2007matroids}
Moshe Babaioff, Nicole Immorlica, and Robert Kleinberg.
\newblock Matroids, secretary problems, and online mechanisms.
\newblock In {\em Proceedings of the eighteenth annual ACM-SIAM symposium on Discrete algorithms}, pages 434--443. Society for Industrial and Applied Mathematics, 2007.

\bibitem{BHZ13}
Mohammadhossein Bateni, Mohammadtaghi Hajiaghayi, and Morteza Zadimoghaddam.
\newblock Submodular secretary problem and extensions.
\newblock {\em ACM Trans. Algorithms}, 9(4):Art. 32, 23, 2013.

\bibitem{DBLP:journals/ior/BeyhaghiGLPS21}
Hedyeh Beyhaghi, Negin Golrezaei, Renato~Paes Leme, Martin P{\'{a}}l, and Balasubramanian Sivan.
\newblock Improved revenue bounds for posted-price and second-price mechanisms.
\newblock {\em Oper. Res.}, 69(6):1805--1822, 2021.

\bibitem{bierhorst2018experimentally}
Peter Bierhorst, Emanuel Knill, Scott Glancy, Yanbao Zhang, Alan Mink, Stephen Jordan, Andrea Rommal, Yi-Kai Liu, Bradley Christensen, Sae~Woo Nam, et~al.
\newblock Experimentally generated randomness certified by the impossibility of superluminal signals.
\newblock {\em Nature}, 556(7700):223--226, 2018.

\bibitem{buchbinder2023lossless}
Niv Buchbinder, Joseph Naor, and David Wajc.
\newblock Lossless online rounding for online bipartite matching (despite its impossibility).
\newblock In {\em Proceedings of the 2023 Annual ACM-SIAM Symposium on Discrete Algorithms (SODA)}, pages 2030--2068. SIAM, 2023.

\bibitem{10.1145/1806689.1806733}
Shuchi Chawla, Jason~D. Hartline, David~L. Malec, and Balasubramanian Sivan.
\newblock Multi-parameter mechanism design and sequential posted pricing.
\newblock In {\em Proceedings of the Forty-Second ACM Symposium on Theory of Computing}, page 311–320, 2010.

\bibitem{ChowMRS64}
Y.~S. Chow, S.~Moriguti, H.~Robbins, and S.~M. Samuels.
\newblock Optimal selection based on relative rank (the "secretary" problem).
\newblock {\em Israel J. Math.}, 2:81--90, 1964.

\bibitem{correa2017posted}
Jos{\'e} Correa, Patricio Foncea, Ruben Hoeksma, Tim Oosterwijk, and Tjark Vredeveld.
\newblock Posted price mechanisms for a random stream of customers.
\newblock In {\em Proceedings of the 2017 ACM Conference on Economics and Computation}, pages 169--186. ACM, 2017.

\bibitem{CorreaSZ19}
Jos{\'{e}}~R. Correa, Raimundo Saona, and Bruno Ziliotto.
\newblock Prophet secretary through blind strategies.
\newblock In {\em Proceedings of the Thirtieth Annual {ACM-SIAM} Symposium on Discrete Algorithms, {SODA} 2019, San Diego, California, USA, January 6-9, 2019}, pages 1946--1961. {SIAM}, 2019.

\bibitem{DBLP:journals/mp/CorreaSZ21}
Jos{\'{e}}~R. Correa, Raimundo Saona, and Bruno Ziliotto.
\newblock Prophet secretary through blind strategies.
\newblock {\em Math. Program.}, 190(1):483--521, 2021.

\bibitem{DEHLS17}
Sina Dehghani, Soheil Ehsani, MohammadTaghi Hajiaghayi, Vahid Liaghat, and Saeed Seddighin.
\newblock Stochastic k-server: How should uber work?
\newblock In {\em 44th International Colloquium on Automata, Languages, and Programming, {ICALP} 2017, July 10-14, 2017, Warsaw, Poland}, pages 126:1--126:14, 2017.

\bibitem{dutting2015polymatroid}
Paul D{\"u}tting and Robert Kleinberg.
\newblock Polymatroid prophet inequalities.
\newblock In {\em Algorithms-ESA 2015}, pages 437--449. Springer, 2015.

\bibitem{dynkin1963optimum}
Eugene~B Dynkin.
\newblock The optimum choice of the instant for stopping a markov process.
\newblock In {\em Soviet Math. Dokl}, volume~4, 1963.

\bibitem{EHLM17}
Hossein Esfandiari, MohammadTaghi Hajiaghayi, Vahid Liaghat, and Morteza Monemizadeh.
\newblock Prophet secretary.
\newblock {\em {SIAM} J. Discret. Math.}, 31(3):1685--1701, 2017.

\bibitem{DBLP:conf/aistats/EsfandiariHLM20}
Hossein Esfandiari, MohammadTaghi Hajiaghayi, Brendan Lucier, and Michael Mitzenmacher.
\newblock Prophets, secretaries, and maximizing the probability of choosing the best.
\newblock In {\em The 23rd International Conference on Artificial Intelligence and Statistics, {AISTATS} 2020}, 2020.

\bibitem{feldman2015combinatorial}
Michal Feldman, Nick Gravin, and Brendan Lucier.
\newblock Combinatorial auctions via posted prices.
\newblock In {\em Proceedings of the Twenty-Sixth Annual ACM-SIAM Symposium on Discrete Algorithms}, pages 123--135. SIAM, 2015.

\bibitem{DBLP:conf/soda/FeldmanSZ15}
Moran Feldman, Ola Svensson, and Rico Zenklusen.
\newblock A simple \emph{O}(log log(rank))-competitive algorithm for the matroid secretary problem.
\newblock In Piotr Indyk, editor, {\em Proceedings of the Twenty-Sixth Annual {ACM-SIAM} Symposium on Discrete Algorithms, {SODA} 2015, San Diego, CA, USA, January 4-6, 2015}.

\bibitem{garg2008stochastic}
Naveen Garg, Anupam Gupta, Stefano Leonardi, and Piotr Sankowski.
\newblock Stochastic analyses for online combinatorial optimization problems.
\newblock In {\em SODA}, pages 942--951. Society for Industrial and Applied Mathematics, 2008.

\bibitem{GilbertM66}
J.~Gilbert and F.~Mosteller.
\newblock Recognizing the maximum of a sequence.
\newblock {\em J. Amer. Statist. Assoc.}, 61:35--73, 1966.

\bibitem{gobel2014online}
Oliver G{\"o}bel, Martin Hoefer, Thomas Kesselheim, Thomas Schleiden, and Berthold V{\"o}cking.
\newblock Online independent set beyond the worst-case: Secretaries, prophets, and periods.
\newblock In {\em International Colloquium on Automata, Languages, and Programming}, pages 508--519. Springer, 2014.

\bibitem{M_Goemans_2015}
Michel~X. Goemans.
\newblock Chernoff bounds, and some applications.
\newblock {\em Lecture notes: https://math.mit.edu/$\sim$goemans/18310S15/chernoff-notes.pdf}, 2015.

\bibitem{Gordon1994}
Louis Gordon.
\newblock A stochastic approach to the gamma function.
\newblock {\em The American Mathematical Monthly}, 101(9):858--865, 1994.

\bibitem{DBLP:conf/wine/GuptaRST10}
Anupam Gupta, Aaron Roth, Grant Schoenebeck, and Kunal Talwar.
\newblock Constrained non-monotone submodular maximization: Offline and secretary algorithms.
\newblock In {\em Proceedings of the 6th International Workshop on Internet and Network Economics (WINE'10)}, 2010.

\bibitem{GuptaSingla}
Anupam Gupta and Sahil Singla.
\newblock Random-order models.
\newblock In Tim Roughgarden, editor, {\em Beyond the Worst-Case Analysis of Algorithms}, pages 234--258. Cambridge University Press, 2020.

\bibitem{HagerupR90}
Torben Hagerup and Christine R{\"{u}}b.
\newblock A guided tour of chernoff bounds.
\newblock {\em Inf. Process. Lett.}, 33(6):305--308, 1990.

\bibitem{hajiaghayi2007automated}
Mohammad~Taghi Hajiaghayi, Robert Kleinberg, and Tuomas Sandholm.
\newblock Automated online mechanism design and prophet inequalities.
\newblock In {\em AAAI}, volume~7, pages 58--65, 2007.

\bibitem{HajiaghayiKP04}
Mohammad~Taghi Hajiaghayi, Robert~D. Kleinberg, and David~C. Parkes.
\newblock Adaptive limited-supply online auctions.
\newblock In {\em Proceedings 5th {ACM} Conference on Electronic Commerce (EC-2004), New York, NY, USA, May 17-20, 2004}, pages 71--80. {ACM}, 2004.

\bibitem{Hill1983}
T.~P. Hill.
\newblock Prophet inequalities and order selection in optimal stopping problems.
\newblock {\em Proc.~American Mathematical Society}, 88(1):131--137, 1983.

\bibitem{hill1982comparisons}
Theodore~P Hill and Robert~P Kertz.
\newblock Comparisons of stop rule and supremum expectations of iid random variables.
\newblock {\em The Annals of Probability}, pages 336--345, 1982.

\bibitem{jaillet2013online}
Patrick Jaillet, Jos{\'{e}}~A. Soto, and Rico Zenklusen.
\newblock Advances on matroid secretary problems: Free order model and laminar case.
\newblock In {\em Integer Programming and Combinatorial Optimization - 16th International Conference, {IPCO} 2013, Valpara{\'{\i}}so, Chile, March 18-20, 2013. Proceedings}, 2013.

\bibitem{jiang2021online}
Zhihao Jiang, Pinyan Lu, Zhihao~Gavin Tang, and Yuhao Zhang.
\newblock Online selection problems against constrained adversary.
\newblock In {\em International Conference on Machine Learning}, pages 5002--5012. PMLR, 2021.

\bibitem{KesselheimKN15}
Thomas Kesselheim, Robert~D. Kleinberg, and Rad Niazadeh.
\newblock Secretary problems with non-uniform arrival order.
\newblock In {\em Proceedings of the Forty-Seventh Annual {ACM} on Symposium on Theory of Computing, {STOC} 2015, Portland, OR, USA, June 14-17, 2015}, pages 879--888. {ACM}, 2015.

\bibitem{KesselheimKN15-arxiv}
Thomas Kesselheim, Robert~D. Kleinberg, and Rad Niazadeh.
\newblock Secretary problems with non-uniform arrival order.
\newblock {\em CoRR}, abs/1502.02155, 2015.

\bibitem{kesselheim2013optimal}
Thomas Kesselheim, Klaus Radke, Andreas T{\"o}nnis, and Berthold V{\"o}cking.
\newblock An optimal online algorithm for weighted bipartite matching and extensions to combinatorial auctions.
\newblock In {\em European Symposium on Algorithms}, pages 589--600. Springer, 2013.

\bibitem{kleinberg2005multiple}
Robert Kleinberg.
\newblock A multiple-choice secretary algorithm with applications to online auctions.
\newblock In {\em Proceedings of the sixteenth annual ACM-SIAM symposium on Discrete algorithms}, pages 630--631. Society for Industrial and Applied Mathematics, 2005.

\bibitem{KW-STOC12}
Robert Kleinberg and S.~Matthew Weinberg.
\newblock Matroid prophet inequalities.
\newblock In {\em Proceedings of the 44th Symposium on Theory of Computing Conference, {STOC} 2012, New York, NY, USA, May 19 - 22, 2012}, pages 123--136, 2012.

\bibitem{knuth1981art}
Donald Knuth.
\newblock The art of computer programming, 2 (seminumerical algorithms).
\newblock {\em (No Title)}, 1981.

\bibitem{krengel1977semiamarts}
Ulrich Krengel and Louis Sucheston.
\newblock Semiamarts and finite values.
\newblock {\em Bull. Am. Math. Soc}, 1977.

\bibitem{krengel1978semiamarts}
Ulrich Krengel and Louis Sucheston.
\newblock On semiamarts, amarts, and processes with finite value.
\newblock {\em Advances in Prob}, 4:197--266, 1978.

\bibitem{DBLP:conf/focs/Lachish14}
Oded Lachish.
\newblock O(log log rank) competitive ratio for the matroid secretary problem.
\newblock In {\em 55th {IEEE} Annual Symposium on Foundations of Computer Science, {FOCS} 2014, Philadelphia, PA, USA, October 18-21, 2014}.

\bibitem{Lindley61}
D.~V. Lindley.
\newblock Dynamic programming and decision theory.
\newblock {\em Appl. Statist.}, 10:39--51, 1961.

\bibitem{LiuLPSS21}
Allen Liu, Renato~Paes Leme, Martin P{\'{a}}l, Jon Schneider, and Balasubramanian Sivan.
\newblock Variable decomposition for prophet inequalities and optimal ordering.
\newblock In {\em {EC} '21: The 22nd {ACM} Conference on Economics and Computation, Budapest, Hungary, July 18-23, 2021}, page 692. {ACM}, 2021.

\bibitem{DBLP:conf/sigecom/LiuLPSS21}
Allen Liu, Renato~Paes Leme, Martin P{\'{a}}l, Jon Schneider, and Balasubramanian Sivan.
\newblock Variable decomposition for prophet inequalities and optimal ordering.
\newblock In {\em {EC} '21: The 22nd {ACM} Conference on Economics and Computation, Budapest, Hungary, July 18-23, 2021}, page 692. {ACM}, 2021.

\bibitem{mahmud2020survey}
Mohammad~Sultan Mahmud, Joshua~Zhexue Huang, Salman Salloum, Tamer~Z Emara, and Kuanishbay Sadatdiynov.
\newblock A survey of data partitioning and sampling methods to support big data analysis.
\newblock {\em Big Data Mining and Analytics}, 3(2):85--101, 2020.

\bibitem{Meyerson-FOCS01}
Adam Meyerson.
\newblock Online facility location.
\newblock In {\em Foundations of Computer Science, 2001. Proceedings. 42nd IEEE Symposium on}, pages 426--431. IEEE, 2001.

\bibitem{PT22}
Bo~Peng and Zhihao~Gavin Tang.
\newblock Order selection prophet inequality: From threshold optimization to arrival time design.
\newblock In {\em Proceedings of the 63rd Annual {IEEE} Symposium on Foundations of Computer Science (FOCS)}, 2022.

\bibitem{rubinstein2016beyond}
Aviad Rubinstein.
\newblock Beyond matroids: Secretary problem and prophet inequality with general constraints.
\newblock {\em arXiv preprint arXiv:1604.00357}, 2016.

\bibitem{RS-SODA17}
Aviad Rubinstein and Sahil Singla.
\newblock Combinatorial prophet inequalities.
\newblock In {\em Proceedings of the Twenty-Eighth Annual ACM-SIAM Symposium on Discrete Algorithms}, pages 1671--1687, 2017.

\bibitem{Samuels81}
S.M. Samuels.
\newblock Minimax stopping rules when the underlying distribution is uniform.
\newblock {\em J. Amer. Statist. Assoc.}, 76:188--197, 1981.

\bibitem{Stewart_Book}
I.~Stewart.
\newblock {\em Galois Theory}.
\newblock CRC Press, 4th edition, 2015.

\bibitem{turan2015random}
Meltem~Sonmez Turan, John~M. Kelsey, and Kerry~A. McKay.
\newblock How random is your {RNG}?
\newblock In {\em Shmoocon Proceedings, Washington, DC}. National Institute of Standards and Technology (NIST), 2015.
\newblock URL: \url{https://tsapps.nist.gov/publication/get_pdf.cfm?pub_id=917957 (Accessed September 25, 2023)}.

\bibitem{vitter1987efficient}
Jeffrey~Scott Vitter.
\newblock An efficient algorithm for sequential random sampling.
\newblock {\em ACM transactions on mathematical software (TOMS)}, 13(1):58--67, 1987.

\bibitem{D_Wajc_2017}
David Wajc.
\newblock Negative association - definition, properties, and applications.
\newblock {\em Lecture notes: https://www.cs.cmu.edu/$\sim$dwajc/notes/Negative 20Association.pdf}, 2017.

\bibitem{yan2011mechanism}
Qiqi Yan.
\newblock Mechanism design via correlation gap.
\newblock In {\em Proceedings of the twenty-second annual ACM-SIAM symposium on Discrete Algorithms}, pages 710--719, 2011.

\bibitem{Yao77}
Andrew~Chi{-}Chih Yao.
\newblock Probabilistic computations: Toward a unified measure of complexity (extended abstract).
\newblock In {\em 18th Annual Symposium on Foundations of Computer Science, Providence, Rhode Island, USA, 31 October - 1 November 1977}, pages 222--227, 1977.

\bibitem{Young95}
Neal~E. Young.
\newblock Randomized rounding without solving the linear program.
\newblock In {\em Proceedings of the Sixth Annual {ACM-SIAM} Symposium on Discrete Algorithms, 22-24 January 1995. San Francisco, California, {USA}}, pages 170--178, 1995.

\end{thebibliography}

 


\end{document}